\documentclass[11pt,reqno]{amsart}

\usepackage{amsfonts}
\usepackage[sans]{dsfont}

\usepackage{amsmath,amssymb,amsthm}
\usepackage[margin=1.0in]{geometry}
\usepackage{enumerate}

\theoremstyle{plain}
\newtheorem{theorem}{Theorem}
\newtheorem{proposition}[theorem]{Proposition}
\newtheorem{lemma}[theorem]{Lemma}
\newtheorem{corollary}[theorem]{Corollary}

\theoremstyle{definition}
\newtheorem{definition}[theorem]{Definition}

\numberwithin{theorem}{section}
\numberwithin{equation}{section}

\newcommand{\R}{\mathbb{R}}
\newcommand{\C}{\mathbb{C}}

\newcommand{\D}{\mathcal{D}}

\def\N{{\mathbb N}}
\newcommand{\CC}{\mathbb {C}}

\newcommand{\cD}{\mathcal{D}}

\newcommand{\cF}{\mathcal{F}}

\newcommand{\cH}{\mathcal{H}}

\def\CO{{\mathcal O}}

\newcommand{\eps}{{\varepsilon}}
\newcommand{\ve}{{\varepsilon}}

\newcommand{\e}{{\epsilon}}
\newcommand{\al}{{\alpha}}
\newcommand{\g}{{\gamma}}
\newcommand{\G}{{\Gamma}}
\newcommand{\del}{{\delta}}
\newcommand{\lam}{{\lambda}}           

\newcommand{\Om}{\Omega}                
\newcommand{\om}{\omega}
\newcommand{\s}{{\sigma}}

\newcommand{\fh}{\mathfrak{h}}

\def\CO{{\mathcal O}}
\renewcommand{\d}{\mathrm{d}}

\def\<{\langle}
\def\>{\rangle}

\newcommand{\im}{\operatorname{Im}}
\newcommand{\re}{\operatorname{Re}}
\newcommand{\supp}{\operatorname{supp}}
\newcommand{\ddim}{\operatorname{d\:\!Im}}
\newcommand{\ddre}{\operatorname{d\:\!Re}}

\newcommand{\ran}{\rangle}
\newcommand{\lan}{\langle}
\newcommand{\ra}{\rightarrow}

\newcommand{\Ran}{\operatorname{Ran}}
\newcommand{\dist}{\operatorname{dist}}
\newcommand{\p}{{\partial}}

\newcommand{\diag}{\operatorname{diag}}

\newcommand{\bfone}{{\bf 1}}
\newcommand{\one}{\mathbf{1}}

\newcommand{\slim}{\mathop{\text{\rm{s-lim}}}}

\newcommand{\DETAILS}[1]{}
\newcommand{\at}{{p}}
\newcommand{\hp}{H_{p}}
\newcommand{\chp}{\mathcal{H}_{p}}

\newcounter{foo}
\Roman{foo}
\def\bigno{\bigskip \noindent}

\begin{document}

\title{On quantum Huygens principle and Rayleigh scattering}

\author[J. Faupin]{J{\'e}r{\'e}my Faupin}
\address[J. Faupin]{Institut de Math{\'e}matiques de Bordeaux \\
UMR-CNRS 5251, Universit{\'e} de Bordeaux 1, 
33405 Talence Cedex, France}
\email{jeremy.faupin@math.u-bordeaux1.fr}
\author[I. M. Sigal]{Israel Michael Sigal} 
\address[I. M. Sigal]{Department of Mathematics \\
University of Toronto, 
Toronto, ON M5S 2E4, Canada}
\email{im.sigal@utoronto.ca}


\begin{abstract}
We prove several minimal photon velocity estimates below the ionization threshold for a particle system coupled to the quantized electromagnetic or phonon field. Using some of these results, we prove the asymptotic completeness (for the Rayleigh scattering) on the states for which the expectation of the photon number is uniformly bounded.
\end{abstract}

\maketitle

\section{ Introduction}\label{sec:intro}

In this paper we study the long-time dynamics of a non-relativistic particle system coupled to the quantized electromagnetic or phonon field. For energies below the ionization threshold, we prove several lower bounds on the growth of the distance of the escaping photons to the particle system. (Here and in what follows we use the term photon for both photon and phonon.) Using some of these results, we prove the asymptotic completeness (for the Rayleigh scattering) on the states for which the expectation of the photon number is bounded uniformly in time.

\medskip

\noindent \textbf{Model.} The state space for our model is given by ${\cH}:=\chp\otimes \cF$ and the dynamics is generated by the Hamiltonian
\begin{equation}\label{H}
H=\hp + H_f + I(g),
\end{equation}
acting on it. Here $\cH_{p}$ is the particle state  space,  $\cF$ is the bosonic Fock space, $\cF\equiv\G(\fh):=\oplus_0^\infty \otimes_s^n \fh$, based on the one-photon space $\fh :=L^2(\R^3)$, $\hp$ is a self-adjoint particle system Hamiltonian, acting on $\cH_{p}$, and $H_f:=\d \Gamma( \om )$ is the photon Hamiltonian, acting on $\cF$, where $\om= \omega(k)$ is the photon dispersion law ($k$ is the photon wave vector) and $\d\G(b)$ denotes the lifting of a one-photon operator $b$ to the photon Fock space,
\begin{equation}\label{dG}
\d \Gamma( b )_{\vert \otimes_s^n \fh} = \sum_{j=1}^{n} \underbrace{1 \otimes \cdots \otimes 1 }_{j-1} \otimes b \otimes \underbrace{1 \otimes \cdots \otimes 1 }_{n - j}\ .
\end{equation}
Here $\otimes_s^n$ stands for the symmetrized tensor product of $n$ factors (for $n=0$, $\fh$ is replaced by $\C$ and $\d \Gamma ( b ) _{ \vert \C } = 0$). The operator $I(g)$ acts on ${\cH}$ and represents an interaction energy, labeled by  a coupling family $g(k)$ of operators acting on the particle space $\chp$.

For \textit{photons}  $ \omega(k)=|k|$, for \textit{acoustic phonons},  $ \omega(k)\asymp |k|$ for small  $|k|$ and  $ c\le \omega(k)\le c^{-1}$, for some $ c>0$, away from $0$, while for \textit{optical phonons},  $c\le \omega(k)\le c^{-1}$, for some $ c>0$, for all  $k$. To fix ideas we consider below only the most difficult case of $ \omega(k)=|k|$. (For photons, to accommodate their polarizations, the one-boson space $L^2(\R^3)$ should be replaced by $L^2(\R^3 ; \C^2)$, but the resulting modifications are trivial, see e.g. \cite{GuSig, Sig2}.) In the simplest case of linear coupling (the dipole approximation in QED or the phonon models), $I(g)$ is given by
\begin{equation}\label{I}
I(g):= \int (g^* (k)\otimes a(k)+g (k)\otimes a^*(k))dk,
\end{equation}
with $a^*(k)$ and  $a(k)$, the creation and annihilation operators, acting on $\cF$ (see Supplement II for definitions).

A primary model for the particle system to have in mind is an electron in a vacuum or in a solid in an external potential $V$. In this case, $\hp:=\e(p) + V(x),\ p=-i\nabla_x,$ with $\e(p)$ being the standard non-relativistic kinetic energy, $\e(p) = |p|^2\equiv-\Delta_x$ (the Nelson model), or the electron dispersion law in a crystal lattice (a standard model in solid state physics), acting on $\chp:=L^2(\R^3)$, and the coupling family is given by
 \begin{equation}\label{gx}
g (k) = |k|^\mu \xi (k) e^{ikx},
\end{equation}
where $\xi (k)$ is the ultraviolet cut-off.  For phonons,   $\mu=1/2$. To have a self-adjoint operator $\hp$ we assume that $V$ is a Kato potential. A key fact here is that there is a spectral point $\Sigma\in \s(H)$, called the ionization threshold, s.t.  below $\Sigma$, the particle system is well localized:
 \begin{equation}\label{exp-bnd}
\|e^{\del |x|}f(H)\|\lesssim 1 ,
\end{equation}
for any $0\le \del< \dist(\supp f,\Sigma)$ and any   $f \in \mathrm{C}_0^\infty( (-\infty,\Sigma))$, i.e. states  decay exponentially in the particle coordinates $x$ (\cite{Gr, BFS1, BFS2}). This can be easily upgraded to an $N-$body system (e.g. an atom or a molecule, see e.g. \cite{GuSig, Sig2}). Another example -- the \textit{spin-boson model} -- will be defined below.

Finally,  the above can be extended to the standard model of non-relativistic quantum electrodynamics in which particles are minimally coupled to the quantized electromagnetic field, which leads to $I (g)$ being quadratic in the creation and annihilation operators $a^\#(k)$.

\medskip

\noindent \textbf{Problem.} In all above cases,  the Hamiltonian $H$ is self-adjoint and generates the dynamics through the Schr\"odinger equation,
\begin{equation}\label{SE}
 i\p_t\psi_t =H\psi_t.
 \end{equation}
As initial conditions, $\psi_0$, we consider states below the ionization threshold,  $\Sigma$, defined in \eqref{Sigma}, i.e. $\psi_0$ in the range of the spectral projection $E_\Delta(H),\ \Delta:=(-\infty,\Sigma)$. In other words, we are interested in processes, like emission and absorption of radiation, or scattering of photons on an electron bound by an external potential (created e.g. by an infinitely heavy nucleus or impurity of a crystal lattice),  in which the particle system (say, an atom or a molecule) is not being ionized.

Denote by $\Phi_j$ and $E_j$ the eigenfunctions and the corresponding eigenvalues of the hamiltonian $H$, below $\Sigma$, i.e. $E_j < \Sigma$. The following are the key characteristics of evolution of a physical system, in progressive order  the refined information they provide and in our context:
\begin{itemize}

\item
\textit{Local decay} stating that some photons are bound to the particle system while others (if any) escape to infinity, i.e. the probability that they occupy any bounded region of the physical space tends to zero, as $t \rightarrow \infty$.

\item \textit{Minimal photon velocity bound}  with speed $\mathrm{c}$ stating that,  as $t\ra \infty$, with probability $\ra 1$, the photons are  either bound to the particle system or depart from it with the distance $\ge c' t$, for any $c'< \mathrm{c}$.

\noindent Similarly, if the probability that at least one photon is at the distance $\ge c'' t$, $c''> \mathrm{c},$ from the particle system vanishes, as $t\ra \infty$, we say that the evolution satisfies the \textit{maximal photon velocity bound}  with speed $\mathrm{c}$.

\item  \textit{Asymptotic completeness} on the interval $(-\infty, \Sigma)$ stating that, for any $\psi_0\in  \Ran\chi_{(-\infty, \Sigma)}(H),$  and any $\e>0$, there are photon wave functions $f_{j\e}\in \cF$, with a finite number of photons, s.t.  the solution, $\psi_t= e^{ - i t H }\psi_0$, of the Schr\"odinger equation, \eqref{SE}, satisfies
\begin{align}\label{AC}
\limsup_{t\ra \infty}\|e^{ - i t H }\psi_0 -\sum_j e^{ - i E_{j} t} \Phi_{j}
\otimes_s e^{ - i H_f t}f_{j\e}\|\le \e.
\end{align}
(It will be shown in the text that $\Phi_{j}\otimes_s f_{j\e}$ is  well-defined, at least for the ground state ($j=0$).) In other words, for any $\e>0$ and with the probability $\ge 1-\e$, the Schr\"odinger evolution $\psi_t$  approaches asymptotically a superposition of states in which  the particle system with a photon cloud bound to it is in one of its bound states $\Phi_{j}$, with additional photons (or possibly none) escaping to infinity with the velocity of light.

\end{itemize}

The reason for $\e>0$ in \eqref{AC} is that for the state $\Phi_{j} \otimes_s f$ to be well defined, as one would expect, one would have to have a very tight control on the number of photons in $f$, i.e. the number of photons escaping the particle system. (See the remark at the end of Subsection \ref{sec:ACpf} for a more technical explanation.) For massive bosons $\e>0$ can be dropped (set to zero), as the number of photons can be bound by the energy cut-off.

We describe the photon position by  the operator $y := i \nabla_k$ on $L^2( \mathbb{R}^3)$, canonically conjugate to the photon momentum $k$ (see \cite{BoFaSig} for a discussion of the notion of  the photon position in our context). We say that the system obeys the quantum Huygens principle if the Schr\"odinger evolution, $\psi_t=e^{-i t H} \psi_0 $, obeys the estimates
\begin{align}\label{HP}
\int_1^\infty dt\  t^{-\al'}
\| \d {\Gamma} (\chi_{\frac{|y|}{ct^{\al}} =1})^{\frac{1}{2}} \psi_t \|^2
 \lesssim  \| \psi_0 \|_0^2,
\end{align}
for some norm $\| \psi_0 \|_0$, some $0< \al'\le 1$, and for any $\al>0$ and $c>0$ such that either  $\al <1$ or $\al =1$ and $c<1$. In other words there are no photons which either diffuse or propagate with speed $<1$. Here $\chi_{ \Om}$ denotes a smoothed out characteristic function of the set $\Om$, which is defined at the end of the introduction. The \textit{maximal velocity estimate}, as proven in \cite{BoFaSig}, states that, for $\mu>0$, any $ \bar{\mathrm{c}}>1$, and $\gamma < \frac{ \mu }{ 2 } \min ( \frac{ \bar{\mathrm{c}} - 1 }{ 3 \bar{\mathrm{c}} - 1 } , \frac{1}{2+\mu} )$,
\begin{equation}\label{maxvel-est}
\big\Vert \d \Gamma \big( \chi_{ \vert y \vert  \geq \bar{\mathrm{c}} t  } \big)^{\frac{1}{2}} \psi_t  \big\Vert \lesssim t^{- \gamma} \big\Vert ( \d \Gamma ( \< y \> ) + 1 )^{\frac12}  \psi_0 \big\Vert.
\end{equation}

Considerable progress has been made in understanding the asymptotic dynamics of non-relativis\-tic particle systems coupled to  quantized electromagnetic or phonon field. The local decay property was proven in \cite{BFS2, BFSS, GGM1,GGM2, FGSig1, FGSig2, BoFa, CFFS}, by  positive commutator techniques and the combination of the renormalization group  and positive commutator methods. The maximal velocity estimate was proven in \cite{BoFaSig}.

An important breakthrough was achieved recently in \cite{DRK}, where the authors proved relaxation to the ground state and uniform bounds on the number of emitted massless bosons in the spin-boson model.

In scattering theory, asymptotic completeness was proven for (a small perturbation of) a solvable model involving a harmonic oscillator (see \cite{Ar, Sp1}), and for models involving massive boson fields (\cite{DerGer2,  FrGrSchl2, FrGrSchl3, FrGrSchl4}). Moreover, \cite{Ger} obtained some important results for massless bosons. Motivated by the many-body quantum scattering, \cite{DerGer2, Ger, FrGrSchl2, FrGrSchl3, FrGrSchl4} defined main notions of the scattering theory on Fock spaces, such as wave operators, asymptotic completeness and propagation estimates.

\medskip

\noindent \textbf{Results.} Now we formulate our results. For notational simplicity we consider \eqref{H}, with the linear coupling \eqref{I}. The coupling operators $g (k)$ are assumed to satisfy
 \begin{equation}\label{g-est}
\|\eta^{|\al|}\p^\al g (k)\|_{\chp} \lesssim |k|^{\mu-|\al|} \xi (k),\quad |\al| \le 2,
\end{equation}
where $\xi (k)$ is the ultra-violet cut-off (a smooth function decaying sufficiently rapidly at infinity) and $\eta$ is an estimating operator on the particle space $\chp$ (a bounded, positive operator with unbounded inverse), satisfying
\begin{equation}\label{eta-bnd}
\|\eta^{-n}f(H)\|\lesssim 1 ,
\end{equation}
for any $n=1, 2$ and  $f \in \mathrm{C}_0^\infty( (-\infty,\Sigma))$. For the particle model discussed in the paragraph containing \eqref{gx}, \eqref{g-est} holds with $\eta = \lan x\ran^{-1}$, where $\lan x \ran = (1+|x|^2)^{1/2}$,  and  the ionization threshold,  $\Sigma$, for which \eqref{eta-bnd} is true, is given by
\begin{equation}\label{Sigma}
 \Sigma : = \lim_{R \to \infty} \inf_{\varphi \in D_{R}} \< \varphi , H \varphi \> ,
\end{equation}
where the infimum is taken over $D_{R} = \{ \varphi \in \D ( H ) |\ \varphi ( x ) = 0 \text{ if }\vert x \vert < R, \Vert \varphi \Vert = 1 \}$ (see \cite{Gr}; $\Sigma$ is close to $\inf\s_{\textrm{ess}}(\hp)$).
 For the spin-boson model defined below, $\eta= \one $.

Below, we assume $\mu> - 1/2$ or $\mu>0$. To apply our techniques to minimally coupled particle systems, where $\mu= - 1/2$, one would have to perform first the generalized Pauli-Fierz transform of \cite{Sig1}, as it is done in \cite{BoFaSig} (see also \cite{GuSig, Sig2}), which brings it to $\mu= 1/2$.

 It is known (see \cite{BFS2, GLL}) that the operator $H$ has \textit{the unique ground state} (denoted here as $\Phi_{\textrm{gs}}$) and that generically (e.g. under the Fermi Golden Rule condition) it has no eigenvalues in the interval $(E_{\textrm{gs}},\Sigma)$, where $E_{\textrm{gs}}$ is the ground state energy (see \cite{BFSS}). \textit{We assume} that  this is exactly the case:
 \begin{equation}\label{EVcond}
\mbox{{\it Fermi's Golden Rule} (\cite{BFS1,BFS2}) holds.}
 \end{equation}
Treatment of the (exceptional) situation when such eigenvalues do occur requires, within our approach, proving a delicate estimate $\| P_\Omega f(H) \| \lesssim \lan g\ran$, where $P_\Omega$ denotes the projection onto $\cH_p \otimes \Omega$ ($\Omega := 1 \oplus 0 \oplus \dots$ is the vacuum in $\cF$) and $f \in \mathrm{C}_0^\infty( ( E_{ \mathrm{gs} } , \Sigma ) \setminus \s_p(H) )$, uniformly in $\dist (\supp f, \s_p(H))$.

In what follows we let $\psi_t$ denote the Schr\"odinger evolution, $\psi_t=e^{-i t H} \psi_0 $, i.e. the solution of Schr\"odinger  equation \eqref{SE}, with an initial condition $ \psi_0$, satisfying $ \psi_0=f( H )\psi_0$, with  $f \in \mathrm{C}_0^\infty( (-\infty,\Sigma))$.

For $A\ge -C$, we denote $\| \psi_0 \|_A:=(\| \psi_0 \|^2+\| (A+C)^{\frac{1}{2}}\psi_0 \|^2)^{1/2}$. We define $\nu(\rho)\ge 0$ by the inequality
 \begin{equation}\label{dGk-bnd}
\langle \psi_t , \d \Gamma( \om^{\rho} ) \psi_t \rangle \lesssim  t^{\nu(\rho)}\| \psi_0 \|^2_\rho,
\end{equation}
where $\|\psi\|_\rho^2=\| \psi \|_H^2 + \| \psi \|_{\d\G(\omega^{\rho})}^2 $. It was shown in  \cite{BoFaSig} (see \eqref{lm-bnd} of Appendix \ref{sec:low-mom}) that, for any $-1 \le \rho \le 1$, the inequality \eqref{dGk-bnd} holds for the the exponent $\nu(\rho)=\frac{1-\rho}{2+\mu}$ (this generalizes an earlier bound due to \cite{Ger}). Also, the bound
\begin{align}\label{Hf-H-bnd}
 \|  \psi_t \|_{H_f} \lesssim \| \psi_0 \|_H
\end{align}
shows that \eqref{dGk-bnd} holds for $\rho=1$ with $\nu(1)= 0$.
 With $\nu(\del)$ defined by \eqref{dGk-bnd}, we prove the following two results.

\begin{theorem}[Quantum Huygens principle]\label{thm:mve1y}
Assume \eqref{g-est} with $\mu > -1/2$ and \eqref{eta-bnd}. Let either $\beta <1$, or $\beta =1$ and $c<1$. Assume
\begin{equation}\label{beta-cond_0}
\beta > \max \Big (  \frac{5}{6} + \frac{\nu(-1)-\nu(0)}{6} , \frac{1}{2} + \frac{1}{ 2 ( \frac32 + \mu ) }  \Big ).
\end{equation}
Then for any initial condition $\psi_0 \in f (H) D(\d\G( \omega^{-1} )^{1/2})$, for some $f \in \mathrm{C}_0^\infty( (-\infty, \Sigma) )$, the Schr\"odinger evolution, $\psi_t$, satisfies, for any $a>1$, the following estimate
\begin{align}\label{mve1y_0}
\int_1^\infty dt\  t^{-\beta-a\nu(0)} \| \d {\Gamma} ( \chi_{\frac{\mid y\mid}{ct^{\beta}} =1})^{\frac{1}{2}} \psi_t \|^2
 \lesssim   \| \psi_0 \|_{-1}^2.
\end{align}
\end{theorem}

\medskip

To formulate our next result we let $\G(\chi)$ be the lifting of a one-photon operator $\chi$ (e.g. a smoothed out characteristic function of $y$) to the photon Fock space, defined by
\begin{equation}\label{def:Gamma}
\Gamma( \chi )  =  \oplus_{n=0}^\infty (\otimes^n \chi) ,
\end{equation}
(so that $\Gamma( e^b )=e^{\d\Gamma( b )}$), and then to the space of the total system. We have
\begin{theorem}[Weak minimal photon escape velocity estimate]\label{thm:mve2y}
Assume \eqref{g-est} with  $\mu > -1/2$, \eqref{eta-bnd} and \eqref{EVcond}. Let the norm $\lan g \ran := \sum_{ |\alpha| \le 2} \| \eta^{| \alpha | } \p^{\alpha} g \|_{L^2(\R^3, \chp)}$ of the coupling function $g$
be sufficiently small and $\nu(-1) < \alpha < 1- \nu(0)$. Then for any initial condition $\psi_0 \in f (H) D(\d\G(\lan y\ran))$, for some $f \in \mathrm{C}_0^\infty( ( E_{\mathrm{gs}}, \Sigma) )$, the Schr\"odinger evolution, $\psi_t$, satisfies the estimate
\begin{equation}\label{mve2y}
\| \Gamma (\chi_{|y| \leq c' t^\alpha} ) \psi_t \|\ \lesssim t^{-\gamma}\| \psi_0 \|_{\d \Gamma ( \langle y \rangle)^2 }\ ,
\end{equation}
where $\gamma < \frac12 \min ( 1-\alpha - \nu(0) , \frac12 (\alpha - \nu(0)-\nu(-1)) )$.
\end{theorem}

\medskip

\noindent \textbf{Remarks.}

1) The estimate \eqref{mve1y_0} is sharp if $\nu(0)=0$.  Assuming this and taking  $\nu(-1) = ( 3/2 + \mu )^{-1}$ (see \eqref{lm-bnd'}), the condition \eqref{beta-cond_0} on $\beta$ in Theorem \ref{thm:mve1y} becomes $\beta > \frac{5}{6}+\frac{1}{6(3/2+\mu)}$, and the condition on $\alpha$ in Theorem \ref{thm:mve2y}, $( 3/2 + \mu )^{-1} < \alpha < 1$.

2) The estimate \eqref{mve2y} states that, as $t\ra \infty$, with probability $\ra 1$, either all photons are attached to the particle system in the combined ground state, or at least one photon departs the particle system with the distance growing at least as $\CO(t^\al)$. (\eqref{mve2y} for $\mu \ge 1/2$, some $\alpha > 0$ and $\psi_0 \in E_\Delta(H)$, with $ \Delta \subset ( E_{\mathrm{gs}} , e_1 - \CO(\lan g \ran) )$ and $e_1$ the first excited eigenvalue of $H_p$, can be derived directly from \cite{BoFa, BoFaSig}.)

3) With some more work, one can remove Assumption \eqref{EVcond} and relax the condition on $\psi_0$ in Theorem \ref{thm:mve2y} to the natural one:  $\psi_0 \in P_\Sigma D(\d\G(\lan y\ran))$, where $P_\Sigma$ is the spectral projection onto the orthogonal complement of the eigenfunctions of  $H$ with the eigenvalues in the interval $ (-\infty, \Sigma)$.

\medskip

Let $N:=\d \Gamma( \one )$ be the photon (or phonon) number operator. Our next result is

\begin{theorem}[Asymptotic Completeness]\label{thm:ac}
Assume \eqref{g-est} with $\mu > 0$, \eqref{eta-bnd} and \eqref{EVcond}. Let the norm $\lan g \ran := \sum_{ |\alpha| \le 2} \| \eta^{| \alpha | } \p^{\alpha} g \|_{L^2(\R^3, \chp)}$ of the coupling function $g$ be sufficiently small. Suppose that
\begin{equation}\label{npb-unif}
\| N^{\frac12} \psi_t \| \lesssim \| N^{\frac12} \psi_0 \| + \| \psi_0 \| ,
\end{equation}
uniformly in $t \in [0 , \infty)$, for any $\psi_0 \in D( N^{1/2} )$. Then the asymptotic completeness holds on $\mathrm{Ran} \, E_{(-\infty , \Sigma) }(H)$.
\end{theorem}

\medskip

As we see from the results above, the uniform bound, \eqref{npb-unif}, on the number of photons (or phonons) emerges as the remaining stumbling block to proving the asymptotic completeness without qualifications.

For massive bosons (e.g. optical phonons),  the inequality \eqref{npb-unif} (as well as \eqref{dGk-bnd}, with $\nu(0)=0$) is easily proven and the proof below simplifies considerably as well. In this case, the result is unconditional. It was first proven in \cite{DerGer2} for the models with  confined particles, and in \cite{FrGrSchl2} for the Rayleigh scattering.

The difficulty in proving this bound for massless particles is due to  the same infrared problem which pervades this field and which was successfully tackled in other central issues, such as the theory of ground states and resonances (see \cite{bach, Sig2} for  reviews), the local decay and the maximal velocity bound. As was mentioned above, for the spin-boson model (see below), a uniform bound, $\< \psi_t , e^{\delta N} \psi_t \> \le C (\psi_0)< \infty,\ \del>0$, on the number of photons, on a dense set of $\psi_0$'s, was recently proven in the remarkable paper \cite{DRK}, which gives substance to our  \textit{conjecture} that the bound \eqref{npb-unif} holds for a dense set of states.

\medskip

\noindent \textbf{Spin-boson model.} Another example of the particle system, and  the simplest one, is the spin-boson model, describing an idealized two-level atom, with state space $\cH_\at \ = \ \CC^2 $, the hamiltonian $H_\at \ = \ \eps \sigma^3$,
where $\sigma^1, \sigma^2, \sigma^3$ are the usual $2\times 2$ Pauli matrices, and $\eps > 0$ is an atomic energy. The coupling family is given by $g (k) = \omega^\mu \kappa (k)\sigma^+$, $\sigma^{\pm}= \frac{1}{2}\:(\sigma^1 \mp i\sigma^2)$.  In this case,  $g$ satisfies \eqref{g-est} with  $\eta=\one$. For the spin-boson model, we can take $\Sigma=\infty$.

\medskip

\noindent \textbf{Approach and organization of the paper.} In this paper, as in earlier works, we use  the method of propagation observables, originating in the many body scattering theory (\cite{SigSof2, SigSof3, HunSigSof, Graf, Yaf, Der}, see \cite{DerGer1, HuSig} for a textbook exposition and a more recent review), and extended to the non-relativistic quantum electrodynamics in \cite{ DerGer2, Ger, FrGrSchl1, FrGrSchl2, FrGrSchl3, FrGrSchl4} and to the  $P(\varphi)_2$ quantum field theory, in \cite{ DerGer3}. We formalize this method in the next section.

After that we prove key  propagation estimates in Sections \ref{sec:pf-mve1} and \ref{sec:pf-mve2}. Instead of $|y|$, these estimates involve the operator $b_\e$ defined as $b_\e :=\frac{1}{2}(v(k) \cdot y + y \cdot v(k) ),$ where $v(k) := \frac{k}{\omega+\e}$, for $\e=t^{-\kappa}$, with some $\kappa>0$. Since the vector field $v(k)$ is Lipschitz continuous and therefore generates a global flow, the operator $b_\e$ is self-adjoint. We show in Section \ref{sec:mve-y} that  these propagation estimates give the  estimates \eqref{mve1y_0} and \eqref{mve2y}.  (The operator $b_\e$ was considered in [I.M.~Sigal and A. Soffer, Unpublished, 2004], as a regularization of the non-self-adjoint operator $b_0$  used in \cite{Ger}. We could have also used  the operators $b_\e$, with $0<\e<\g_0:= {\rm dist} (\Delta, \sigma_p (H_{el}))$ constant, $b := \frac{1}{2}(\frac{k}{\omega} \cdot y + \frac{k}{\omega} \cdot y)$, or $\tilde b :=\frac{1}{2}(k \cdot y + k \cdot y)$.  Using $b_\e$ avoids some (trivial) technicalities, as compared to the other two operators. At the expense of slightly lengthier computations but gaining simpler technicalities, one can also modify $b_\e$ to make it bounded, by multiplying it with the cut-off function $\chi_{ \vert y \vert  \leq \bar{\mathrm{c}} t  }$, with  $ \bar{\mathrm{c}}>1$ such that the maximal velocity estimate \eqref{maxvel-est} holds, or use the smooth vector field  $v(k) :=\frac{k}{\sqrt{\omega^2+\e^2}},$ instead of $v(k) :=\frac{k}{\omega+\e}$.)

Theorem \ref{thm:ac} is proven in Section \ref{sec:pf-ac}. As it is standard in the scattering theory, to prove the asymptotic completeness, we establish the existence of  the Deift-Simon wave operator  $W_+$,  mapping solutions of the Schr\"odinger equation into the scattering data (see \cite{DerGer2, FrGrSchl2, Ger} and \cite{SigSof1, Graf, Yaf, Der} for earlier works). We prove the existence of $W_+$ in Subsection \ref{sec:DSOs} and then deduce from it Theorem \ref{thm:ac} in Subsection \ref{sec:ACpf}. A low momentum bound of \cite{BoFaSig} and some standard technical statements are given in Appendices \ref{sec:low-mom},  \ref{sec:tech} and \ref{sec:tech2}.

The paper is essentially self-contained. In order to make it more accessible to non-experts, we included Supplement I giving standard definitions, proof of the existence and properties of the wave operators, and Supplement II defining and discussing the creation and annihilation operators.

 \medskip

\noindent \textbf{Notations.} For functions $A$ and $B$, we will use the notation $A\lesssim B$ signifying that $A\le C B$ for some absolute (numerical) constant $0<C < \infty$.
The symbol $E_\Delta $ stands for the characteristic function of a set $\Delta$, while $\chi_{ \cdot \le 1 }$ denotes a smoothed out characteristic function of the interval $(-\infty , 1]$, that is it is in $\mathrm{C}^\infty( \mathbb{R} )$,  is non-decreasing, and $ = 1$ if $x \le 1/2$ and $= 0$ if $x \ge 1$. Moreover, $\chi_{ \cdot \ge 1 } := \one - \chi_{ \cdot \le 1 }$ and $\chi_{ \cdot = 1 }$ stands for the derivative of $\chi_{ \cdot \ge 1 }$. Given a self-adjoint operator $a$ and a real number $\al$, we write $\chi_{ a \le \al } := \chi_{ \frac{ a }{ \al } \le 1 }$, and likewise for $\chi_{ a \ge \al }$. Finally, $D(A)$ denotes the domain of an operator $A$.

\medskip

\noindent {\bf Acknowledgements.} The first author is grateful to Jean-Fran\c{c}ois Bony, J\"urg Fr\"ohlich and Christian G\'erard for useful discussions. The last author is grateful to Volker Bach, Jean-Fran\c{c}ois Bony, J\"urg Fr\"ohlich, Marcel Griesemer and Avy Soffer for many discussions and collaboration. His research was supported in part by NSERC under Grant No. NA7901.

\bigskip

\section{ Method of propagation observables}\label{sec:propobs}

Many steps of our proof use the method of propagation observables which we formalize in what follows. In this section we consider the Hamiltonian \eqref{H} and assume \eqref{g-est} and \eqref{eta-bnd}. Let $\psi_t = e^{-i t H}  \psi_0$. The method reduces propagation estimates for our system say of the form
\begin{equation}\label{weak-prop-est'}
\int_0^\infty dt \|G_t^{1/2}\psi_t\|^2 \lesssim \|\psi_0\|^2_\#,
\end{equation}
for some norm $\|\cdot\|_\# \ge \|\cdot\|$, to differential inequalities for certain families $\phi_t$ of positive, one-photon operators on the one-photon space  $L^2(\R^3)$. Let
\begin{equation*}
d \phi_t:= \partial_t \phi_t + i [ \om, \phi_t ] ,
\end{equation*}
and let $\nu (\rho)\ge 0$ be determined by the estimate \eqref{dGk-bnd}.  We isolate the following useful class of families of  positive, one-photon operators:
\begin{definition}\label{def:po-one}
A family of positive operators  $\phi_t$ on $L^2(\R^3)$ will be called a \textit{one-photon weak propagation observable}, if it has the following properties
\begin{itemize}
\item  there are $\del\ge0$ and a family $p_t$ of non-negative operators,  such that
\begin{equation}\label{phit-cond1}
\| \omega^{-\delta/2} \phi_t \omega^{-\delta/2}\|\lesssim t^{-\nu(\delta)}\  \mbox{and}\ \quad
d \phi_t \ge p_t + \sum_{\textrm{finite}}\mathrm{rem}_i ,
\end{equation}
where $\mathrm{rem}_i$ are one-photon operators satisfying
 \begin{equation}\label{rem-est}
  \| \omega^{-\rho_i/2} \, \mathrm{rem}_i \, \omega^{-\rho_i/2}\| \lesssim t^{-\lam_i} ,
\end{equation}
for some $\rho_i$ and $\lam_i$, s.t. $\lam_i > 1+\nu(\rho_i)$,
\item for some $\lam'>1 + \nu ( \delta )$ and with $\eta$ satisfying \eqref{eta-bnd},
\begin{equation}\label{phitg-est}
\big( \int \| \eta \phi_t g \|_{\chp}^2 \omega^{-\delta}  d^3 k \big)^{\frac{1}{2}}
\lesssim t^{-\lam'}.
\end{equation}
(Here $\phi_t$ acts on $g$ as a function of $k$.)
\end{itemize}

Similarly, a family of operators  $\phi_t$ on $L^2(\R^3)$ will be called a \textit{one-photon strong propagation observable}, if
\begin{equation}\label{phit-cond2}
d \phi_t \le -p_t + \sum_{\textrm{finite}}\mathrm{rem}_i ,
\end{equation}
with $p_t\ge 0$, $\mathrm{rem}_i$ are one-photon operators satisfying \eqref{rem-est} for some $\lam_i > 1+\nu(\rho_i)$, and \eqref{phitg-est} holds for some $\lam'>1+\nu(\delta)$.
\end{definition}

The following proposition reduces proving inequalities of the type of \eqref{weak-prop-est'} to showing that $\phi_t$ is a one-photon weak or strong propagation observable, i.e. to \textit{one-photon estimates} of $d \phi_t$ and $\phi_t g$.

\begin{proposition}\label{prop:prop-obs-meth}
If  $\phi_t$ is a one-photon weak (resp. strong) propagation observable, then we have either the weak estimate, \eqref{weak-prop-est'}, or the strong propagation estimate,
 \begin{align}\label{strong-prop-est''}
&\lan \psi_t, \Phi_{t} \psi_t\ran +\int_0^\infty dt \|G_t^{1/2}\psi_t\|^2\lesssim \|\psi_0\|^2_\diamondsuit+ \|\psi_0\|^2_* ,
\end{align}
with the norm $\|\psi_0\|^2_\#:=\|\psi_0\|^2_\diamondsuit+ \|\psi_0\|^2_* $, where  $\Phi_{t} :=  \d\G( \phi_t)$ and $G_t := \d \G( p_t)$, on the subspace $f( H )\cH\subset {\cH}$, with $f \in \mathrm{C}_0^\infty( (-\infty,\Sigma))$. Here $\|\psi_0\|_*:=\|\psi_0\|_\del $ and $\|\psi_0\|_\diamondsuit = \sum\|\psi_0\|_{\rho_i}$.
\end{proposition}

Before proceeding to the proof we present some useful definitions. Consider families $\Phi_t$ of operators on $\cH$ and introduce the Heisenberg derivative
\begin{equation*}
D \Phi_t:= \partial_t \Phi_t + i \big[ H, \Phi_t \big] ,
\end{equation*}
with the property
\begin{align}\label{dt-D}
\partial_t \< \psi_t, \Phi_t \psi_t  \> =\< \psi_t, D\Phi_t  \psi_t  \>.
\end{align}
\begin{definition}\label{def:po}
A family of operators  $\Phi_t$ on a subspace $\cH_1\subset {\cH}$  will be called a (second quantized) \textit{weak propagation observable}, if for all $\psi_0\in \cH_1$, it has the following properties
\begin{itemize}
\item $\sup_t \lan \psi_t, \Phi_{t} \psi_t\ran\lesssim \|\psi_0\|_*^2$;
\item $D\Phi_t\ge G_t +\mathrm{Rem}$, where $G_t\ge 0$ and $ \int_0^\infty dt \langle \psi_t , \mathrm{Rem} \, \psi_t \rangle \lesssim \|\psi_0\|^2_\diamondsuit, $
\end{itemize}
for some norms $\|\psi_0\|_*,\ \|\cdot\|_\diamondsuit\ge \|\cdot\|$. Similarly, a family of operators  $\Phi_t$ will be called a \textit{strong propagation observable}, if it has the following properties
\begin{itemize}
\item $\Phi_t$ is a family of non-negative operators;
\item $D\Phi_t\le -G_t +\mathrm{Rem}$, where $G_t\ge 0$ and
$ \int_0^\infty dt \langle \psi_t , \mathrm{Rem} \, \psi_t \rangle \lesssim \|\psi_0\|^2_\diamondsuit, $
\end{itemize}
for some norm $\|\cdot\|_\diamondsuit\ge \|\cdot\|$.
\end{definition}
If $\Phi_t$ is a weak propagation observable, then  integrating the corresponding differential inequality sandwiched by $\psi_t$'s and using the estimate on $\lan \psi_t, \Phi_{t} \psi_t\ran$ and on the remainder $\mathrm{Rem}$, we obtain the
(weak propagation) estimate \eqref{weak-prop-est'}, with $\|\psi_0\|^2_\#:=\|\psi_0\|^2_\diamondsuit+ \|\psi_0\|^2_*$.
If $\Phi_t$ is a strong propagation observable, then the same procedure leads to the (strong propagation) estimate
\begin{align}\label{strong-prop-est'}
&\lan \psi_t, \Phi_{t} \psi_t\ran +\int_0^\infty dt \|G_t^{1/2}\psi_t\|^2\lesssim \|\psi_0\|^2_\diamondsuit+ \|\psi_0\|^2_*.
\end{align}

\medskip

\noindent \textit{Proof of Proposition \ref{prop:prop-obs-meth}}.  Let $\Phi_t := \d\G( \phi_t )$. To prove the above statement we use the relations
\begin{equation}\label{DdG-rel}
D_0 \d\G(\phi_t)=  \d\G(d \phi_t) ,\qquad i[I(g), \d\G(\phi_t)]= - I( i \phi_t g),
\end{equation}
where  $D_0 $ is the free Heisenberg derivative,
\begin{equation*}
D_0 \Phi_t:= \partial_t \Phi_t + i [ H_0, \Phi_t ],
\end{equation*}
valid for any family of one-particle operators $\phi_t$, to compute
 \begin{align}\label{DPhi}
D\Phi_{t} = \d {\Gamma} (d\phi_t) - I ( i \phi_t g).
\end{align}

Denote $\lan A\ran_\psi:=\lan\psi, A\psi \ran$. Applying the Cauchy-Schwarz inequality, we find the following version of a standard estimate
\begin{eqnarray}\label{I-est}
|\langle I (g)\rangle_{\psi} |
\leq \big( \int \| \eta g \|_{\chp}^2
\omega^{-\delta}  d^3 k \big)^{\frac{1}{2}}
\| \eta^{-1} \psi \|\|  \psi \|_{\d\Gamma( \omega^\delta )}.
\end{eqnarray}
Using that $\psi_t=f_1( H )\psi_t$, with $f_1 \in \mathrm{C}_0^\infty( (-\infty,\Sigma)),\ f_1 f=f ,$ and using  \eqref{eta-bnd}, we find $\|  \eta^{-1}  \psi_t \| $ $\lesssim \| \psi_t \|$. Taking this into account, we see that the equations \eqref{I-est}, \eqref{phitg-est} and  \eqref{Hf-H-bnd}
yield
\begin{eqnarray}\label{I-est2}
|\langle I (i \phi_t g)\rangle_{\psi_t}| \lesssim t^{-\lam' + \nu(\delta) }\| \psi_0 \|_\delta^2.
\end{eqnarray}

Next, using \eqref{rem-est},  we find $\mathrm{rem}_i \le \| \omega^{-\rho_i/2} \, \mathrm{rem}_i \, \omega^{-\rho_i/2} \| \omega^{\rho_i} \lesssim t^{-\lam_i } \omega^{\rho_i}$. This gives $\d\G(\mathrm{rem_i}) \lesssim t^{-\lam_i} \d \Gamma ( \omega^{ \rho_i} )$,
which, due to the bound \eqref{dGk-bnd}, leads to the estimate
\begin{equation}\label{Rem-est}
 \lan \d\G(\mathrm{rem_i})\ran_{\psi_t}\lesssim t^{-\lam_i+\nu(\rho_i)}\|\psi_0\|^2_{\rho_i}.
\end{equation}

In the strong case, \eqref{phit-cond2} and \eqref{DPhi} imply
 \begin{align}\label{DPhi-est-s}
D\Phi_{t} \le - \d {\Gamma} (p_t) + \sum_{\textrm{finite}} \d {\Gamma} (\mathrm{rem_i}) - I ( i \phi_t g),
\end{align}
which together with \eqref{I-est2} and \eqref{Rem-est} implies that  $\Phi_t$ is a strong propagation observable.

In the weak case, \eqref{phit-cond1} and \eqref{DPhi} imply
 \begin{align}\label{DPhi-est-w}
D\Phi_{t} \ge  \d {\Gamma} (p_t) + \sum_{\textrm{finite}} \d {\Gamma} (\mathrm{rem_i}) - I ( i \phi_t g).
\end{align}
Next, since $\phi_t \le \| \omega^{-\delta/2} \phi_t \omega^{-\delta/2} \| \omega^{\delta} \lesssim t^{-\nu(\delta)} \omega^{\delta} $, we have $\d\G(\phi_t) \lesssim t^{-\nu(\delta)} \d \Gamma ( \omega^{ \delta} )$. Using this estimate and using again the bound  \eqref{dGk-bnd}, we obtain
\begin{eqnarray}\label{Phit-est}
\lan \psi_t, \Phi_{t} \psi_t\ran\lesssim t^{-\nu(\delta)}\langle \d\G(\omega^{\delta}) \rangle_{\psi_t} \lesssim \|\psi_0\|^2_\delta.
\end{eqnarray}
Hence  $\Phi_t$ is a weak propagation observable.
\qed

\begin{proposition}\label{prop:prop-obs-meth2}
Let $\phi_t$ be a one-photon family satisfying
\begin{itemize}

\item   either, for some  $\delta\ge0$ ,
\begin{equation}\label{phit-cond1'}
\| \omega^{-\delta/2} \phi_t \omega^{-\delta/2}\|\lesssim t^{-\nu(\delta)}\  \mbox{and}\ \quad
d \phi_t \ge p_t - d\tilde \phi_t + \mathrm{rem} ,
\end{equation}
or
\begin{equation}\label{phit-cond2'}
d \phi_t \le -p_t+ d\tilde \phi_t + \sum_{\mathrm{finite}} \mathrm{rem_i} ,
\end{equation}
where $p_t\ge 0$, $\mathrm{rem}_i$ are one-photon operators satisfying \eqref{rem-est}, and $\tilde\phi_t$ is a weak propagation observable,

\item  \eqref{phitg-est} holds.

\end{itemize}
 Then, depending on whether \eqref{phit-cond1'} or \eqref{phit-cond2'} is satisfied, $\Phi_{t} :=  \d\G( \phi_t)$ is a weak, or strong, propagation observable, with the norm $\|\psi_0\|_\diamondsuit = \|\psi_0\|_\rho$, on the subspace $f( H )\cH\subset {\cH}$, with $f \in \mathrm{C}_0^\infty( (-\infty,\Sigma))$,   and therefore we have either the weak or strong propagation estimates, \eqref{weak-prop-est'} or \eqref{strong-prop-est'}, on this subspace.
\end{proposition}
\begin{proof} Given Proposition \ref{prop:prop-obs-meth2}  and its proof, the only term we have to control is $\d\G( d\tilde\phi_t)$. Using that $\tilde\phi_t$ is a weak propagation observable and using \eqref{dt-D}, \eqref{DPhi} and \eqref{I-est2} for $\tilde\Phi_t :=\d\G( \tilde\phi_t)$, we obtain
\begin{equation}\label{dGtildephi-est}
| \int_0^{\infty} dt \lan \tilde \Phi_t \ran_{\psi_t}| \lesssim \|\psi_0\|^2_\# ,
\end{equation}
 with $\|\psi_0\|^2_\#:=\|\psi_0\|^2_\diamondsuit+ \|\psi_0\|^2_*$ ($\|\psi_0\|_\diamondsuit$ and $ \|\psi_0\|_*$ might be different now), which leads to the desired estimates.
 \end{proof}

\bigno{\bf Remarks.}

1) Proposition \ref{prop:prop-obs-meth} reduces a proof of propagation estimates for the dynamics \eqref{SE} to estimates involving the \textit{one-photon} datum $(\om, g)$ (an `effective one-photon system'), parameterizing the hamiltonian \eqref{H}.
(The remaining datum $\hp$ does not enter our analysis explicitly, but through the bound states of $\hp$ which lead to the localization in the particle variables, \eqref{exp-bnd}).

2) The condition on the remainder in \eqref{phit-cond1} can be weakened to $\mathrm{rem}=\mathrm{rem}'+\mathrm{rem}''$, with $\mathrm{rem}'$ and $\mathrm{rem}''$ satisfying \eqref{rem-est} and
\begin{equation}\label{rem-est''}
| \mathrm{rem}'' | \lesssim \chi_{|y|\ge \bar c t},
\end{equation}
for $\bar c$ as in \eqref{maxvel-est}, respectively. Moreover, \eqref{rem-est} can be further weakened to
\begin{equation}\label{rem-est-w}
\int_0^\infty |Ê\langle \psi_t , \d\G( \mathrm{rem}_i ) \psi_t \rangle | <\infty.
\end{equation}

3) An iterated form of Proposition \ref{prop:prop-obs-meth2} is used to prove Theorem \ref{thm:mve1y}.

\bigskip

\section{The first propagation estimate}\label{sec:pf-mve1}

Let $\nu(\del)\ge 0$ be the same as in \eqref{dGk-bnd} and recall the operator $b_\e$ defined in the introduction. We write it as
\begin{align}\label{be}
b_\e :=\frac{1}{2}(\theta_\e \nabla \omega \cdot y + y \cdot \nabla \omega  \theta_\e),\ \quad \mbox{where}\ \quad \theta_\e :=\frac{\omega}{\omega_{\e}},\ \omega_{\e}:=\omega + \e,\ \e=t^{-\kappa}.
\end{align}

\begin{theorem}\label{thm:mve1}
Assume \eqref{g-est} with $\mu > -1/2$ and \eqref{eta-bnd}. Let $\nu(  - 1 ) - \nu( 0 ) < \kappa < 1$.

If  either $\beta <1$, or $\beta =1$ and $c<1$, and
\begin{equation}\label{eq:cond2beta}
\beta > \max( ( 3/2 + \mu )^{-1}, (1+\kappa)/2 , 1 - \kappa + \nu(  - 1 ) - \nu ( 0 ) ),
\end{equation}
then for any initial condition $\psi_0 \in f (H) D(\d\G( \omega^{-1} )^{1/2})$, for some $f \in \mathrm{C}_0^\infty( (-\infty, \Sigma) )$, the Schr\"odinger evolution, $\psi_t$, satisfies, for any $a>1$, the following estimates
\begin{align}\label{mve1}
\int_1^\infty dt\  t^{-\beta-a\nu(0)}
\| \d {\Gamma} (\chi_{\frac{b_\e}{ct^{\beta}} =1})^{\frac{1}{2}} \psi_t \|^2 \lesssim  \| \psi_0 \|_{-1}^2 .
\end{align}

If  $\nu(0) = 0$, $\mu>0$, and $\beta$ satisfies \eqref{eq:cond2beta} and $\beta < \frac{1}{\bar c}$, with $\bar c > 1$, then, with  the notation $\chi\equiv \chi_{(\frac{|y|}{\bar c t})^2 \le 1}$,
\begin{align}\label{mve1'}
\int_1^\infty dt\  t^{-\beta}
\| \d {\Gamma} (\theta_\e^{1/2}\chi\chi_{\frac{b_\e}{ct^{\beta}} =1}\chi\theta_\e^{1/2})^{\frac{1}{2}} \psi_t \|^2
 \lesssim  \| \psi_0 \|_{0}^2 .
\end{align}
\end{theorem}

\begin{proof}
We will use the method of propagation observables outlined in Section \ref{sec:propobs}. We consider the one-parameter family of one-photon operators
\begin{align}\label{mv-obs}
\phi_t:= t^{-a \nu(0) } \chi_{v \geq 1},\quad v := \frac{b_\e}{c t^{\beta}} ,
\end{align}
where $a>1$. To show that $\phi_{t}$ is a weak one-photon propagation observable, we obtain differential inequalities for  $\phi_t$.  We use the notation $\chi_\beta\equiv \chi_{v \geq 1}$.
To compute $d \phi_t$,  we use the expansion
\begin{align}\label{dphit-exp}
d  \phi_t=  t^{-a\nu(0)}(d v) \chi_\beta' +  \sum_{i=1}^2 \mathrm{rem}_i  , \quad \mathrm{rem}_1:= t^{-a\nu(0)}[d \chi_\beta -  (d v) \chi_\beta'],\quad \mathrm{rem}_2:= -a\nu(0)t^{-1}\phi_{t}.
\end{align}
Using the definitions in \eqref{be}, we compute
\begin{align}\label{dv}
d v &= \frac{\theta_\e}{c t^{\beta}} -\frac{\beta b_\e}{c t^{\beta+1}} + \frac{1}{c t^{\beta}}\p_t b_\e.
\end{align}

Next,  we have $\p_t b_\e=\frac{\kappa}{2 t^{1+\kappa}} (\omega_{\e}^{-1} \theta_\e \nabla \om \cdot y+ {\hbox{ h.c.}})$ on $D( b_\e )$, which, due to the relation $\frac{1}{2} (\omega_{\e}^{-1} \theta_\e \nabla \om \cdot y+ {\hbox{ h.c.}}) = \omega_{\e}^{-1/2} b_\e \omega_{\e}^{-1/2}$, becomes
\begin{align}\label{dtb2}
\p_t b_\e=\frac{\kappa}{ t^{1+\kappa}} \omega_{\e}^{-1/2} b_\e \omega_{\e}^{-1/2}.
\end{align}
Using that (see Lemma \ref{lm:comm_est} of Appendix \ref{sec:tech})
\begin{align*}
\omega_{\e}^{-1/2} b_\e \omega_{\e}^{-1/2} \chi'_\beta = \omega_{\e}^{-1/2} b_\e \chi'_\beta \omega_{\e}^{-1/2} + \CO( t^\kappa ),
\end{align*}
and that $b_\e \ge 0$ on $\supp \chi'_\beta$, we obtain
\begin{align}\label{dtb3}
\frac{1}{c t^{\beta}}\p_t b_\e \chi_\beta' \ge -\frac{\mathrm{const}}{ t^{1 + \beta-\kappa}}.
\end{align}
The relations \eqref{dphit-exp}--\eqref{dtb3}, together with $\frac{ b_\e }{ c t^\beta } \chi'_\beta \le \chi'_\beta$,   imply
\begin{align}\label{dphi-est2}
&d  \phi_t\ge   (\frac{\theta_\e}{c t^{\beta}}- \frac{\beta}{t}) \chi'_\beta + \sum_{i=1}^3 \mathrm{rem}_i ,
\end{align}
where $\mathrm{rem}_1$ and $\mathrm{rem}_2$ are given in \eqref{dphit-exp} and
\begin{align}\label{rem2}
&\mathrm{rem}_3 =\CO( t^{-1- \beta+ \kappa- a \nu(0) }).
\end{align}
This,  together with $\theta_\e =1-\frac{t^{-\kappa}}{\omega_\e}$ and $\omega_\e^{-1} \chi'_\beta = \omega_\e^{-1/2} \chi'_\beta \omega_\e^{-1/2} + \CO( t^{-\beta + \kappa} )$ (see again Lemma \ref{lm:comm_est} of Appendix \ref{sec:tech}), implies
\begin{align}\label{dphi-est3}
&d  \phi_t\ge (\frac{1}{c t^{\beta}}- \frac{\beta}{t}) \chi'_\beta + \sum_{i=1}^4 \mathrm{rem}_i ,\ \quad \mathrm{rem}_4:=\frac{1}{c t^{\beta+\kappa+ a \nu(0) }} \omega_\e^{-1/2} \chi'_\beta \omega_\e^{-1/2}.
\end{align}

We have $\| \phi_{t} \|\le t^{-a\nu(0)}$ and therefore, due to \eqref{dGk-bnd},  the first estimate in \eqref{phit-cond1}  holds.  If either $\beta <1$ (and $t$ sufficiently large), or $\beta =1$ and $c<1$, then $p_t := (\frac{1}{c t^{\beta}}- \frac{\beta}{t})$ is non-negative, which implies the second estimate in  \eqref{phit-cond1}.
Thus \eqref{phit-cond1} holds. By the definition \eqref{dphit-exp} and Corollary \ref{cor:rem1-est} of Appendix \ref{sec:tech} for $i=1$, and by an explicit form  for $i=2, 3, 4$, we have the estimates
\begin{equation}\label{remi-est}
  \| \omega^{-\rho_i/2} \, \mathrm{rem}_i \, \omega^{-\rho_i/2}\|\lesssim t^{-\lam_i},
\end{equation}
$i=1, 2, 3,4$, with $\rho_1 = \rho_2 = \rho_3 = 0$, $ \rho_4 = -1$, $\lam_1=2\beta - \kappa +a\nu(0)$, $\lam_2=1+a\nu(0)$, $\lam_3 = 1 + \beta - \kappa + a \nu(0)$, and $\lam_4= \beta + \kappa +a\nu(0)$. We remark here that the $i=2$ term is absent if $\nu(0)=0$.

The relation \eqref{remi-est} together with the \textit{assumption} $\kappa\le 1$ implies \eqref{rem-est} with $\rho =\rho_i$ and  $\lam =\lam_i$, for $\mathrm{rem}=\mathrm{rem}_i$, provided $\lam_i> 1+ \nu (\rho_i)$.

Finally, \eqref{phitg-est} with $\lam' < a \nu( 0 ) + (\frac{3}{2} + \mu)\beta$, holds, by \cite[Lemma 3.1]{BoFaSig}, with $b_\e$ instead of $|y|$ (See Lemma \ref{lem:chig-est} in Appendix \ref{sec:tech} of the present paper.). Hence  $\phi_t$ is a  weak one-photon propagation observable, provided $2\beta> 1 + \kappa + \nu (0)-a\nu(0)$, $ \beta - \kappa > \nu (0) - a\nu(0)$, $\beta + \kappa >1 + \nu( -1 ) - a \nu(0)$, and $(\frac{3}{2} + \mu )\beta>1$. Therefore, by Proposition \ref{prop:prop-obs-meth} and under the conditions on the parameters above,
\begin{align}\label{est1}
\int_1^\infty dt\  t^{-\beta-a\nu(0)}\| d {\Gamma} ( \chi'_\beta )^{\frac{1}{2}} \psi_t \|^2 \lesssim  \| \psi_0 \|_{-1}^2 .
\end{align}
This, due to the definition of $\chi'_\beta$, implies  the estimate \eqref{mve1}.

\smallskip

 We  now prove \eqref{mve1'}. We use again the notation $\chi_\beta\equiv \chi_{v \geq 1}$, where $v :=\frac{b_\e}{c t^{\beta}}$, and we denote $w:=(\frac{|y|}{\bar c t})^2$. We consider the one-parameter family of one-photon operators
\begin{align}\label{mv-obs_b}
\phi_t:=\chi\chi_\beta\chi,
\end{align}
and show that $\phi_t$ is a weak one-photon propagation observable.  We have  $\|\phi_t \|\le 1$ and therefore, due to \eqref{dGk-bnd} and  the assumption $\nu(0)=0$, the first estimate in \eqref{phit-cond1}  holds. Now, we show the second estimate in  \eqref{phit-cond1}.
To compute $d \phi_t$,  we use the expansion
\begin{align}\label{dphit-exp_2}
d  \phi_t= \chi (d v) \chi_\beta'\chi&+\chi' (d w) \chi_\beta\chi +\chi\chi_\beta (d w)\chi'  + \sum_{i=1,2} \mathrm{rem}_i ,
\end{align}
where
\begin{align}\label{rem1_2}
\mathrm{rem}_1& :=\chi (d \chi_\beta -  (d v) \chi_\beta')\chi ,\ \quad
 \mathrm{rem}_2:= (d \chi - ( d w ) \chi' ) \chi_\beta \chi + \mathrm{h.c.} .
\end{align}

As in \eqref{dv}--\eqref{dtb3}, we have
\begin{align}\label{dv4}
\chi (d v) \chi_\beta'\chi \ge \chi ( \frac{\theta_\e}{c t^{\beta}} -\frac{\beta b_\e}{c t^{\beta+1}} ) \chi_\beta' \chi + \mathrm{rem}_3 ,
\end{align}
where $\mathrm{rem}_3 = \CO( t^{-1-\beta+\kappa} )$. We consider the term $ -\frac{\beta b_\e}{c t^{\beta+1}}$ in \eqref{dv4}. Since $b_\e = \theta_\e^{1/2} b \theta_\e^{1/2}$, we obtain, using in particular Lemma \ref{lm:comm_est} of Appendix \ref{sec:tech}, that
\begin{align*}
\chi b_\e \chi'_\beta \chi &= \chi (\chi'_\beta)^{1/2} \theta_\e^{1/2} b \theta_\e^{1/2} (\chi'_\beta)^{1/2} \chi \notag \\
&= \theta_\e^{1/2} (\chi'_\beta)^{1/2} \chi b \chi (\chi'_\beta)^{1/2} \theta_\e^{1/2}  + \CO( t^\kappa ),
\end{align*}
and the maximal velocity cut-off gives $\chi b \chi \le \bar c t$. Thus, commuting again $\chi$ through $\theta_\e^{1/2}$ and  $(\chi_\beta')^{1/2}$, we obtain
\begin{align}\label{dv-part-est_3_2}
- \chi \frac{\beta b_\e}{c t^{\beta+1}} \chi_\beta'\chi \ge - \frac{\beta \bar c}{c t^{\beta}}  \chi \theta_\e^{1/2} \chi_\beta' \theta_\e^{1/2}\chi + \CO(\frac{1}{ t^{1 + \beta-\kappa}}).
\end{align}
Proceeding in the same way for the term $\frac{\theta_\e}{c t^{\beta}}$ in \eqref{dv4} gives
\begin{align}\label{dv-part-est}
\chi(\frac{\theta_\e}{c t^{\beta}} -\frac{\beta b_\e}{c t^{\beta+1}}) \chi_\beta'\chi \ge \frac{1-\beta \bar c}{c t^{\beta}}  \chi\theta_\e^{1/2} \chi_\beta' \theta_\e^{1/2}\chi + \CO(\frac{1}{ t^{2\beta-\kappa}}).
\end{align}

Next, we compute $d w =2(\frac{b}{(\bar c t)^2}  - (\frac{|y|}{\bar c t})^2\frac{1}{t})$, where, recall, $b = \frac12 ( \nabla \omega \cdot y + \mathrm{h.c.})$. By Lemma \ref{lm:comm_est} of Appendix \ref{sec:tech}, we have
\begin{align}\label{dphit-exp_3_2}
\chi' (d w) \chi_\beta\chi +\chi\chi_\beta (d w) \chi'  = - 2 (\chi_\beta)^{1/2} (-\chi' \chi)^{1/2} (d w) (-\chi' \chi)^{1/2} (\chi_\beta)^{1/2} + \CO( \frac{1}{ t^{1+\beta-\kappa} } ) .
\end{align}
Using that $d w\le (\frac{1}{\bar c }-1)\frac{1}{t} $ on the support of $\chi'$ and that $\chi'\le 0$ and $\bar c>1$, we obtain
\begin{align}\label{dchi-part-est_2}
(-\chi' \chi)^{1/2} (d w) (-\chi' \chi)^{1/2} \ge (1  - \frac{1}{\bar c })\frac{1}{t}(-\chi' \chi).
\end{align}
The relations \eqref{dphit-exp_2}, \eqref{dv4}, \eqref{dphit-exp_3_2} and \eqref{dchi-part-est_2} imply
\begin{align}\label{dphit-exp2}
d  \phi_t &\ge p_t +\tilde p_t- \sum_{i=1, 2, 3,4} \mathrm{rem}_i,
\end{align}
where $\mathrm{rem}_4= \CO(\frac{1}{ t^{2\beta-\kappa}})$ and
\begin{align}\label{p-rem_2}
&p_t := \frac{1-\beta \bar c}{c t^{\beta}} \theta_\e^{1/2} \chi\chi_\beta' \chi\theta_\e^{1/2},\\&
\tilde p_t := (1  - \frac{1}{\bar c })\frac{1}{t} \chi_\beta^{1/2}(-\chi')\chi\chi_\beta^{1/2}.
\end{align}

The terms $p_t$ and $\tilde p_t$ are non-negative, provided $\beta <1/\bar c$ and $\bar c>1$. Together with the assumption $\nu(0)$, this implies \eqref{phit-cond1}.  Next, we claim the estimates
\begin{equation}\label{remi-est_2}
  \|  \mathrm{rem}_i \|\lesssim t^{-\lam},
\end{equation}
$i=1, 2, 3,4$, with  $\lam=2\beta-\kappa$. Indeed, the definition \eqref{rem1_2} and Corollary \ref{cor:rem1-est} of Appendix \ref{sec:tech} imply \eqref{remi-est_2} for $i=1$. The estimate  for $i= 3,4$ are obvious. To estimate $\mathrm{rem}_2$, we write
\begin{align*}
(d \chi-  ( d w ) \chi' ) \chi_\beta \chi = (d \chi - ( d w ) \chi' ) \frac{ b_\e }{ c t^\beta } \tilde \chi_\beta \chi ,
\end{align*}
where $\tilde \chi_\beta = ( \frac{ b_\e }{Êc t^\beta } )^{-1} \chi_\beta$, and $b_\e = \theta_\e b + i \e \omega_\e^{-2}$. Using that, by Lemma \ref{a65} of Appendix \ref{sec:tech},
\begin{equation*}
\big \|Êd \chi-  ( d w ) \chi' \| \lesssim t^{-1},
\end{equation*}
and commuting $b$ through $\tilde \chi_\beta$ gives
\begin{align}\label{eq:F1}
(d \chi- (  d w ) \chi' ) \chi_\beta \chi = \frac{1}{ c t^\beta } (d \chi - ( d w ) \chi' ) \theta_\e \tilde \chi_\beta b \chi + \CO( \frac{1}{ t^{1+\beta-\kappa} } ).
\end{align}
By Lemma \ref{a65}, we also have
\begin{equation*}
\big \|Ê(d \chi-  ( d w ) \chi' ) \omega \| \lesssim t^{-2}.
\end{equation*}
Combining this with \eqref{eq:F1} and the estimates $\omega_\e^{-1} = \CO( t^\kappa )$ and $b \chi = \CO( t )$, we obtain
\begin{align}\label{eq:F2}
(d \chi- ( d w ) \chi' ) \chi_\beta \chi =  \CO( \frac{1}{ t^{1+\beta-\kappa} } ) ,
\end{align}
and hence the estimate for $i=2$ follows.

The relation \eqref{remi-est_2} implies \eqref{rem-est} with $\lam = 2\beta - \kappa$, for $\mathrm{rem}=\mathrm{rem}_i$, provided $2\beta-\kappa > 1$. Finally, as above, \eqref{phitg-est} holds with $\lam' < a \nu( 0 ) + (\frac{3}{2} + \mu)\beta$ by Lemma \ref{lem:chig-est} of Appendix \ref{sec:tech}. This yields \eqref{mve1'}.
\end{proof}

\bigskip

\section{The second propagation estimate}\label{sec:pf-mve2}

We introduce the norm $\lan g \ran := \sum_{ |\alpha| \le 2} \| \eta^{| \alpha | } \p^{\alpha} g \|_{L^2(\R^3, \chp)}$,
for the coupling function $g$.

\begin{theorem}\label{thm:mve2}
Assume \eqref{g-est} with  $\mu > -1/2$, \eqref{eta-bnd} and \eqref{EVcond}. Let $\langle g \rangle$ be sufficiently small, $\nu(-1) < \kappa < 1$, and $0 < \alpha < 1$.
Let $f \in \mathrm{C}_0^\infty( (E_{ \mathrm{gs} } , \Sigma) )$ and $\psi_0 \in \cD := f (H) D(\d\G(\lan y\ran) )$.
Then the Schr\"odinger evolution, $\psi_t$, satisfies the estimate
\begin{equation}\label{mve2}
\| \Gamma (\chi_{b_\e \leq c' t^\alpha} )^{\frac{1}{2}} \psi_t \|\ \lesssim t^{-\del} \| \psi_0 \|_{\d \Gamma ( \langle y \rangle)^2 }\ ,\end{equation}
for $0\le \del < \frac12 \min(\kappa - \nu(-1) , 1 - \kappa , 1-\alpha - \nu(0) )$ and any $c'>0$.
\end{theorem}

We define $B_\e: = \d \Gamma (b_\e)$.  As is \cite[Proposition B.3 and Remark B.4]{BoFaSig}, one verifies that $\cD \subset D( \d \Gamma( \langle y \rangle ) ) \subset D( B_\e )$. The proof of Theorem \ref{thm:mve2} is based on the following result  (cf. \cite{SigSof2, HunSigSof}).
\begin{proposition}\label{prop:B-prop-est-t}
Under the conditions of Theorem \ref{thm:mve2}, the evolution  $\psi_t = e^{-iHt} \psi_0$ obeys
\begin{equation}\label{B-prop-est}
\| \chi_{B_\e \le c t } \psi_t \|\ \lesssim t^{-\del'} \| \psi_0\|_{\d \Gamma ( \langle y \rangle)^2},
\end{equation}
where $\del':=\frac12 \min( \frac{1-C\lan g\ran}{c} -1 - \kappa, 1 - \kappa , \kappa - \nu(-1))$.
\end{proposition}
\noindent \textbf{Remark.} The constant $C$ is independent of $\gamma_0 := \mathrm{dist} (E_{Ê\mathrm{gs} } , \supp f )$ (but the implicit constant appearing in the right hand side of \eqref{B-prop-est} does depend on $\gamma_0$).
\begin{proof}
Let $\e>0$ be a constant. Let $\rho<\min(\frac{1-C\lan g\ran}{c} -1 , 1)$ where $C>0$ is a positive constant defined below. Consider the propagation observable
\begin{equation*}
\Phi_t:= - t^{\rho} \varphi \bigg( \frac{ B_\e}{ct} \bigg),
\end{equation*}
where $\varphi \big( \frac{ B_\e}{ct} \big):= \big(\frac{B_\e }{c t}-2 \big) \chi_{B_\e \leq c t}$. Note that $\varphi\le 0$, but $\varphi'\ge 0$.  Let $\varphi' =\varphi_1^2$. The relations below are understood in the sense of quadratic forms on $\cD$. The IMS formula gives
\begin{equation}\label{DPhi-deco}
D\Phi_t =M+R ,
\end{equation}
where $M:=- t^{\rho} \varphi_1 D ((c t)^{-1} B_\e) \varphi_1 -\rho  t^{-1+\rho}\varphi$ and
\begin{equation}\label{R}
R := \frac{1}{c t^{1-\rho}}[[ B_1, \varphi_1 ], \varphi_1 ]+ t^{\rho}\big([H, \varphi]-\frac{1}{2ct}(\varphi' B_1 +B_1\varphi')\big),
\end{equation}
where $B_1 := i[H, B_\e]$. First, we compute the main term,  $M$, in \eqref{DPhi-deco}.  We leave out a standard proof of $f(H) \in \mathrm{C}^1 ( B_\e )$ (see e.g. \cite[Theorem 8]{FGSig1}) and standard domain questions such as that $ \cD \subset D(B_\e)$. We have
\begin{equation}\label{D}
D \bigg( \frac{B_\e}{c t}\bigg) = \frac{1}{c t} DB_\e - \frac{1}{t} \frac{ B_\e}{c t}.
\end{equation}

The computations below are understood  in the sense of quadratic forms on $ \cD$.
Since $DB_\e=i [H_f,B_\e]=N_\e$, where $N_\e:= \d \Gamma (\theta_\e)$, we have
\begin{equation}\label{DB}
 DB_\e = N_\e +\tilde{I},
\end{equation}
where $\tilde{I} :=i [I(g), B_\e]$. To estimate the operator $N_\e$ from below, we use that  $\theta_\e =1-\frac{\e}{\omega_\e},$ to obtain
\begin{align}\label{Ne-lwbnd}
N_\e \ge  N-\e \d \Gamma (\omega_\e^{-1}) .
\end{align}

Next, we estimate the term $\varphi_1 \d \Gamma (\omega_\e^{-1})\varphi_1$. Using
\begin{align*}
[ \d \Gamma ( \omega_\e^{-1} ) , i ( \frac{ B_\e }{ ct } - z )^{-1} ] = -(ct)^{-1} ( \frac{ B_\e }{ ct } - z )^{-1} \d \Gamma ( \theta_\e \omega_\e^{-2} ) ( \frac{ B_\e }{ ct } - z)^{-1},
\end{align*}
we obtain that
\begin{align*}
&\| [ \d \Gamma ( \omega_\e^{-1} ) , ( \frac{ B_\e }{ ct } - z )^{-1} ] ( N + 1 )^{-1} \| \lesssim t^{-1} \e^{-2} | \mathrm{Im} z |^{-2} ,
\end{align*}
and hence, since $B_\e$ commutes with $N$, the Helffer-Sj{\"o}strand formula shows that
\begin{align*}
\| [ \d \Gamma ( \omega_\e^{-1} ) , \varphi_1 ] ( N + 1 )^{-1} \| \lesssim t^{-1} \e^{-2}.
\end{align*}
Since, in addition, $\| \d \Gamma( \omega_\e^{-1} ) u \| \le \| \d \Gamma( \omega^{-1} ) u \|$, we deduce that
\begin{align*}
\| \d \Gamma ( \omega_\e^{-1} ) \varphi_1 ( \d \Gamma ( \omega^{-1} ) + t^{-1} \e^{-2} (N + 1) )^{-1} \| \lesssim 1,
\end{align*}
and therefore, by interpolation and  \eqref{dGk-bnd}, we arrive at
\begin{align}\label{dGom-est}
\lan \varphi_1 \d \Gamma (\omega_\e^{-1}) \varphi_1 \ran_{\psi_t} \lesssim t^{\nu(-1)} \| \psi_0 \|_{-1}^2 + t^{-1+\nu(0)} \e^{-2} \| \psi_0 \|_0^2.
\end{align}

By the condition $\mu > -1/2$ and \eqref{I-est} (with $\delta = 0$), we have $ \tilde{I}  \ge - C \lan g\ran( N + \eta^{-2} + 1)  $. Combining this with the definition of $M$, \eqref{eta-bnd}, \eqref{D},  \eqref{DB},  \eqref{Ne-lwbnd} and \eqref{dGom-est}, we obtain
\begin{align}\label{M-est1}
\langle M  \rangle_{\psi_t} \le & -\frac{1}{c t^{1-\rho}}\langle \varphi_1 [(1-C\lan g\ran) N - t^{-1}B_\e -C\lan g\ran]\varphi_1+c\rho \varphi \rangle_{\psi_t} \notag\\
&+\frac{C}{ t^{1-\rho}}(\e t^{\nu(-1)}\|  \psi_0\|_{-1}^2+t^{-1 + \nu(0)}\e^{-1} \| \psi_0\|_0^2).
\end{align}

Let $\Om:=1\oplus 0 \oplus \dots$ be the vacuum in $\cF$ and  $P_\Omega$, the orthogonal projection on the subspace $\chp\otimes \Om$, $P_\Omega \Psi:= \lan \Om, \Psi\ran_\cF \otimes \Om$. We now use the following
\begin{lemma}\label{lm:unif_dist}
Assume \eqref{g-est} with $\mu > -1/2$, \eqref{eta-bnd} and \eqref{EVcond}. Let $\lan g \ran$ be sufficiently small and $f \in \mathrm{C}_0^\infty( ( E_{ \mathrm{gs} } , \Sigma ) )$. Then
\begin{align}\label{eq:unif_dist}
\| P_\Omega e^{-itH} f(H) u  \| \lesssim  t^{-s} \| \lan \tilde B \ran u \|, \quad s < 1/2 ,
\end{align}
where $\tilde B = \d \Gamma( \tilde b )$ with, recall, $\tilde b = \frac{1}{2} ( k \cdot y + y \cdot k )$.
\end{lemma}
\begin{proof}
We use the local decay properties established in \cite{FGSig2} and \cite{BFSS}. Let $c_j := (e_j + e_{j+1})/2$ and $\delta_j := e_{j+1} - e_j$. We decompose the support of $f$ into different regions, writing
\begin{align}\label{eq:decom-supp-f}
f(H) = f(H) \chi_{ H \le c_0 } + \sum_{ \mathrm{finite} } f(H) \chi_j(H) ,
\end{align}
where $\chi_j(H)$ denotes a smoothed out characteristic function of the interval $[ c_j - \delta_j / 4 , c_{j+1} + \delta_{j+1}/4 ] $. Using $P_\Omega = P_\Omega \lan \tilde B \ran$, and \cite{FGSig2}, we obtain
\begin{align}\label{eq:unif_dist_1}
& \| P_\Omega e^{-itH} f(H) \chi_{ H \le c_0 } u \| = \| \lan \tilde B \ran^{-1} e^{-itH} f(H) \chi_{ H \le c_0 } u \| \lesssim  t^{-s} \| \lan \tilde B \ran u \|,
\end{align}
for $s < 1/2$.

 To estimate $\| P_\Omega e^{-itH} f(H) \chi_j(H) u \|$, we let $\tilde \chi_j(H) := f(H) \chi_j(H)$. In \cite{BFSS}, assuming \eqref{EVcond}, a conjugate operator $\tilde B_j$ is constructed in such a way that the commutators $[\tilde \chi_j(H) , \tilde B_j ]$ and $[[ \tilde \chi_j(H) , \tilde B_j ] , \tilde B_j ]$ are bounded. Moreover, the Mourre estimate
\begin{align*}
\tilde \chi_j(H) [ H , i \tilde B_j ] \tilde \chi_j(H) \geq m_0 \tilde \chi_j(H)^2,
\end{align*}
holds for some positive constant $m_0$. By an abstract result of \cite{HunSigSof}, this implies
\begin{equation*}
\big \| \langle \tilde B_j \rangle^{- s} e^{-i t H} \tilde \chi_j (H) \langle \tilde B_j \rangle^{-s} \big \| \lesssim  t^{- s} ,
\end{equation*}
for $s < 1$. Since the operator $\tilde B_j$ is of the form $\tilde B_j = \tilde B + M_j$, where $M_j$ is a bounded operator, it then follows that
\begin{equation*}
\big \| \langle \tilde B \rangle^{- s} e^{-i t H} \tilde \chi_j (H) \langle \tilde B \rangle^{-s} \big \| \lesssim t^{- s} ,
\end{equation*}
and hence, using again that $P_\Omega \lan \tilde B \ran = P_\Omega$, we obtain
\begin{align}\label{eq:unif_dist_2}
& \| P_\Omega e^{-itH} \tilde \chi_j(H) u \| = \| \lan \tilde B \ran^{-1} e^{-itH} \tilde \chi_j(H) u \| \lesssim  t^{-s} \|  \lan \tilde B \ran u \|.
\end{align}
Equations \eqref{eq:decom-supp-f}, \eqref{eq:unif_dist_1} and \eqref{eq:unif_dist_2} give \eqref{eq:unif_dist}.
\end{proof}

 Together with $\varphi_1P_\Omega=P_\Omega$, the estimate \eqref{eq:unif_dist} gives
\begin{align}\label{varphi1P0}
\langle \varphi_1 P_\Omega \varphi_1  \rangle_{\psi_t} =\langle  P_\Omega  \rangle_{\psi_t} \lesssim t^{-2s} \| \langle \tilde B \rangle \psi_0 \|^2 \lesssim t^{-2s} \|  \psi_0 \|^2_{ \tilde B^2 }.
\end{align}
Combining this with $N \ge \one - P_\Omega$ and \eqref{M-est1}, we obtain
\begin{align}\label{M-est12}
\langle M  \rangle_{\psi_t} \le& -\frac{1}{c t^{1-\rho}}\langle \varphi_1 [ 1 - t^{-1}B_\e -C\lan g\ran]\varphi_1+c\rho \varphi \rangle_{\psi_t} \notag\\
&+\frac{C}{ t^{1-\rho}}(\e t^{\nu(-1)}\|  \psi_0\|_{-1}^2+t^{-1 + \nu(0)}\e^{-1} \| \psi_0\|_0^2 + t^{-2s} \|  \psi_0 \|^2_{ \tilde B^2 }Ê).
\end{align}

Now, using the definition $\varphi \big( \frac{ B_\e}{ct} \big):= \big(\frac{B_\e }{c t}-2 \big) \chi_{B_\e \leq c t}$, we compute
\begin{align}\label{comp1}
 \frac{ B_\e}{ c t}\varphi'+ \rho  (-\varphi)&=  \frac{ B_\e}{ c t}(\chi +(\frac{ B_\e}{ c t}-2)\chi')- \rho  (\frac{ B_\e}{ c t}-2)\chi \notag\\
 &=  ((1-\rho)\frac{ B_\e}{ c t}+2\rho)\chi + \frac{ B_\e}{ c t} (\frac{ B_\e}{ c t}-2)\chi'.
\end{align}
Next, using that $\frac{ B_\e}{ct} \chi \le \chi$,  $\frac{ B_\e}{ct} ( \frac{ B_\e}{ c t} - 2 ) \chi'  \le ( \frac{ B_\e}{ c t} - 2 )\chi'$, we find furthermore
\begin{equation}\label{comp2}
 \frac{ B_\e}{ c t}\varphi'+ \rho  (-\varphi) \le (1+\rho)\chi + ( \frac{ B_\e}{ c t} - 2 ) \chi' = \rho\chi +\varphi' \le (1+\rho)\varphi'.
\end{equation}
This, together with \eqref{M-est12}, with $\varphi_1^2=\varphi' $, gives
\begin{align}\label{M-est2}
\lan M  \ran_{\psi_t} \leq& -\big[\frac{\sigma}{c}-1-\rho\big]\frac{1}{t^{1-\rho}}
\lan \varphi'  \ran_{\psi_t} \notag\\
& + \frac{C}{ t^{1-\rho}}(\e t^{\nu(-1)}\|  \psi_0\|_{-1}^2+t^{-1+\nu(0)}\e^{-1} \| \psi_0\|_0^2 + t^{-2s} \|  \psi_0 \|^2_{\d \Gamma( \lan y \ran)^2 } ),
\end{align}
where $\s:=1-C\lan g\ran $.

Next,  we show that the remainder, $R$, in \eqref{DPhi-deco} is bounded as
\begin{equation}\label{R-est}
\|(1+ \eta^{-2} + N)^{-1/2} R (1+ \eta^{-2} + N)^{-1/2}\|\lesssim t^{-2} \epsilon^{-1} .
\end{equation}
Indeed, proceeding as in the proof of Lemma \ref{lem:rem1-est}, using the Helffer-Sj{\"o}strand formula, one verifies that
\begin{align}
& \|(1+ \eta^{-2} + N)^{-1/2} R (1+ \eta^{-2} + N)^{-1/2}\| \notag \\
&\lesssim t^{-2}\| (1+ \eta^{-2} + N)^{-1/2} B_2  (1+ \eta^{-2} + N)^{-1/2}\| , \label{R-est1}
\end{align}
where $B_2:=[ B_\e, [B_\e , H]]$. Now, an elementary computation (see \eqref{DdG-rel}) gives  $B_2=\d \Gamma (\e \theta_\e \omega_\e^{-2})  + I(b_\e^2 g).$ Using  $\e \theta_\e \omega_\e^{-2}  \leq \epsilon^{-1}$ and $\|I( \eta b_\e^2 g) (1+N)^{-1/2} \| \lesssim \| \eta b_\e^2 g \| \lesssim \e^{-1}$ since $\mu > -1/2$, we obtain
\begin{equation}\label{B2-est}
\| (1+ \eta^{-2} + N)^{-1/2} B_2  (1+ \eta^{-2} + N)^{-1/2}\| \lesssim t^{-2} \e^{-1} ,
\end{equation}
which together with \eqref{R-est1} implies \eqref{R-est}. Together with Equations \eqref{DPhi-deco} and  \eqref{M-est2} and the fact that $\| \eta^{-2} f(H) \| \lesssim 1$,  this implies
\begin{align}
\lan D \Phi_t  \ran_{\psi_t}  \le & -(\frac{\sigma}{c} -1-\rho)t^{-1+\rho}\lan \varphi' \ran_{\psi_t} \notag\\
& + C \big ( \e t^{\nu(-1)+\rho-1}\|  \psi_0\|_{-1}^2+  t^{-2+\nu(0)+\rho}\e^{-1} \| \psi_0\|_0^2 + t^{-1+\rho-2s} \|  \psi_0 \|^2_{ \tilde B^2 } \big ). \label{DPhi-est}
\end{align}

Thus, choosing $s$ such that $2 s - \rho >0$, \eqref{DPhi-est}, together with the observation $\Phi_t \geq  t^{\rho} \chi_{B_\e \leq c t}$, the conditions $\frac{\sigma}{c} -1-\rho>0$, $\rho<1 \le 2 - \nu(0)$, the trivial inequalities $\| \psi_0 \|^2_0 \le \| \psi_0 \|^2_{ \d \Gamma( \lan y \ran ) }$, $\| \psi_0 \|^2_{ \tilde B^2 } \lesssim \| \psi_0 \|^2_{ \d \Gamma( \lan y \ran )^2 }$, and Hardy's inequality $\| \psi_0 \|_{-1}^2 \lesssim \| \psi_0 \|^2_{\d \Gamma ( \langle y \rangle)}$  implies that
\begin{align*}
 t^{\rho} \langle\chi\rangle_{\psi_t} &\le \langle \Phi_t\rangle_{\psi_t} = \langle \Phi_t\rangle_{\psi_t} |_{t=0} + \int\limits_0^t \langle D \Phi_s \rangle_{\psi_s} ds \notag\\
&\le \langle - B_{\e} \chi_{B_{\e} \leq 0} \rangle_{\psi_0} +C( \e^{-1}+\e t^{\rho + \nu(-1)} + 1)  \| \psi_0 \|_{\d \Gamma ( \lan y \ran )^2 }^2.
\end{align*}
Using $\langle - B_{\e} \chi_{B_{\e} \leq 0}\rangle_{\psi_0} \lesssim  \| \psi_0\|_{\d \Gamma ( \langle y \rangle)}^2$,  and choosing $\e = t^{-\kappa}$, we find
\begin{equation*}
\langle\chi\rangle_{\psi_t} \le C( t^{-\rho+\kappa}+ t^{\nu(-1) - \kappa} + t^{-\rho} )  \| \psi_0 \|_{\d \Gamma ( \langle y \rangle)^2 }^2,
\end{equation*}
which in turn gives \eqref{B-prop-est}.
\end{proof}

\begin{proof}[Proof of Theorem \ref{thm:mve2}]
Since $N:=\d \Gamma (1)$ and $B_\e:=\d \Gamma (b_\e)$, commute we have
\begin{align}
 \Gamma (\chi_{b_\e \leq c' t^\alpha}) &\leq \chi_{B_\e \leq
c' Nt^\alpha}= \chi_{B_\e \leq c'Nt^\alpha}
(\chi_{N \leq c'' t^\gamma} + \chi_{N \geq c'' t^\gamma})\notag \\
 &\leq  \chi_{B_\e \leq c t^\nu} + \chi_{N\geq c'' t^\gamma},
\label{Gchi-bnd}\end{align}
where $\nu := \alpha +\gamma$ and $c:=c'c''$. We choose $c'' \ll 1/c'$, so that $0 < c \ll 1$. Next, we have
\begin{align*}
\| \chi_{N\geq c'' t^\gamma} \psi_t \|\ &\leq
(c'')^{-\frac{\gamma}{2}} t^{-\frac{\gamma}{2}}\ \| \chi_{N
\geq c'' t^\gamma} N^{\frac{1}{2}} \psi_t \|  \notag\\
& \leq
(c'')^{-\frac{\gamma}{2}} t^{-\frac{\gamma}{2}}\
\| N^{\frac{1}{2}} \psi_t \|,
\end{align*}
which, together with \eqref{dGk-bnd} (with $\rho=0$), implies
\begin{equation}\label{chiN-bnd'}
\| \chi_{N\geq c'' t^\gamma} \psi_t \|\ \lesssim t^{-\frac{\gamma}{2}+\frac{\nu(0)}{2}}\
\|  \psi_0 \|_0 .
\end{equation}
The inequality \eqref{Gchi-bnd} with $ \nu=1$, Proposition \ref{prop:B-prop-est-t} and the inequality \eqref{chiN-bnd'} (with $\g=1-\alpha$)  imply the estimate \eqref{mve2}.
\end{proof}

\bigskip

\section{ Proof of Theorem \ref{thm:ac}}\label{sec:pf-ac}

\subsection{Partition of unity}\label{sec:pu}

 We begin with a construction of a partition of unity which separates photons close to the particle system from those departing it. Following \cite{DerGer2, FrGrSchl2} (cf. the many-body scattering construction), it is  defined by first constructing a partition of unity $(j_0, j_\infty),\ j_0^2+ j_\infty^2= \one,$ on the one-photon space,  $\fh=L^2(\R^3)$, with $j_0$ localizing a photon to  a region near the particle system (the origin) and $j_\infty$ near infinity, and then associating with it the map $j:\fh \ra \fh \oplus \fh$, given by $j: h \ra j_0h \oplus j_\infty h$. After that we lift the map $j$ to  the Fock space $\cF:=\G(\fh)$ by using  ${\Gamma} (j): \G(\fh)\ra \G(\fh \oplus \fh)$ (defined  in \eqref{def:Gamma}).
Next, we consider the adjoint map  $j^*: h_0  \oplus h_\infty  \ra j_0^* h_0 + j_\infty^* h_\infty$, which we also lift  to  the Fock space $\cF:=\G(\fh)$ by using ${\Gamma} (j^*): \G(\fh \oplus \fh)\ra \G(\fh)$.
By definition, the operator ${\Gamma} (j)$ has the following properties
\begin{align}\label{Gamma-prop1}
{\Gamma} ( j)^*={\Gamma} (j^*),\ \quad {\Gamma}(\tilde j){\Gamma} (j)={\Gamma} (\tilde j j).
\end{align}
 Since  $j^*j=j_0^2+ j_\infty^2= \one$,
this implies the relation $ {\Gamma}( j)^*{\Gamma} (j)=\one $, which is what we mean by a partition of unity of the Fock space $\cF:=\G(\fh)$.

We refine this construction further by defining the unitary map $U: \G(\fh \oplus \fh)\ra \G(\fh)\otimes \G(\fh),$ through the relations
\begin{equation}\label{U}
U\Om = \Om \otimes \Om ,\ \quad Ua^*(h)=[a^*(h_1)\otimes \one  +\one \otimes a^*(h_2)]U,
\end{equation}
for any $h = (h_1 , h_2) \in \mathfrak{h} \oplus \mathfrak{h}$,
 and introducing the operators
\begin{equation}\label{checkGamma}
\check{\Gamma} (j):=U {\Gamma} (j):  \G(\fh)\ra \G(\fh)\otimes \G(\fh).
\end{equation}
We lift $\Gamma (j)$, as well as $\check{\Gamma} (j)$, from  the Fock space $\cF:=\G(\fh)$ to the full state space $\cH:=\chp\otimes \cF$, so that e.g. $\check{\Gamma} (j):  \cH\ra \cH\otimes \G(\fh)$. Now, the partition of unity relation on $\cH$ becomes $ \check{\Gamma}( j)^*\check{\Gamma} (j)=\one $ (in particular, $ \check{\Gamma}( j)$ is an isometry).

Finally, we specify $j_0$ to be the operator $\chi_{b_\e \leq c t^\al} $, with $b_\e$ defined in the introduction, and $j_\infty$ is defined by $j_0^2 + j_\infty^2= \one$ and is of the form $\chi_{b_\e \geq c t^\al} $, where $\e:=t^{-\kappa}$, and $\al$ and $\kappa$ satisfy $1 - \mu / ( 6 + 3 \mu ) < \alpha < 1$  and $ 1+\nu( - 1 )-\al  <\kappa < \frac{1}{2}( 5\al -3 )$.

\medskip

\subsection{Deift-Simon wave operators}\label{sec:DSOs}

 We define the auxiliary space $\hat\cH:=\cH\otimes \cF$,  which will serve as our repository of asymptotic dynamics, which is governed by the hamiltonian $\hat{H}:=H\otimes \one +\one\otimes H_f$ on  $\hat\cH$.
With the partition of unity  $\check{\Gamma} (j)$, we associate the Deift-Simon wave operators,
\begin{equation} \label{W-def}
W_\pm := \slim_{t\to\infty} W(t),\ \quad \mbox{where}\ \quad  W(t) := e^{i\hat{H} t}\check{\Gamma} (j) e^{-i Ht},
\end{equation}
which map the original dynamics, $e^{-i Ht}$, into auxiliary one, $e^{-i\hat Ht}$ (to be further refined later). Our goal is  to prove

\begin{theorem}\label{thm:MVEraDS}
Assume \eqref{g-est} with $\mu > 0$, \eqref{eta-bnd} and \eqref{npb-unif}. Then the Deift-Simon wave operators exist on $\Ran  E_{(-\infty, \Sigma)}(H)$ and satisfy
\begin{equation}\label{W+Pgs}
W_+P_{\mathrm{gs}}=P_{\mathrm{gs}} ,
\end{equation}
and, for any smooth, bounded function $f$,
\begin{equation}\label{W+intertw}
W_+f({H})=f (\hat{H})W_+.
\end{equation}
\end{theorem}
\begin{proof}
We want to show that the family $W(t):=e^{i\hat{H} t} \check{\Gamma} (j) e^{-iHt}$ form a strong Cauchy sequence as $t\ra\infty$. To this end, we define ${\chi}_{m}:={\chi}_{\hat N\le m}$ and $\overline{\chi}_{m}:={\chi}_{\hat N\ge m}$, where $\hat N:=N\otimes\one +\one \otimes N$, the number operator on $\hat \cH$, so that ${\chi}_{m} + \overline{\chi}_{m}= \one$. First, we show that, for any $\psi_0\in D(N^{\frac12})$,
\begin{align}\label{barchik-bnd}
 \sup_{t}\|\overline{\chi}_{m} W(t)\psi_0\|\lesssim m^{-\frac12} \| \psi_0 \|_{N} .
\end{align}
Indeed, by the assumption \eqref{npb-unif},
\begin{align}\label{eq:Z1}
\|\hat N^{\frac12}e^{i\hat{H} t} \check{\Gamma} (j) e^{-iHs}\psi_0\| \lesssim  \| \hat N^{\frac12} \check{\Gamma} (j) e^{-iHs}\psi_0 \| + \| \check{\Gamma} (j) e^{-iHs}\psi_0 \|.
\end{align}
The boundedness of $\check{\Gamma} (j)$ implies $\| \check{\Gamma} (j) e^{-iHt}\psi_0 \| \le \| \psi_0 \| \le \| \psi_0 \|_N$. Moreover, we claim that
\begin{align}\label{eq:C(4)0}
\check{\Gamma}(j) N = \hat N \check{\Gamma}(j),
\end{align}
Indeed, a straightforward computation gives ${\Gamma}(j) \d \Gamma( c ) =  \d \Gamma( \underline c ){\Gamma}(j) + \d {\Gamma}( j , j c - \underline c j )$, where $\underline c = \diag( c , c ):\ \mathfrak{h} \oplus \mathfrak{h} \to \mathfrak{h} \oplus \mathfrak{h}$ and
  \begin{equation}\label{dG'}
\d \Gamma( a, c )\vert_{  {\otimes_s^n}\fh} = \sum_{j=1}^{n} \underbrace{a \otimes \cdots \otimes a }_{j-1} \otimes c \otimes \underbrace{a \otimes \cdots \otimes a }_{n - j}.
\end{equation}
It follows from this relation and the equalities $U\d \Gamma( \underline c )=( \d \Gamma( c ) \otimes \one + \one \otimes \d \Gamma( c ) ) U$ that (\cite{DerGer2, FrGrSchl2})
\begin{align}\label{eq:C(3)0}
\check{\Gamma}(j) \d \Gamma( c ) &= ( \d \Gamma( c ) \otimes \one + \one \otimes \d \Gamma( c ) ) \check{\Gamma}(j) + \d \check{\Gamma}( j , j c - \underline c j ) ,
\end{align}
where and $\d \check{\Gamma}( a , c ) := U\d {\Gamma}( a , c )$. For $c=\one$, the latter relation gives \eqref{eq:C(4)0}. Equation \eqref{eq:C(4)0} implies $\hat N^{\frac12} \check{\Gamma} (j) = \check{\Gamma} (j) N^{\frac12}$, and this relation, boundedness of $\check{\Gamma} (j)$ and  the assumption \eqref{npb-unif} give
\begin{align*}
\|\hat N^{\frac12}  \check{\Gamma} (j) e^{-iHs}\psi_0\|= \|\check{\Gamma} (j) N^{\frac12} e^{-iHs}\psi_0\|  \lesssim   \| \psi_0\|_N ,
 \end{align*}
and therefore, by \eqref{eq:Z1}, $\|\hat N^{\frac12}e^{i\hat{H} t} \check{\Gamma} (j) e^{-iHs}\psi_0\| \lesssim \| \psi_0\|_N$. Since this is true uniformly in $t,s$, it implies $\| \hat N^{\frac12}W(t)\psi_0\| \lesssim \| \psi_0\|_N $, which yields \eqref{barchik-bnd}. Equation \eqref{barchik-bnd} implies that
\begin{align}\label{Wbark-bnd}
\sup_{t, t'}\|\overline{\chi}_{m}(W(t')- W(t))\psi_0\|\lesssim  m^{-\frac12}.
\end{align}

Now we show that, for any $m>0$ and for any $\psi_0\in D(N^{\frac12})\cap \Ran E_{(-\infty, \Sigma)}(H)$,
\begin{align}\label{Wk-conv}
\|{\chi}_{m}(W(t')- W(t))\psi_0\|\ra 0 ,
\end{align}
 as $t, t' \ra \infty$. This together with \eqref{Wbark-bnd} implies that $W(t)$ form a strong Cauchy sequence. Lemma \ref{lm:f(H)W}, proven below,
 implies that, in order to show \eqref{Wk-conv}, it suffices to prove
\begin{align}\label{Wk-conv'}
\|{\chi}_{m}f(\hat H)(W(t')- W(t))\psi_0\|\ra 0,
\end{align}
to which we now proceed. We write
\begin{align}\label{Wt-repr}
(W(t')- W(t))\psi_0=\int_t^{t'}ds\partial_s W(s)\psi_0
\end{align}
and compute $\partial_t W (t) = e^{i\hat{H} t} G e^{-iHt}$, where $G:= i (\hat{H} \check{\Gamma} (j) - \check{\Gamma}(j)  H)+ \p_t \check{\Gamma}(j)$. We write $ G= G_0 + G_1,$ where
\begin{equation*}
G_0:= i (\hat{H}_f \check{\Gamma} (j) - \check{\Gamma}(j)  H_f)+ \p_t \check{\Gamma}(j)
\end{equation*}
 and
\begin{equation}\label{G1}
G_1 := i( I ( g) \otimes \bfone ) \check{\Gamma} (j) - \check{\Gamma} (j) I ( g).
\end{equation}

We consider $G_0$.  Using $(H_{p} \otimes \one \otimes \one)(\one \otimes \check{\Gamma} (j)) = (\one \otimes \check{\Gamma} (j)) (H_{p} \otimes \one)$ and using the notation
 $\underline d j :=i(\underline{\omega} j - j \omega) + \partial_t j$, with $\underline{\omega} = \mathrm{diag} ( \omega , \omega )$,   and \eqref{eq:C(3)0},
we compute readily
\begin{equation}\label{G0}
 G_0= U  \d {\Gamma} (\underline j, \underline dj ) = \d \check{\Gamma} (\underline j, \underline dj).
 \end{equation}
Write $j' = ( j'_0 , j'_\infty )$, where $j'_0, j'_\infty$ are the derivatives of $j_0, j_\infty$ as functions of $v =\frac{b_\e}{c t^{\al}}$. We first find a convenient decomposition of $\underline dj$. Using  $\underline d j f = (dj_0 f, dj_\infty f)$, with $d c_t =i[\omega, c_t ] + \partial_t c_t$, \eqref{dv} and Corollary \ref{cor:rem1-est}, we compute
\begin{align}\label{udj}
\underline dj &=  (j'_0, j'_\infty) (\frac{\theta_\e}{c t^{\al}} - \frac{\alpha b_\e}{c t^{\al+1}}) + \mathcal{O} (t^{-2\al+\kappa}).
\end{align}
We insert the maximal velocity partition of unity  $\chi_{(\frac{|y|}{\bar c t})^2 \le 1}+\chi_{(\frac{|y|}{\bar c t})^2 \ge 1}=\one$, with $\bar c > 1$, into this formula and use the notation $\chi\equiv \chi_{(\frac{|y|}{\bar c t})^2 \le 1}$ and the relation $\frac{ b_\e}{c t^{\al}} j'_\# = \CO(1) j'_\#$, valid due to the localization of $j'_\#$, to obtain
\begin{align}\label{udj2}
&\underline dj = \frac{1}{c t^{\al}}\theta_\e^{1/2}\chi ( j'_0 , j'_\infty )\chi
\theta_\e^{1/2}+\mathrm{rem}_t ,\\
&\mathrm{rem}_t = \CO(t^{-1}) \chi( j'_0 , j'_\infty )\chi  + \mathcal{O} (t^{-2\al-\kappa}) + \CO( t^{-\al}) \chi_{(\frac{|y|}{\bar c t})^2 \ge 1}.
\end{align}
These relations give
\begin{align}\label{tildeG0-deco}
G_0=G_0'+ \mathrm{Rem}_{t},
\end{align}
where $G'_0 := \frac{1}{ct^\alpha} U \d \Gamma ( j ,  \underline c_t )$, with  $\underline c_t =(c_0, c_\infty):= ( \theta_\e^{1/2}\chi  j'_0 \chi  \theta_\e^{1/2},  \theta_\e^{1/2}\chi j'_\infty \chi  \theta_\e^{1/2})$, and
\begin{equation*}
\mathrm{Rem}_t := G_0 - G'_0 = U \d \Gamma ( j , \mathrm{rem}_t ).
\end{equation*}

Next, we write $A := \sup_{\|\hat\phi_0\| = 1}| \int_t^{t'}ds\lan \hat\phi_s,  G_0 \psi_{s}\ran|$, where $\hat\phi_s :=e^{-i\hat{H} s} f(\hat H) {\chi}_{m} \hat \phi_0$. By \eqref{checkG-ineq'} of Appendix \ref{sec:tech2}, $G'_0$ satisfies
\begin{align}\label{G0'-est}
|\lan \hat\phi,  G'_0 \psi\ran | &\le \frac{1}{ct^\alpha} \big (  \| \d \Gamma( |c_0| )^{\frac12} \otimes \one \hat \phi \| \, \| \d {\Gamma} ( |c_0| )^{\frac12} \psi \| \notag \\
&\quad+ \| \one \otimes \d \Gamma( |c_\infty| )^{\frac12} \hat \phi \| \, \| \d {\Gamma} ( |c_\infty| )^{\frac12} \psi \| \big ).
\end{align}
By the Cauchy-Schwarz inequality, \eqref{G0'-est} implies
\begin{align*}
\int_t^{t'}ds | \lan \hat\phi_s,  G'_0 \psi_{s}\ran | & \le \Big ( \int_t^{t'} ds
\| \d \Gamma(|c_0|  )^{\frac12} \otimes \one \hat\phi_s\|^2 \Big )^{\frac12} \Big ( \int_t^{t'}ds \| \d {\Gamma} (|c_0|  )^{\frac12} \psi_{s}\|^2 \Big )^{\frac12} \notag \\
& + \Big ( \int_t^{t'} ds\| \one \otimes \d \Gamma( |c_\infty| )^{\frac12} \hat\phi_s\|^2 \Big )^{\frac12} \Big ( \int_t^{t'}ds\| \d {\Gamma} (|c_\infty| )^{\frac12} \psi_{s}\|^2 \Big )^{\frac12}.
\end{align*}
Since $|c_0|$, $|c_\infty|$ are of the form $\theta_\e^{1/2} \chi \chi_{b_\e= c t^\al} \chi \theta_\e^{1/2}$, the minimal velocity estimate \eqref{mve1'} implies
\begin{equation*}
\int_1^{\infty} ds \, s^{-\alpha} \| \widehat{\d\Gamma}_\# ( |c| )^{\frac12} \hat\phi_s\|^2 \lesssim   \| {\chi}_{m}\hat\phi_0 \|_0^2 \lesssim m \| \hat\phi_0 \|^2,
\end{equation*}
where $\widehat{\d\Gamma}_\# ( |c| )^{\frac12}$ stands for $\d \Gamma( |c_0| )^{\frac12} \otimes \one$ or $\one \otimes \d \Gamma( |c_\infty| )^{\frac12}$, and
\begin{equation*}
\int_1^\infty ds\, s^{-\al} \| \d {\Gamma}_\# (  |c| )^{\frac12} \psi_{s}\|^2 \lesssim   \| \psi_0 \|_{0}^2 ,
\end{equation*}
with $\d\Gamma_\# ( |c| )^{\frac12} = \d \Gamma( |c_0| )^{\frac12}$ or $\d \Gamma( |c_\infty| )^{\frac12}$. The last three relations give
\begin{align}\label{eq:G'0tt'}
\sup_{\|\hat\phi_0\| = 1} | \int_t^{t'}ds\lan \hat\phi_s,  G'_0 \psi_{s}\ran| \ra 0,\ \quad t, t' \ra \infty.
\end{align}

Likewise, applying \eqref{eq:est-dG(j,c1c2)} of Appendix \ref{sec:tech2} first with $c_1=c_2=1$, next with $c_1=1$ and $c_2=\chi_{(\frac{|y|}{\bar c t})^2 \ge 1}$, and then applying \eqref{checkG-ineq'} with $c_0 = \chi j_0\chi $ and $c_\infty = \chi j_\infty \chi$, we see that $\mathrm{Rem}_t$ satisfies
 \begin{equation}\label{Rem-est'}
|Ê\langle \hat \phi , \mathrm{Rem}_{t} \psi \ran | \lesssim \| \hat N^{\frac12} \hat \phi \| \Big ( t^{-2\alpha + \kappa}  \| N^{\frac12} \psi \| + t^{-1}  \| \d\Gamma(\chi j_\infty'\chi )^{\frac12} \psi \| +t^{-\al}\|\d\G( \chi^2_{(\frac{|y|}{\bar c t})^2 \ge 1})^{\frac12}\psi \| \Big ).
\end{equation}
Now, using \eqref{Rem-est'} and the Cauchy-Schwarz inequality, we obtain
\begin{align}\label{Rem-est''}
&| \int_t^{t'}ds\lan \hat\phi_s, \mathrm{Rem}_s \psi_{s}\ran| \le \Big ( \int_t^{t'} ds \, s^{-\tau} \| \hat N^{\frac12} \hat \phi_s \|^2 \Big )^{\frac12} \Big \{ \Big ( \int_t^{t'} ds \, s^{-2(2\alpha - \kappa)+\tau} \| N^{\frac12}  \psi_s \|^2 \Big )^{\frac12} \notag \\
&\qquad +  \Big ( \int_t^{t'} ds \, s^{-2  + \tau} \| \d\Gamma(\chi  j_\infty'\chi )^{\frac12} \psi_s \|^2 \Big )^{\frac12}+  \Big ( \int_t^{t'} ds \, s^{-2\al  + \tau} \| \d\G(\chi^2_{(\frac{|y|}{\bar c t})^2 \ge 1})^{\frac12} \psi_s \|^2 \Big )^{\frac12} \Big \}.
\end{align}
Let $ \tau > 1$ and $\alpha = 2 - \tau $. Then by the estimate \eqref{mve1},
\begin{align*}
& \int_1^{\infty} ds \, s^{-2  + \tau} \| \d\Gamma(\chi  j_\infty'\chi )^{\frac12} \psi_s \|^2  \lesssim  \| \psi_0 \|_{-1}^2 ,
\end{align*}
provided $\alpha < \frac{1}{\bar c}$, and by the maximal velocity estimate \eqref{maxvel-est},
\begin{align*}
& \int_1^{\infty} ds \, s^{-2\al  + \tau} \| \d\Gamma(\chi^2_{(\frac{|y|}{\bar c t})^2 \ge 1})^{\frac12}  \psi_s \|^2  \lesssim  \| \psi_0 \|_{\d\Gamma(\lan y\ran)},
\end{align*}
provided that $\alpha > 1 - 2\gamma / 3$, where, recall, $\gamma < \frac{ \mu }{ 2 } \min ( \frac{ \bar{\mathrm{c}} - 1 }{ 3 \bar{\mathrm{c}} - 1 } , \frac{1}{2+\mu} )$. One verifies that $\bar c>1$ can be chosen such that the two conditions above are satisfied.  Moreover, by Assumption \eqref{npb-unif},
\begin{align*}
& \int_1^{\infty} ds \, s^{-2(2\alpha - \kappa)+\tau} \| N^{\frac12}  \psi_s \|^2 \lesssim  \| \psi_0 \|_{N},
\end{align*}
provided that $5 \alpha > 3 + 2 \kappa$. This and the fact that, by Assumption \eqref{npb-unif}, the first integral on the r.h.s. of \eqref{Rem-est''} converge yield
\begin{align}\label{eq:Remtt'}
\sup_{\|\hat\phi_0\| = 1} | \int_t^{t'}ds\lan \hat\phi_s, \mathrm{Rem}_s \psi_{s}\ran| \to 0, \ \quad t, t' \ra \infty .
\end{align}
Equations \eqref{eq:G'0tt'} and \eqref{eq:Remtt'} imply
\begin{equation}\label{tildeG0-contr}
A = \|\int_t^{t'}ds {\chi}_{m} f(\hat H)e^{i\hat{H} s} G_0 \psi_{s}\| \ra 0,\ \quad t, t' \ra \infty.
\end{equation}

Now we turn to $ G_1$. We use the definition $\check{\Gamma}(j) := U \Gamma( j )$ to obtain $\check{\Gamma}(j) a^\#(h)=Ua^\#(j h) \Gamma( j )$, then \eqref{U}, and then $j_0^* j_0 + j_\infty^* j_\infty = 1 ,$ to derive
\begin{align}\label{eq:C''0}
\check{\Gamma}(j) a^\#(h) &= ( a^\#(j_0 h) \otimes \one + \one \otimes a^\#(j_\infty h) ) \check{\Gamma}(j),
\end{align}
 where $a^\#$ stands for $a$ or $a^*$, which implies
\begin{align}\label{eq:checkG-I(g)}
\check{\Gamma}(j) I( g ) = ( I( j_0 g ) \otimes \one + \one \otimes I( j_\infty g ) ) \check{\Gamma} ( j ).
\end{align}
 The equation \eqref{eq:checkG-I(g)} gives
\begin{align}\label{eq:G1}
G_1 = ( I( (1 - j_0) g ) \otimes \one - \one \otimes I( j_\infty g ) ) \check{\Gamma}(j) .
\end{align}
Due to \cite[Lemma 3.1]{BoFaSig} (see Appendix \ref{sec:tech}, Lemma \ref{lem:chig-est}), we have $\| j_\infty  g \|_{L^2}\lesssim t^{-\lam}$, $\| (1 - j_0) g \|_{L^2}\lesssim t^{-\lam}$ with $\lam<(\mu+\frac32)\al$. This, \eqref{I-est} (with $\delta = 0$), and $\hat N^{\frac12} \check{\Gamma} (j) = \check{\Gamma} (j) N^{\frac12}$ imply that
\begin{equation}\label{eq:G1_2}
 \| f( \hat H ) G_1 (N+1)^{-\frac12} \| \lesssim t^{-(\mu+\frac32)\alpha} .
 \end{equation}
This together with Assumption \eqref{npb-unif} implies that $\| f( \hat H ) G_1 \psi_t \| \lesssim t^{-(\mu+\frac32)\alpha} \| \psi_0 \|_0$, and hence
\begin{align*}
\| \int_t^{t'} ds f( \hat H ) e^{i \hat H s} G_1 \psi_s \| \to 0 , \qquad t,t' \to \infty,
\end{align*}
provided that $\alpha > ( \mu + \frac{3}{2} )^{-1}$. This together with \eqref{tildeG0-contr} gives \eqref{Wk-conv'}, and therefore \eqref{Wk-conv},  which, as was mentioned above, together with \eqref{Wbark-bnd} shows that $W(t)$ is a Cauchy sequence as $t\ra\infty$. This implies the existence of $W_+$.

Finally, the proofs of \eqref{W+Pgs} and \eqref{W+intertw} are standard. We present the second one.  By \eqref{W-def}, we have $W_\pm e^{i\hat{H} s}= \slim e^{i\hat{H} t}\check{\Gamma} (j) e^{-i H(t+s)} =\slim e^{i\hat{H} (t'-s)}\check{\Gamma} (j) e^{-i Ht'}= e^{i\hat{H} s}W_+$, which implies \eqref{W+intertw}. \end{proof}

Now we establish the following lemma used in the proof of Theorem \ref{thm:MVEraDS}.
\begin{lemma}\label{lm:f(H)W}
Under the conditions of Theorem \ref{thm:MVEraDS}, for any $f \in \mathrm{C}_0^\infty( \Delta )$, $\Delta \subset (E_{\mathrm{gs}} , \Sigma )$, and $\psi_0 \in \Ran  E_\Delta (H) \cap D( N^{\frac12} )$,
\begin{align}
\| (\hat N + 1)^{-\frac12} ( \check{\Gamma}(j) f(H)  - f( \hat{H} ) \check{\Gamma}(j) ) \psi_t\| \lesssim t^{-\alpha} \| \psi_0 \|_0. \label{eq:e0}
\end{align}
\end{lemma}
\begin{proof}
We compute, using the Helffer-Sj{\"o}strand formula, $ \check{\Gamma}(j) f(H) \psi_t - f( \hat{H} ) \check{\Gamma}(j) \psi_t=R$, where
\begin{align}
R&:= \frac{1}{\pi} \int \partial_{\bar z} \widetilde{f} (z) ( \hat H - z )^{-1} ( \hat H \check{\Gamma}(j) - \check{\Gamma}(j) H ) ( H - z )^{-1} \psi_t
 \ddre z \ddim z. \label{eq:e2}
\end{align}
We have $\hat H \check{\Gamma}(j) - \check{\Gamma}(j) H = \tilde G_0 -i G_1$, where $\tilde G_0:= U \d {\Gamma} (j, \underline{\omega} j - j \omega)$ and $ G_1:=
 ( I ( g) \otimes \bfone ) \check{\Gamma} (j) - \check{\Gamma} (j) I ( g)$ was defined in \eqref{G1}.

We consider $\tilde G_0$. As in the proof of Theorem \ref{thm:MVEraDS}, we have $\underline{\omega} j - j \omega =  ( [ \omega , j_0 ] , [ \omega , j_\infty ] )$, and, by Corollary \ref{cor:rem1-est},
\begin{align}\label{eq:s1}
[ \omega , j_\# ] = \frac{ \theta_\e }{ c t^\alpha}  j'_\# + r ,
\end{align}
where $j_\#$ stands for $j_0$ or $j_\infty$, $j'_\#$ is the derivative of $j_\#$ as a function of $\frac{b_\e}{c t^{\al}}$, and $r$ satisfies $\|r \| \lesssim  t^{-2 \alpha+\kappa}$. Since $\theta_\e \le 1$ and since $\kappa < \alpha$, we deduce that $[ \omega , j_\# ] = \CO( t^{-\alpha} )$. By \eqref{eq:est-dG(j,c1c2)} of Appendix \ref{sec:tech2}, we then obtain that
\begin{align*}
\| (\hat N+1)^{-\frac12} \tilde G_0 (N+1)^{-\frac12} \| \lesssim t^{-\alpha}.
\end{align*}
Since $H \in C^1(N)$, we have $\|(N+1)^{\frac12} ( H-z )^{-1} (N+1)^{-\frac12} \| \lesssim | \mathrm{Im} \, z |^{-2}$, and likewise $\|(\hat N+1)^{-\frac12} ( \hat H-z )^{-1} (\hat N+1)^{\frac12} \| \lesssim | \mathrm{Im} \, z |^{-2}$. Moreover, by Assumption \eqref{npb-unif}, $\|(N+1)^{\frac12} e^{-iHt} (N+1)^{-\frac12} \| \lesssim 1$, and $\|(\hat N+1)^{-\frac12} e^{i\hat H t} (\hat N+1)^{\frac12} \| \lesssim 1$. The previous estimates imply
\begin{align}\label{eq:s3}
\| ( \hat N + 1 )^{-\frac12} e^{ i \hat H t } ( \hat H - z )^{-1} \tilde G_0 ( H - z )^{-1} \psi_t \| \lesssim t^{-\alpha} | \mathrm{Im} z |^{-4} \| \psi_0 \|_N.
\end{align}
As in \eqref{eq:G1}--\eqref{eq:G1_2}, we have in addition
\begin{align*}
\| (\hat N+1)^{-\frac12} G_1 E_\Delta(H) \| \lesssim t^{-(\mu+\frac32)\alpha} ,
\end{align*}
and hence
\begin{align}\label{eq:e3b}
\| ( \hat N + 1 )^{-\frac12} e^{ i \hat H t } ( \hat H - z )^{-1} G_1 ( H - z )^{-1} \psi_t \| \lesssim t^{-(\mu+\frac32)\alpha} | \mathrm{Im} z |^{-3} \| \psi_0 \|.
\end{align}

From \eqref{eq:e2}, \eqref{eq:s3}, \eqref{eq:e3b} and the properties of the almost analytic extension $\tilde f$, we conclude that \eqref{eq:e0} holds.
\end{proof}

\medskip

\subsection{Scattering map}\label{sec:scat-map}

We define the space $\mathcal{H}_{ \mathrm{fin} } := \cH_p \otimes \cF_{ \mathrm{fin} } \otimes \mathcal{F}_{ \mathrm{fin} }$, where $\mathcal{F}_{ \mathrm{fin} } \equiv \mathcal{F}_{ \mathrm{fin} }( \mathfrak{h} )$ is the subspace of $\cF$ consisting of vectors  $\Psi=(\psi_n)^{\infty}_{n=0}\in \cF$ such that $\psi_n=0$, for all but finitely many $n$, and  the (\textit{scattering}) map $I: \mathcal{H}_{ \mathrm{fin} } \to \mathcal{H}$ as the extension by linearity of the map (see \cite{HuSp,DerGer2, FrGrSchl2})
\begin{align}\label{I-def}
I:  \Phi\otimes \prod_1^n a^*( h_{i} )\Om \ra
 \prod_1^n a^*( h_{i} ) \Phi,
\end{align}
for any $\Phi\in \mathcal{H}_p \otimes \cF_{ \mathrm{fin} }$  and  for any $ h_{1} , \dots h_n \in  \mathfrak{h}$. (Another useful representation of $I$ is $I: \Phi\otimes f \ra \left( \begin{array}{c} p+q \\ p \end{array} \right)^{1/2} \Phi\otimes_s f$,  for any  $\Phi\in \cH_p \otimes (\otimes_s^p \fh)$ and $f \in \otimes_s^q \fh$). As already clear from \eqref{I-def}, the operator $I$ is unbounded. Let
\begin{align}\label{fh0}
\mathfrak{h}_0 := \{ h \in L^2( \R^3 ) , \int dk ( 1 + \omega^{-1} ) |h(k)|^2 < \infty \} .
\end{align}
Properties of the operator $I$ used below are recorded in the following
\begin{lemma}[\cite{DerGer2, FrGrSchl2, Ger}] \label{lem:GI-rel}
For any operator $j : {h} \to j_0{h} \oplus j_\infty{h}$ and $ n \in \mathbb{N},$ the following relations hold
\begin{align}
&\check{\Gamma}(j)^* = I \Gamma( j_0^* ) \otimes \Gamma( j_\infty^* ), \label{checkG-I} \\
   & D(( H + i)^{-n/2})\otimes (\otimes_s^n \mathfrak{h}_0)\subset D(I).\label{DomI}
\end{align}
\end{lemma}
\begin{proof} Since the operators involved act only on the photonic degrees of freedom, we ignore the particle one.
For $g,h \in \mathfrak{h}$, we define embeddings $ i_0{g} := ( g , 0 ) \in \mathfrak{h} \oplus \mathfrak{h}$ and $i_\infty h := (0,h) \in \mathfrak{h} \oplus \mathfrak{h}$. By the definition of $U$ (see \eqref{U}), we have the relations  $U^* a^*(g) \otimes \one = a^*( i_0 g ) U^*$, and $U^* \one \otimes a^*(h) = a^*( i_\infty  h )U^*$. Hence, using in addition $U^* \Omega \otimes \Omega = \Omega$, we obtain
\begin{align*}
 U^*\prod_1^m a^*(g_i)\Omega \otimes \prod_1^n  a^{*}( h_{i} ) \Om & =  \prod_1^m a^*( i_0{g_i} ) \prod_1^n a^{*}( i_\infty  h_{i} ) \Om. \notag
\end{align*}
By the definition of ${\Gamma}(j)$ and the relations $ j^* i_0{g} = j_0^* g $ and $j^*i_\infty  h = j_\infty^* h$, this gives
 \begin{align}\label{eq:C1}
{\Gamma}(j)^* U^*\prod_1^m a^*(g_i)\Omega \otimes \prod_1^n  a^{*}( h_{i} ) \Om & =\prod_1^n a^*( j_\infty^*{g_i} ) \prod_1^m  a^{*}( j_0^* h_{i} )  \Om.
\end{align}
 Now, by the definition of $\check{\Gamma}(j)$ (see \eqref{U}), we have  $\check{\Gamma}(j)^* = {\Gamma}(j)^* U^* $. On the other hand by \eqref{I-def}, the r.h.s. is $I \Gamma( j_0^* ) \otimes \Gamma( j_\infty^* )\prod_1^m a^*(g_i)\Omega \otimes \prod_1^n  a^{*}( h_{i} ) \Om$.
This proves \eqref{checkG-I}.

To prove  \eqref{DomI}, we use the following elementary properties (\cite{FrGrSchl2,Ger}):
\begin{align}\label{lm:higher-order}
\mbox{The operator}\ \quad H_f^n ( H + i)^{-n}\ \quad \mbox{is bounded}\ \quad  \forall n \in \mathbb{N} ,
\end{align}
and, for any $h_1 , \cdots h_n \in \mathfrak{h}_0$, where $\mathfrak{h}_0$ is defined in \eqref{fh0},
\begin{align}\label{multipl-a-bnd}
\| a^*(h_1) \cdots a^*(h_n) (H_f+1)^{-n/2} \| \le C_n \| h_1 \|_\omega \cdots \| h_n \|_\omega ,
\end{align}
where $\| h \|_\omega := \int dk (1+\omega^{-1}) |h(k)|^2$. The previous two estimates and the representation \eqref{I-def} imply that  for any $\Phi \in D(( H + i)^{-n/2})$
 and $h_1, \cdots , h_n \in \mathfrak{h}_0$, we have $\|I \Phi  \otimes \prod_1^n a^*( h_{i} )\Om  \| \le C_n \| h_1 \|_\omega \cdots \| h_n \|_\omega \|( H + i)^{n/2}\Phi \|< \infty.$
This gives the second statement of the lemma.
\end{proof}

\medskip

\subsection{Asymptotic completeness}\label{sec:ACpf}

Recall that  $P_{\mathrm{gs}}$ denotes the orthogonal projection onto the ground state subspace of $H$.  Below, the symbol $ C(\e') o_t( 1 )$ stands for a positive function of $\e$ and $t$ such that $\| C(\e) o_t(1) \| \to 0$ as $t \to \infty$ and we denote by $\chi_{\Om}(\lam)$ a smoothed out characteristic function of a set $\Om$.
In this section we prove the following result.
\begin{theorem}\label{thm:MVEraAC}
Assume the conditions of Theorem \ref{thm:ac} and let $a< \Sigma,\ \Delta =[E_{\mathrm{gs}}, a]\subset \R$. Then, for every $\e'>0$ there is $\phi_{0\e'}$, s.t.
\begin{align}\label{psit-asymp}
\limsup_{t \to \infty } \| \psi_{t} - I (e^{ - i E_{\mathrm{gs}} t} P_{\mathrm{gs}} \otimes e^{ - i H_f t}\chi_{[0,a-E_{\mathrm{gs}} ] }(H_f)) \phi_{0\e'} \| = \mathcal{O}(\e') ,
\end{align}
which implies \eqref{AC}.
\end{theorem}
\begin{proof}
Let $\alpha,\beta,\kappa$ be fixed such that the conditions of Theorems \ref{thm:mve1}, \ref{thm:mve2} and \ref{thm:MVEraDS} hold, with $\alpha = \beta$. Let $(j_0, j_\infty):= (\chi_{b_\e \leq c t^\al} ,  \chi_{b_\e \geq c t^\al})$ be the partition of unity defined in Subsection \ref{sec:pu}. Since $j_0^2+j_\infty^2=1$, the operator $\check{\Gamma}(j)$ is, as mentioned above, an isometry. Using the relation $\Gamma(j)^*\Gamma(j)=\one$, the boundedness of $\check{\Gamma}(j)^*$, and the existence of $W_+$, we obtain
\begin{align}\label{psite1}
\psi_{t} &=   \check{\Gamma} (j)^* e^{-i \hat H t}  e^{i \hat H t} \check{\Gamma}(j) e^{- i Ht} \psi_{0}  =   \check{\Gamma} (j)^* e^{ - i \hat Ht}  \phi_{0} + o_t(1),
\end{align}
where $\phi_{0} := W_+ \psi_{0}$.
Next, using the property $W_+\chi_{\Delta} ({H})=\chi_{\Delta} (\hat{H})W_+$, which gives $W_+ \psi_{0}=\chi_{\Delta} (\hat{H})W_+ \psi_{0}$,  and $\chi_{\Delta} (\hat{H})  = (\chi_{[E_{\mathrm{gs}},a]} (H) \otimes \chi_{[0,a-E_{\mathrm{gs}}]}(H_f)) \chi_{\Delta} (\hat{H})$, and again using $\chi_{\Delta} (\hat{H})W_+ \psi_{0} = W_+ \psi_{0}=\phi_{0}$, we obtain
\begin{align}\label{phi0e-appr1}
\phi_{0} =(\chi_{[E_{\mathrm{gs}},a]} (H) \otimes \chi_{[0, a-E_{\mathrm{gs}}]}(H_f))\phi_{0}.
\end{align}
For all $\e'>0$, there is $\del'=\del'(\e')>0$, such that
\begin{equation}\label{phi0e-appr2}
 \|(\chi_{[E_{\mathrm{gs}},a]} (H) \otimes \one) \phi_{0} - (\chi_{\Delta_{\e'} }(H) \otimes \one ) \phi_{0} - (P_{\mathrm{gs}} \otimes \one) \phi_{0}\| \le \e',\
\end{equation}
$\mbox{with}\ \Delta_{\e'} = [ E_{\mathrm{gs}} + \del' , a ] $. The last two relations give
\begin{align}\label{phi0e-appr3}
\phi_{0}=((\chi_{\Delta_{\e'} }(H)+P_{\textrm{gs}}) \otimes \chi_{[0,a-E_{\mathrm{gs}}]}(H_f))\phi_{0} +  \mathcal{O}( \e' ).
\end{align}

Now, let $\phi_{0,\e'} \in \cF_{ \mathrm{fin} }( D( \d\Gamma( \lan y \ran ) ) ) \otimes \cF_{ \mathrm{fin} }( \mathfrak{h}_0 )$
be such that $ \| \phi_{0} - \phi_{0\e'} \| \le \e'$. (We require that the `first components' of $\phi_{0\e'}$ are in $D( \d\Gamma( \lan y \ran ) )$ in order to apply the minimal velocity estimate below, and that the `second components' are in $\cF_{ \mathrm{fin} }( \mathfrak{h}_0 )$ in order that $( P_{\mathrm{gs}} \otimes \one ) \phi_{0\e'}$ is in $D(I)$). This together with \eqref{psite1} and \eqref{phi0e-appr3} gives
\begin{align}\label{psite2}
\psi_{t} & =   \check{\Gamma} (j)^* e^{ - i \hat H t }((\chi_{\Delta_{\e'} }(H)+P_{\textrm{gs}}) \otimes \chi_{[0, a-E_{\mathrm{gs}}]}(H_f)) \phi_{0\e'}  + \mathcal{O}( \e' )  + o_t(1).
\end{align}

Furthermore, let $(\tilde j_0 , \tilde j_\infty )$ be of the form $\tilde j_0 = \tilde \chi_{b_\e \le c t^\alpha }$, $\tilde j_\infty = \tilde{\chi}_{b_\e \ge c t^\alpha}$ where $\tilde \chi$, has the same properties as $\chi$, and satisfy $j_0  \tilde j_0 =j_0 ,\ j_\infty \tilde j_\infty =j_\infty$. Then, by \eqref{checkG-I}, the  adjoint $\check{\Gamma}(j)^*$  to the operator $\check{\Gamma}(j)$ can be  represented as
\begin{align}\label{G*'}
\check{\Gamma}(j)^* = \check{\Gamma}(j)^* \big ( \Gamma( \tilde j_0 ) \otimes \Gamma ( \tilde j_\infty )\big).
\end{align}
Using this equation in \eqref{psite2},  together with the relations $e^{ - i \hat Ht}=e^{ - i H t } \otimes   e^{ - i H_f t }$ and $e^{ - i H t}P_{\textrm{gs}}=e^{- i E_{\textrm{gs}}t}P_{\textrm{gs}}$, gives
\begin{align}\label{psite3}
\psi_{t} & =  \check{\Gamma}(j)^* \psi_{t\e'} + A+B+ C + \mathcal{O}(\e') + o_t( 1 ) ,
\end{align}
where
\begin{align}
&\psi_{t\e'} :=(e^{ - i E_{\mathrm{gs}} t} P_{\mathrm{gs}} \otimes e^{ - i H_f t}\chi_{[0, a-E_{\mathrm{gs}}]}(H_f)) \phi_{0\e'}  , \label{Psit}\\
&A  := \check{\Gamma}(j)^* (\Gamma( \tilde j_0 )e^{ - i H t } \chi_{\Delta_{\e'} }(H) \otimes \Gamma( \tilde j_\infty )e^{ - i H_f t}\chi_{[0, a-E_{\mathrm{gs}}]}(H_f)) \phi_{0\e'},
\label{A} \\
& B:= \check{\Gamma}(j)^* ((\Gamma(\tilde j_0 )-\one)e^{- i E_{\textrm{gs}}t}P_{\textrm{gs}} \otimes \Gamma( \tilde j_\infty )e^{ - i H_f t}\chi_{[0, a-E_{\mathrm{gs}}]}(H_f)) \phi_{0\e'}, \label{B}\\
&C:= \check{\Gamma}(j)^* (e^{-i E_{\textrm{gs}}t}P_{\textrm{gs}} \otimes (\Gamma( \tilde j_\infty )-\one)e^{ - i H_f t} \chi_{[0, a-E_{\mathrm{gs}}]}(H_f)) \phi_{0\e'} ,\label{C}
\end{align}
Since $\Gamma(j)^*$ is bounded, the minimal velocity estimate, \eqref{mve2}, gives (here we use that the first components of $\phi_{0\e'}$ are in $D(\d\Gamma( \lan y \ran ) )$)
\begin{align*}
\| A\|\le \big \| ( \Gamma ( \tilde j_0) e^{ - i H t } \chi_{\Delta_{\e'} }(H) \otimes \one ) \phi_{0\e'} \big \| = C(\e') o_t(1) .
\end{align*}

Now we consider the term given by $B$. We begin with
\begin{align}
 \big \|B \|  &\le \big \|   ( \Gamma ( \tilde j_0) - \one ) P_{\textrm{gs}} \big \| . \label{eq:b3}
\end{align}
Since $0\le \tilde j_0\le 1$, we have that $0\le \one - \Gamma ( \tilde j_0)\le \one $. Using this, the relations $\one - \Gamma ( \tilde j_0) \le\d \Gamma( \tilde{\chi}_{b_\e \ge c t^\alpha})$ and $\d \Gamma( \tilde{\chi}_{b_\e \ge c t^\alpha} )\le t^{-2\alpha}  \d \Gamma ( b_\e^{2} )$, we obtain the bound
\begin{align} \label{eq:b4}
\big \| ( \Gamma ( \tilde j_0) - \one ) u \|^2 & \le  \| ( \one - \Gamma ( \tilde j_0) )^{\frac12} u \|^2 \le \| \d \Gamma( \bar{\tilde{\chi}}_{b_\e \ge c t^\alpha} )^{\frac12} u \|^2\notag \\
&\le t^{-2\alpha} \| \d \Gamma ( b_\e^{2} )^{\frac12} u \|^2.
\end{align}
Using the pull-through formula, one verifies that $\d \Gamma ( b_\e^2 )^{\frac12} P_{\mathrm{gs}}$ is bounded and that $\| \d \Gamma ( b_\e^2 )^{\frac12} P_{\mathrm{gs}} \| = \CO( t^\kappa )$ (see Appendix \ref{sec:tech2}, Lemma \ref{lm:domainGS}).
Hence, since $\kappa < \alpha$, the above estimates yield
\begin{align}\label{B-est}
&\big \|B \big \| =  o_t(1).
\end{align}

Next, using $ \Gamma(j_\infty)  e^{ - i H_f t }  =  e^{ - i H_f t } \Gamma( e^{ i \om t} j_\infty e^{ - i \om t} ) $ and $ e^{i\om t} b_\e e^{-i\om t}=b_\e+\theta_\e t$, it is not difficult to verify  (see Appendix \ref{sec:tech2}, Lemma \ref{lem:Gjinfty-est})  that
\begin{align*}
\big \|C \big \|\le \big \| \one \otimes ( \Gamma( e^{i \om t} \tilde j_\infty e^{- i \om t} ) - \one )  \phi_{0\e'} \big \| \to 0 ,
\end{align*}
as $t \to \infty$, and hence we obtain
\begin{align}\label{C-est}
\big \|C \big \| = C(\e') o_t(1).
\end{align}
Inserting the previous estimates into \eqref{psite3} shows that
\begin{align} \label{psite4}
\psi_{t} & = \check{\Gamma}(j)^* \psi_{t\e'} + \mathcal{O}(\e') + C(\e') o_t( 1 ).
\end{align}

Next, we want to pass from $\check{\Gamma}(j)^*$ to $I$ using the formula \eqref{checkG-I}. To this end we use estimates of the type \eqref{B-est} and \eqref{C-est} in order to remove the term $ \Gamma( j_0 ) \otimes \Gamma ( j_\infty) $. Hence, we need to bound $I$, for instance by introducing a cutoff in $N$.  Let $\chi_{m} := \chi_{N\le m}$ and $\bar \chi_{m} := \one - \chi_{m}$ and write $\check{\Gamma}(j)^* \psi_{t\e'} = \chi_m \check{\Gamma}(j)^* \psi_{t\e'}+ \bar \chi_m \check{\Gamma}(j)^* \psi_{t\e'}$. Using that $N^{1/2} \check{\Gamma}(j)^* = \check{\Gamma}(j)^* \hat{N}^{1/2}$, and that $\psi_{t\e'} \in D( \hat N^{1/2} )$ (see Appendix \ref{sec:tech2}, Lemma \ref{lm:domainGS}), we estimate
\begin{align*}
\| \bar \chi_m \check{\Gamma}(j)^* \psi_{t\e'} \| \lesssim m^{-\frac12} \| \hat N^{1/2} \psi_{t\e'} \| = m^{-\frac12} C( \e' ).
\end{align*}
Now, we can use \eqref{checkG-I} to obtain
\begin{align} \label{psite5}
\psi_{t} & = \chi_m I \big ( \Gamma( j_0 ) \otimes \Gamma ( j_\infty )\big)\psi_{t\e'}  + \mathcal{O}(\e') + C(\e') o_t( 1 ) + C(\e') o_m(1).
\end{align}
Using $\|Ê\chi_m I \| \le 2^{m/2}$ together with estimates of the type \eqref{B-est} and \eqref{C-est}, we find (here we need the cutoff $\chi_m$)
\begin{align} \label{psite6}
\psi_{t} & =  \chi_m I \psi_{t\e'} + \mathcal{O}(\e') + C(\e',m) o_t( 1 ) + C(\e')o_m(1).
\end{align}
Since $\phi_{0\e'} \in  \cH  \otimes \cF_{ \mathrm{fin}} ( \mathfrak{h}_0 )$, we can write $\psi_{t\e'}$ as $\psi_{t\e'}= \Phi_{\mathrm{gs}} \otimes f_{t\e'} $, with $f_{t\e'} \in \cF_{ \mathrm{fin} } ( \mathfrak{h}_0 )$, and therefore $\psi_{t\e'}\in D(I)$ (here we need that $f_{\e'}$ is in $\cF_{\mathrm{fin}}( \mathfrak{h}_0 )$). Hence $\chi_m I \psi_{t\e'} = I \psi_{t\e'} + C(\e') o_m(1)$. Combining this with
\eqref{psite6} and remembering \eqref{Psit}, we obtain
\begin{align} \label{psite7}
\psi_{t} & =  I(e^{ - i E_{\mathrm{gs}} t} P_{\mathrm{gs}} \otimes e^{ - i H_f t}\chi_{[0,a-E_{\mathrm{gs}}]}(H_f)) \phi_{0\e'} + \mathcal{O}( \e' ) + C(\e',m) o_t( 1 ) + C(\e')o_m(1).
\end{align}
Letting $t\to \infty$, next $m \to \infty$, the equation \eqref{psit-asymp} follows.
\end{proof}

\noindent \textbf{Remark.}  The reason for $\e'$ in the statement of the theorem is we do not know whether $\Ran ( P_{\textrm{gs}} \otimes 1 )W_+ \psi_{0} \in D(I)$.
Indeed, if the latter were true, then the relations \eqref{psite7},  \eqref{phi0e-appr3} and  $ \| \phi_{0} - \phi_{0\e'} \| \le \e'$, where $\phi_{0} := W_+ \psi_{0}$, would give
\begin{align} \label{psite8}
\psi_{t} & =  I(e^{ - i E_{\mathrm{gs}} t} P_{\mathrm{gs}} \otimes e^{ - i H_f t}\chi_{[0,a-E_{\mathrm{gs}}]}(H_f)) \phi_{0} + \mathcal{O}( \e' ) + C(\e',m) o_t( 1 ) + C(\e')o_m(1),
\end{align}
which, after letting $t\to \infty$, next $m \to \infty$ and then $\e' \to 0$,  gives
\begin{align} \label{psite9}
\lim_{t\to \infty}\|\psi_{t} -  I(e^{ - i E_{\mathrm{gs}} t} P_{\mathrm{gs}} \otimes e^{ - i H_f t}\chi_{[0,a-E_{\mathrm{gs}}]}(H_f)) W_+ \psi_{0}\|=0.
\end{align}

\bigskip

\section{Proof of minimal velocity estimates}\label{sec:mve-y}

In this section we use Theorems \ref{thm:mve1} and \ref{thm:mve2} to prove the minimal velocity estimates of Theorems \ref{thm:mve1y} and \ref{thm:mve2y}.

\begin{proof}[Proof of Theorem \ref{thm:mve1y}]

To prove \eqref{mve1y_0}, we use several iterations of Proposition \ref{prop:prop-obs-meth2}.  We  consider the one-parameter family of one-photon operators
\begin{equation*}
 \phi_{t}:=t^{-a\nu(0)}\chi_{w \geq 1},
 \end{equation*}
 with $w:=\big(\frac{|y|}{c' t^{\beta}}\big)^2$, $a>1$, and $\nu(\del)\ge 0$, the same as in \eqref{dGk-bnd}. We use the expansion \eqref{dphit-exp}. We compute
\begin{align}\label{dv3}
dw &= \frac{2 b}{(c' t^{\beta})^2} -\frac{2\beta w}{ t},
\end{align}
where, recall, $b:=\frac{1}{2} ( \nabla \omega \cdot y + {\hbox{ h.c.}})$. We use the notation $\tilde \chi_\beta\equiv \chi_{w \geq 1}$. We write $b=b_\e+\e\frac{1}{2} (\frac{1}{\omega_{\e}} \nabla \omega \cdot y + {\hbox{ h.c.}})$, where, recall, $ \omega_{\e}:=\omega + \e,\ \e := t^{-\kappa}.$ We choose $\kappa>0$ satisfying
\begin{align}\label{kappa-ineq}
4\beta - 3 > \kappa >2- 2\beta+\nu(-1)-\nu(0).
\end{align}
Using the notation $v:=\frac{b_\e}{c t^\beta}$ and the partition of unity $\chi_{v \ge 1}+\chi_{ v \le 1}=\one$, we find $b_\e \ge c t^\beta + (b_\e-c t^\beta)\chi_{ v \le 1}$. Commutator estimates of the type considered in Appendix \ref{sec:tech} (see Lemma \ref{lem:local-est}) give $\chi_{ v \le -1 } (\tilde \chi'_\beta )^{1/2} = \CO( t^{-\beta + \kappa} )$ for $\tilde c > c/2$, which, together with $b_\e ( \tilde \chi'_\beta )^{1/2} = \CO( t^\beta )$, yields
\begin{align*}
(\tilde \chi'_\beta)^{1/2}  b_\e \chi_{ v \le 1}(\tilde \chi'_\beta)^{1/2} \ge - \tilde c t^\beta (\tilde \chi'_\beta)^{1/2} \chi_{ v \le 1}(\tilde \chi'_\beta)^{1/2} - C t^\kappa\tilde \chi'_\beta.
\end{align*}
The last two estimates, together with $ v \le 1$ on $\supp \tilde\chi_{v\ge 1}'$,  give $d   \phi_t\ge  p_t - \tilde p_{t} + \mathrm{rem}$, where
\begin{align*}
& p_t:= \frac{ 2 }{ t^{a \nu(0) } } \Big(\frac{c}{(c')^2  t^{\beta}}-\frac{\beta }{ t } \Big ) \tilde \chi'_\beta ,\\
& \tilde p_{t}:= \frac{2 (\tilde c+c) }{c'^2 t^{\beta+a\nu(0)}} (\tilde \chi'_\beta)^{1/2}  \chi_{ v \le 1}(\tilde \chi'_\beta)^{1/2} ,
 \end{align*}
and $\mathrm{rem}=\sum_{i=1}^4 \mathrm{rem}_i,$ with $\mathrm{rem}_1$ given by \eqref{dphit-exp} with $\chi_\beta$ replaced by $\tilde\chi_\beta$,
\begin{align*}
\mathrm{rem}_2:=\frac{c}{(c'  t^\beta)^2t^{\kappa+a\nu(0)}}  (\frac{1}{\omega_{\e}} \nabla \omega \cdot y + {\hbox{ h.c.}})\tilde \chi'_\beta ,
\end{align*}
$\mathrm{rem}_3= \CO( t^{-2\beta + \kappa - a\nu(0) }) $, and $\mathrm{rem}_4:=-a\nu(0)t^{-1}\phi_{t}$. If $\beta =1$, then we choose $c > (c')^2$ so that $p_t\ge 0$.

As in the proof of Theorem \ref{thm:mve1},  we deduce that the remainders $\mathrm{rem}_i,\ i=1, 2, 3,4$, satisfy the estimates \eqref{remi-est}, $i=1, 2, 3,4$, with  $\rho_1=\rho_2= -1$, $\rho_3=\rho_4=0 $, $\lam_1= 2\beta +a\nu(0)$,  $\lam_2= 2\beta + \kappa +a\nu(0)$, $\lam_3= 2\beta -\kappa+a\nu(0)$ and $\lam_4= 1+a\nu(0)$. In particular, the estimate for $i=1$ follows from Lemma \ref{a65}. Since $2\beta >1+\nu(-1)-a\nu(0)$ and  $2\beta - \kappa > 1$, the remainder $ \mathrm{rem}=\sum_{i=1}^4 \mathrm{rem}_i$ gives an integrable term. (Note $\mathrm{rem}_2=0$, if $\nu(0)=0$.)

Now, we estimate the contribution of $ \tilde p_{t}$. Let  $\g=2\beta-1\le \beta$ and decompose $\tilde p_{t} = p_{t1}+  p_{t2} ,$ where
 \begin{align*}
 & p_{t1}:= \frac{\mathrm{const}}{ t^{\beta+a\nu(0)}} (\tilde \chi'_\beta)^{1/2} \chi_{c_1 t^\g\le b_\e\le c t^\beta}(\tilde \chi'_\beta)^{1/2} ,\\
 &p_{t2}:=\frac{\mathrm{const}}{ t^{\beta+a\nu(0)}} (\tilde \chi'_\beta)^{1/2} \chi_{ b_\e\le c_1 t^\g}(\tilde \chi'_\beta)^{1/2} ,
\end{align*}
with $ c_1<1$, if $\g=1$, and $c_1<\beta (c')^2 $ if $\g<1$, and $\mathrm{const} =\frac{c'+c}{c'}$. First, we estimate the contribution of $p_{t1}$. Since $[ (\tilde \chi'_\beta)^{1/2} , ( \chi_{c_1 t^\g\le b_\e\le c t^\beta})^{1/2} ] = \CO( t^{-\beta + \kappa} )$ (see Lemma \ref{lm:comm_est} of Appendix \ref{sec:tech}) and since $2\beta - \kappa > 1$, it suffices to estimate the contribution of $\frac{\mathrm{const}}{ t^{\beta}} \chi_{c_1 t^\g\le b_\e\le c t^\beta}$. To this end, we use  the propagation observable
\begin{align}\label{phit1}
\phi_{t1}:=  t^{-a \nu(0)} h_\beta \chi_\g ,
\end{align}
where
$h_\beta\equiv h(\frac{b_\e}{c t^{\beta}}),\ h(\lam):=\int_\lam^\infty ds\chi_{s\le 1}$, and  $\chi_\g\equiv \chi_{\frac{b_\e}{c_1 t^{\g}}\ge 1}$.  As in \eqref{dtb3}, we have
\begin{align}\label{dtb4}
\frac{1}{c t^{\gamma}} h_\beta \p_t b_\e \chi_\gamma' \le \frac{\mathrm{const}}{ t^{1 + \gamma-\kappa}}, \quad \frac{1}{c t^{\beta}}h_\beta' \p_t b_\e \chi_\gamma \ge -\frac{\mathrm{const}}{ t^{1 + \beta-\kappa}}.
\end{align}
Using this together with \eqref{dv}, we compute
\begin{align*}
d   \phi_{t1} \le (\frac{\theta_\e}{c t^{\beta+a\nu(0}}- \frac{\beta b_\e}{c t^{\beta+1+a\nu(0)}})h_\beta'  \chi_\g+ h_\beta  \chi_\g'(\frac{\theta_\e}{c_1 t^{\g+a\nu(0)}}- \frac{\g b_\e}{c_1 t^{\g+1+a\nu(0)}}) + \sum_{i=1}^3 \mathrm{rem}_i',
\end{align*}
where  $\mathrm{rem}_1'$ is a sum of two terms of the form of $\mathrm{rem}_1$ given in \eqref{dphit-exp}, with $\chi_\beta$ replaced by $h_\beta$, in one, and by $\chi_\g$, in the other, $\mathrm{rem}_2':= \CO( t^{-1-\gamma + \kappa - a \nu(0)} )$, and $\mathrm{rem}'_3:=-a\nu(0)t^{-1}\phi_{t1}$. We estimate  $\theta_\e- \frac{\beta b_\e}{t}\ge 1-\frac{1}{\om_\e t^{\kappa}} - \frac{\beta c}{ t^{1-\beta}}$ on $\supp h_\beta'$ and $\theta_\e - \frac{\g b_\e}{ t } \le 1-\frac{1}{\om_\e t^{\kappa}} - \frac{ \g c_1 }{2 t^{1-\gamma} }$ on $\supp \chi_\g'$. Using $h_\beta'\le 0$,  $\chi_\g'\ge 0$, $ h_\beta \le 1 - \frac{ b_\e }{ ct^\beta }$ and $\frac{ b_\e }{ ct^\beta } = \CO( t^{-\beta+\g } )$ on $\supp \chi_\g'$, this gives
 \begin{align*}
&d   \phi_{t1}\le -  p_{t1}' +\tilde p_{t1}+ \mathrm{rem}',
\end{align*}
with $\mathrm{rem}':=\sum_{i=1}^4 \mathrm{rem}_i',\  \mathrm{rem}_4':= \om^{-1/2} \CO( t^{-\beta-\kappa-a\nu(0)}) \om^{-1/2} $, and
\begin{align*}
p_{t1}'&:= t^{-a\nu(0)} (1 - \frac{\beta}{ t }) h_\beta'  \chi_\g ,\qquad \tilde p_{t1}:=\frac{1}{c t^{\gamma+a\nu(0)}} \chi_\g' .
\end{align*}
By \eqref{mve1}, the term $\tilde p_{t1}$ gives an integrable contribution. We deduce as above that the remainders $\mathrm{rem}_i',\ i=1, 2, 3,4$, satisfy the estimates \eqref{remi-est}, $i=1, 2, 3,4$, with $\rho_1=\rho_2 = \rho_3  = 0$, $\rho_4 = -1$,  $\lam_1= 2\gamma - \kappa +a\nu(0)$,  $\lam_2= 1+\gamma-\kappa +a\nu(0)$, $\lam_3 = 1 + a \nu ( 0 )$, and $\lam_4 = \beta + \kappa + a \nu( 0 )$. Since $2\gamma - \kappa > 1$,  $\gamma > \kappa$, and $\beta+\kappa>1+\nu(-1)-a\nu(0)$, the remainder $\mathrm{rem}'=\sum_i \mathrm{rem}_i'$ is integrable. Finally, \eqref{phitg-est} with $\lam' < a \nu( 0 ) + (\frac{3}{2} + \mu)\gamma$ holds by Lemma \ref{lem:chig-est} of Appendix \ref{sec:tech}.  Hence, $\phi_{t1}$ is a strong one-photon propagation observable and therefore we have the estimate
 \begin{align}\label{p1-est'}
\int_1^\infty dt\|\d\G(p_{t1})^{1/2}\psi_t\|^{2}& \lesssim \int_1^\infty dt\|\d\G(p_{t1}')^{1/2}\psi_t\|^{2} \lesssim \| \psi_0 \|_{-1}^2 .
\end{align}
(In fact, by multiplying the observable \eqref{phit1} by $t^\delta$ for an appropriate $\delta>0$, we can obtain a stronger estimate.)

Now, we consider $p_{t2}$. Let $f_\beta\equiv f ( w )$, where $f(\lam):=\chi_{\lambda\ge 1}$ and, recall, $w =  \big(\frac{|y|}{c' t^{\beta}}\big)^2$, and $h_\g\equiv h( v_\gamma )$, with $h(\lam):=\int_\lam^\infty ds\chi_{s\le 1}$ and $v_\gamma := \frac{b_\e}{c_1 t^\g}$. We use  the propagation observable
\begin{align}\label{phit2}
\phi_{t2}:= t^{-a\nu(0) } (  f_\beta h_\g +h_\g f_\beta ).
\end{align}
 Using \eqref{dv}, \eqref{dtb2}, \eqref{dv3}, $b=b_\e+\e\frac{1}{2} (\frac{1}{\omega_{\kappa}} \nabla \omega \cdot y + {\hbox{ h.c.}})$, $ b_\e\le c_1 t^\g$ on $\supp \chi_{ v_\gamma \le 1}$, $\gamma = 2\beta-1$ and $[ (f_\beta')^{1/2} ,  h_\g ] = \CO( t^{-\g + \kappa})$ (see Lemma \ref{lm:comm_est} of Appendix \ref{sec:tech}), we compute
 \begin{align*}
d \phi_{t2} \le & t^{-a\nu(0) } \big ( (\frac{c_1}{(c')^2 }-\beta)\frac{2 }{ t}(f_\beta')^{1/2}  h_\g(f_\beta')^{1/2}+ f_\beta h_\g' (dv_\gamma)+ (dv_\gamma) h_\g'f_\beta  \big ) +\sum_{i=1}^4 \mathrm{rem}_i'',
\end{align*}
where $dv_\gamma =\frac{\theta_\e}{c_1 t^\g}- \frac{\g b_\e}{c_1 t^{\g+1}}$, $\mathrm{rem}_1''$ is a term of the form of $\mathrm{rem}_1$ given in \eqref{dphit-exp},  with $\chi_\beta$ replaced by $f_\beta$, likewise, $\mathrm{rem}_2''$ is a term of the form of $\mathrm{rem}_1$ given in \eqref{dphit-exp},  with $\chi_\beta$ replaced by $h_\g$, and $\mathrm{rem}''_3 = \CO( t^{-1-\g+\kappa-a\nu(0) })$ and $\mathrm{rem}''_4:=-a\nu(0)t^{-1}\phi_{t2}$. To estimate $d v_\gamma=\frac{\theta_\e}{c_1 t^\g}- \frac{\g b_\e}{c_1 t^{\g+1}}$, we use that $f_\beta' \ge 0$, $h'_{\g}\le 0$, $\theta_\e = 1 - t^{-\kappa} \omega_\e^{-1}$, $v_\gamma h'_\g \le h'_\g$, and  $f_\beta  h_\g' (dv_\gamma) + (dv_\gamma) h_\g' f_\beta =  -f_\beta^{1/2}  ( - h_\g' )^{1/2} (dv_\gamma) ( - h_\g' )^{1/2} f_\beta^{1/2}+\CO( t^{-\g + \kappa})$ (see again Lemma \ref{lm:comm_est} of Appendix \ref{sec:tech}), to obtain
 \begin{align*}
&d   \phi_{t2}\le -  p_{t2}' + \mathrm{rem}'',
\end{align*}
with $\mathrm{rem}'':=\sum_{i=1}^6 \mathrm{rem}_i''$, $\mathrm{rem}''_5 = \CO( t^{-2\g + \kappa - a \nu(0) }) $, $\mathrm{rem}''_6 = \om^{-1/2} \CO( t^{-\g-\kappa - a \nu(0) }) \om^{ -1/2}$ and (at least for $t$ sufficiently large)
\begin{align*}
p_{t2}':= t^{- a \nu (0) } \big [ - (\frac{2c_1}{(c')^2  }-2\beta)\frac{1 }{ t}(f_\beta')^{1/2}  h_\g(f_\beta')^{1/2}+ (1 - \frac{\g c_1 }{t^{1-\g}})\frac{1}{ c_1 t^\g}f_\beta^{1/2}  h_\g' f_\beta^{1/2} \big ].
\end{align*}
Since $\frac{c_1}{(c')^2  }<\beta$ and either $\g<1$ or $\g=1$ and $ c_1<1$, and $f_\beta' \ge 0$ and  $h_\g'\le 0$, both terms in the square braces on the r.h.s. are non-positive. We deduce as above that the remainders $\mathrm{rem}_i'',\ i=1,\dots,6$, satisfy the estimates \eqref{remi-est}, $i=1,\dots, 6$, with  $\rho_1=\rho_6= -1$, $\rho_2 =\rho_3=\rho_4=\rho_5=0$, $\lam_1= 2\beta +a\nu(0)$, $\lam_2= \lam_5 = 2\gamma-\kappa + a\nu(\del)$, $\lam_3 = 1+\gamma-\kappa + a\nu(0)$, $\lam_4 = 1+a\nu(0)$, $\lam_6 = \gamma + \kappa + a\nu(0)$.  Since  $2\beta > \gamma+\kappa > 1+\nu(-1)-a\nu(0)$, $2\gamma - \kappa > 1$ and $\gamma > \kappa$, the condition \eqref{rem-est} is satisfied. Moreover, \eqref{phitg-est} with $\lam' < a \nu( 0 ) + (\frac{3}{2} + \mu)\beta$ holds by \cite[Lemma 3.1]{BoFaSig}. Therefore $\phi_{t2}$ is a strong one-photon propagation observable and we have the estimate
 \begin{align}\label{p2-est'}
\int_1^\infty dt\|\d\G(p_{t2})^{1/2}\psi_t\|^{2}& \lesssim \int_1^\infty dt\|\d\G(p_{t2}')^{1/2}\psi_t\|^{2} \lesssim \| \psi_0 \|_{-1}^2 .
\end{align}
(In fact, by multiplying the observable \eqref{phit2} by $t^\delta$ for an appropriate $\delta>0$, we can obtain a stronger estimate.)

Since  $\tilde p_{t}=   p_{t1}+  p_{t2},$  estimates \eqref{p1-est'} and \eqref{p2-est'} imply the estimate
\begin{align}\label{p-est}
\int_1^\infty dt\|\d\G(p_t)^{1/2}\psi_t\|^{2}  \lesssim  \| \psi_0 \|_{-1}^2,
\end{align}
which due to $ \tilde\chi_{\beta}'\approx \chi_{v= 1}$, implies the estimate \eqref{mve1y_0}.
\end{proof}

\medskip

\begin{proof}[Proof of Theorem \ref{thm:mve2y}]
To prove \eqref{mve2y}, we  begin with the following estimate, proven in the localization lemma \ref{lem:local-est} of Appendix \ref{sec:tech}:
\begin{align}\label{local-est}
\chi_{b_\e \geq c' t^\al}\chi_{\mid y \mid \leq c t^{\alpha}}=\CO(t^{-(\al-\kappa)}),
\end{align}
for $\e= t^{-\kappa}$, $\kappa<\al$, and $c < c'/2$. Now, let $\chi_{b_\e \leq  c' t^\al}^2+\chi_{b_\e \geq  c' t^\al}^2=\one$ and write
\begin{align}\label{deco1}
\chi_{\mid y \mid \leq c t^{\alpha}}^2=\chi_{b_\e \leq  c' t^\al}\chi_{\mid y \mid \leq c t^{\alpha}}^2\chi_{b_\e \leq c'  t^\alpha}+R \le \chi_{b_\e \leq  c' t^\al}^2+R,
\end{align}
where $R:=\chi_{b_\e \leq  c' t^\al}\chi_{\mid y \mid \leq c  t^{\alpha}}^2\chi_{b_\e \geq  c' t^\al}+\chi_{b_\e \geq  c' t^\al}\chi_{\mid y \mid \leq c t^{\alpha}}^2\chi_{b_\e \leq  c' t^\al}+\chi_{b_\e \geq  c' t^\al}\chi_{\mid y \mid \leq c t^{\alpha}}^2\chi_{b_\e \geq  c' t^\al}$.
The estimates \eqref{local-est} and \eqref{deco1}  give
\begin{align}\label{local-est''}
\chi_{\mid y \mid \leq c t^{\alpha}}^2 \le \chi_{b_\e \leq c't^\alpha}^2 + \CO(t^{-(\al-\kappa)}),
\end{align}
which in turn implies
\begin{equation}\label{G-bnd}
\| \Gamma (\chi_{|y| \leq c t^\alpha} )^{1/2} \psi \|\ \lesssim \| \Gamma (\chi_{b_\e \leq c' t^\alpha} )^{1/2} \psi \| + C t^{-(\al-\kappa) / 2} \| (N+1)^{1/2}\psi \|.
\end{equation}
This, together with \eqref{mve2}, yields \eqref{mve2y}.
\end{proof}

\bigskip

\appendix

\section{ Photon \# and low momentum estimate}\label{sec:low-mom}

Recall the notation $\lan A\ran_\psi:=\lan\psi, A\psi \ran$. The idea of the proof of the following estimate follows \cite{Ger}  and \cite{BoFaSig}.
\begin{proposition}\label{prop:lm-bnd}
Assume \eqref{g-est} with $\mu > -1/2$. Let $\psi_0 \in D(\d\G( \omega^{\rho} )^{1/2})$.  Then for any $\rho\in [-1, 1]$,
\begin{equation}\label{lm-bnd}
\langle \d\G(\omega^{\rho})\rangle_{\psi_t} \lesssim  t^{\nu(\rho)} \| \psi_0 \|_H^2 + \langle \d\G(\omega^{\rho})\rangle_{\psi_0} ,\ \nu(\rho) = \frac{1-\rho}{2+\mu}.
\end{equation}
\end{proposition}
\begin{proof}
Decompose $d\G(\omega^{\rho})=K_1 + K_2$, where
\begin{equation*}
K_1 := \d \Gamma (\omega^{\rho}\chi_{t^{\alpha}\omega \leq 1})\ \quad \mbox{ and }\ \quad
K_2 := \d \Gamma (\omega^{\rho}\chi_{t^{\alpha}\omega \geq 1}).
\end{equation*}
Then, by \eqref{Hf-H-bnd},
\begin{equation}\label{N2-bnd}
\langle K_2\rangle_\psi \leq \langle \d \Gamma (t^{\alpha (1-\rho)} \omega
\chi_{t^{\alpha}\omega \geq 1}) \rangle_{\psi_t} \le t^{\alpha (1-\rho)} \langle
H_f\rangle_{\psi_t}\lesssim  t^{\alpha (1-\rho)} \|\psi_0\|_H^2.
\end{equation}
On the other hand, we have by \eqref{DPhi},
\begin{equation}\label{DN1}
D K_1 =  \d \Gamma (\alpha \omega^{1-\rho} t^{\alpha -1} \chi'_{t^{\alpha}\omega \leq 1}) - I ( i \omega^{\rho}\chi_{t^{\alpha}\omega \leq 1} g).
\end{equation}
Since $\| g(k) \|_{\cH_p} \lesssim |k|^\mu \xi(k)$ (see \eqref{g-est}), we obtain
\begin{eqnarray}
\int \omega^{2\rho} \chi_{t^\alpha \omega \leq 1} \| g(k) \|_{\cH_p}^2 (\omega^{-1}+1) d^3 k
\lesssim  t^{-2 (1+\mu+\rho) \alpha}.
\end{eqnarray}
This together with \eqref{I-est} and \eqref{Hf-H-bnd} gives
\begin{eqnarray}
|\langle I ( i \omega^{\rho}\chi_{t^{\alpha} \omega\leq 1} g )\rangle_{\psi_t}| \lesssim t^{- (1+\mu+\rho) \alpha}\| \psi_0 \|_H^2.
\end{eqnarray}
Hence, by \eqref{DN1}, since $\partial_t \langle K_1\rangle_{\psi_t} = \langle DK_1\rangle_{\psi_t}$, $\chi_{t^{\alpha}\omega \leq 1}' \leq 0$, we obtain
\begin{equation*}
\partial_t \langle K_1\rangle_{\psi_t} \lesssim  t^{-(1+\mu+\rho) \alpha} \| \psi_0 \|_H^2
\end{equation*}
and therefore
\begin{equation}\label{N1-bnd}
\langle K_1\rangle_{\psi_t} \leq C t^{\nu'} \| \psi_0 \|_H^2 + \langle \d\G( \omega^{-\rho})\rangle_{\psi_0},
\end{equation}
where $\nu'=1-(1+\mu+\rho) \alpha $, if $(1+\mu+\rho) \alpha <1$ and $\nu'=0$, if $(1+\mu+\rho) \alpha >1$. Estimates \eqref{N1-bnd} and \eqref{N2-bnd} with $\alpha = \frac{1}{2+\mu}$, if $\rho < 1$, give \eqref{lm-bnd}. The case $\rho = 1$ follows from \eqref{Hf-H-bnd}.
\end{proof}

\begin{corollary}\label{cor:chiN-bnd}
Assume \eqref{g-est} with $\mu > -1/2$, let $\psi_0 \in D(\d\G( \omega^{-\rho} )^{1/2})$, and denote $K_\rho:=\d\G(\omega^{-\rho})$.  Then for any $\g\ge 0$ and any $c> 0$,
\begin{equation}\label{chiN-bnd}
\| \chi_{K_\rho \geq c t^\gamma} \psi_t \| \lesssim t^{-\frac{\gamma}{2} + \frac{1+\rho}{2(2+\mu)}}\| \psi_0 \|_H^2 +t^{-\frac{\gamma}{2}}\langle K_\rho\rangle_{\psi_0}.
\end{equation}
\end{corollary}
\begin{proof}
We have
\begin{equation*}
\| \chi_{K_\rho\geq c t^\gamma} \psi_t \|\ \leq c^{-\frac{\gamma}{2}} t^{-\frac{\gamma}{2}}\ \| \chi_{K_\rho \geq c t^\gamma} K_\rho^{\frac{1}{2}} \psi_t \| \leq c^{-\frac{\gamma}{2}} t^{-\frac{\gamma}{2}}\ \| K_\rho^{\frac{1}{2}} \psi_t \|
\end{equation*}
Now applying \eqref{lm-bnd} we arrive at \eqref{chiN-bnd}.
\end{proof}

\noindent \textbf{Remark.}
A minor modification of the proof above give the following bound for $\rho<0$ and $\nu_1(\rho):={\frac{-\rho}{\frac{3}{2}+\mu}}$,
\begin{equation}\label{lm-bnd'}
\langle \d\G(\omega^{\rho})\rangle_{\psi_t} \lesssim  t^{ \nu_1(\rho)} \big( \| \psi_t \|_N^2 +  \| \psi_0 \|_H^2\big) + \langle \d\G( \omega^{\rho} ) \rangle_{\psi_0} .
\end{equation}

\bigskip

\section{Commutator estimates}\label{sec:tech}

In this appendix, we estimate some localization terms and commutators appearing in Section \ref{sec:pf-mve1}. Recall that $b_\e :=\frac{1}{2} (\theta_\e \nabla \omega \cdot y + {\hbox{ h.c.}})$, where $\theta_\e =\frac{\omega}{\omega_{\e}},\ \omega_{\e}:=\omega + \e$, $\e= t^{-\kappa},$ with $\kappa\ge 0$. The following lemma is a straightforward consequence of the Helffer-Sj{\"o}strand formula. We do not detail the proof.
\begin{lemma}\label{lm:comm_est}
Let $h,\tilde h$ be smooth function satisfying the estimates $\big\vert \partial_s^n h(s) \big\vert \leq \mathrm{C}_n \< s \>^{-n} \text{ for } n \geq 0$ and likewise for $\tilde h$. Let $w_\alpha = |y| / (c_1 t^\alpha )$, $v_\beta = b_\e / ( c_2 t^\beta )$, with $0<\alpha,\beta\le1$. The following estimates hold
\begin{align*}
&[ h ( w_\alpha ) , \omega ] = \CO(t^{-\alpha}) , \qquad [ \tilde h ( v_\beta ) , \omega ] = \CO(t^{-\beta}) ,  \qquad[ h ( w_\alpha ) , b_\e ] = \CO(t^{\kappa}) ,  \notag \\
&[ h( w_\alpha ) , \tilde h ( v_\beta ) ] = \CO( t^{-\beta+\kappa} ) , \qquad b_\e [ h( w_\alpha ) , \tilde h ( v_\beta ) ] = \CO( t^{\kappa} ) .
\end{align*}
\end{lemma}

Now we prove the following abstract result.

\begin{lemma}\sl \label{lem:rem1-est}
Let $h$ be a smooth function satisfying the estimates $\big\vert \partial_s^n h(s) \big\vert \leq \mathrm{C}_n \< s \>^{-n} \text{ for } n \geq 0$. Assume that the commutators $ [ v ,  \omega ]$ and $[v , [ v ,  \omega ] ]$ are bounded, and for some $z$ in $\mathbb{C} \setminus \mathbb{R}$, $(v-z)^{-1}$ preserves $D(\omega)$.  Then the operator
$r:= [ h(v) , \omega ] - [v , \omega ] h'(v)$
is bounded as
 \begin{equation}\label{comm-rem-est}
\|r\|\lesssim  \big\Vert [ v , [ v ,  \omega ] ]\Vert .
\end{equation}
\end{lemma}
\begin{proof}
 We would like to use the Helffer--Sj{\"o}strand formula for $h$. Since $h$ might not decay at infinity, we cannot directly express $h ( v )$ by this formula. Therefore, we approximate $h ( v )$ as follows.   Consider $\varphi \in \mathrm{C}_{0}^{\infty} ( \R ; [ 0 , 1 ] )$ equal to $1$ near $0$ and $\varphi_{R} ( \cdot ) = \varphi( \cdot / R )$ for $R > 0$. Let $\widetilde{h}$ be an almost analytic extensions of $h$
such that $\widetilde{h} \vert_{\mathbb{R}} = h$,
\begin{equation} \label{a66}
\supp \widetilde{h} \subset \big\{ z \in \C ; \ \vert \im z \vert \leq \mathrm{C} \< \re z \> \big\} ,
\end{equation}
$\vert \widetilde{h} (z) \vert \leq \mathrm{C}$
and, for all $n \in \N$,
\begin{equation} \label{a67}
\Big\vert \partial_{\bar z} \widetilde{h} (z) \Big\vert \leq \mathrm{C}_n \< \re z \>^{\rho - 1 - n} \vert \im z \vert^n .
\end{equation}
Similarly let $\widetilde{\varphi} \in \mathrm{C}_{0}^{\infty} ( \C )$ be an almost analytic extension of $\varphi$ satisfying these estimates.  As a quadratic form on $D ( \omega )$, we have
\begin{equation}
\big[ h (  v  ) , \om \big] = \slim_{R \to \infty} \big[ ( \varphi_{R} h ) ( v ) ,  \omega \big] .
\end{equation}
Since $(v-z)^{-1}$ preserves $D(\omega)$ for some $z$ in the resolvent set of $v$ (and hence for any such $z$, see \cite[Lemma 6.2.1]{AmBoGe96_01}), we can compute, using the Helffer--Sj{\"o}strand  representation for $( \varphi_{R} h ) (  v  )$,
\begin{align}
\big[ ( \varphi_{R} h ) ( v ) ,  \omega \big] &= \frac{1}{\pi} \int \partial_{\bar z} ( \widetilde{\varphi}_{R} \widetilde{h} ) (z) \big[ ( v - z )^{-1} , \omega \big] \ddre z \ddim z    \nonumber \\
&= - \frac{1}{\pi} \int \partial_{\bar z} ( \widetilde{\varphi}_{R} \widetilde{h} ) (z) ( v - z )^{-1} [ v , \omega ] ( v - z )^{-1} \ddre z \ddim z     \nonumber \\
&= [ v , \omega]  ( \varphi_{R} h )^{\prime} ( v ) + r_{R} ,  \label{a72}
\end{align}
as a quadratic form on $D(\omega)$, where
\begin{align}
r_{R} &= - \frac{1}{\pi  }  \int \partial_{\bar z} ( \widetilde{\varphi}_{R} \widetilde{h} ) (z) \big [ ( v - z )^{-1} , [ v , \omega ] \big ] ( v - z )^{-1}  \ddre z \ddim z    \nonumber \\
&= \frac{1}{\pi  } \int \partial_{\bar z} ( \widetilde{\varphi}_{R} \widetilde{h} ) (z) ( v - z )^{-1} [ v , [ v , \omega ] ] ( v - z )^{-2} \ddre z \ddim z . \label{a69}
\end{align}

Now, using $ (v - z )^{-1}= \CO \big(   \vert \im z \vert^{- 1} \big)$, we obtain that
\begin{gather}
\big\Vert( v - z )^{-1} [ v , [ v , \omega ] ] ( v - z )^{-2}\big\Vert \lesssim \vert \im z \vert^{- 3}  \big\Vert [ v , [ v , \omega] ]\big\Vert . \label{a70}
\end{gather}
Besides, for all $n \in \N$,
\begin{equation} \label{a71}
\vert \partial_{\bar z} ( \widetilde{\varphi}_{R} \widetilde{h} ) (z) \vert \leq \mathrm{C}_n \< \re z \>^{\rho - 1 - n} \vert \im z \vert^n ,
\end{equation}
where $\mathrm{C}_n > 0$ is independent of $R \geq 1$. Using \eqref{a69} together with \eqref{a70},
we see that there exists $\mathrm{C} > 0$ such that $\Vert  r_{R}  \Vert \leq \mathrm{C} \big\Vert [ v , [ v , \omega ] ]\Vert $, for all $R \geq 1$. Finally, since $( \varphi_{R} h )^{\prime} ( v )$ converges strongly to $h^{\prime} ( v )$,
the lemma follows from \eqref{a72} and the previous estimate.
\end{proof}

We want apply the lemma above to the \emph{time-dependent} self-adjoint operator $v:=\frac{b_\e}{c t^{\beta}}$.
\begin{corollary}\sl \label{cor:rem1-est}
Let  $h$ be a smooth function satisfying the estimates $\big\vert \partial_s^n h(s) \big\vert \leq \mathrm{C}_n \< s \>^{-n} \text{ for } n \geq 0$ and let $v:=\frac{b_\e}{ct^{\beta}}$, where $c>0$, $\e= t^{-\kappa},$ with $0 \le \kappa\le \beta \le 1$. Then the operator
$r:= dh(v) -  (d v) h'(v)$ is bounded as
 \begin{equation}
\|r\|\lesssim t^{-\lam},\ \lam := 2\beta-\kappa.
\end{equation}
\end{corollary}
\begin{proof}
Observe that
\begin{align*}
d h(v) -  (d v) h'(v) = [ h(v) , i \omega ] - [ v , i \omega ] h'( v ) + \p_t h(v) - ( \p_t v ) h'(v).
\end{align*}
It is not difficult to verify that $(v-z)^{-1}$ preserves $D(\omega)$ for any $z \in \mathbb{C} \setminus \mathbb{R}$. Hence it follows from the computations
\begin{align}\label{eq:j5}
[ v , i \omega ] = t^{-\beta} \theta_\e, \qquad [ v , [ v , i \omega ] ] = t^{-2\beta} \theta_\e \omega_\e^{-2} \e ,
\end{align}
that we can apply Lemma \ref{lem:rem1-est}. The estimate
\begin{gather}
[ v , [ v , \omega ] ]=\CO (  \om_\e^{- 1} t^{- 2\beta} )= \CO \big(   t^{- 2\beta+\kappa} \big)
\end{gather}
then gives
\begin{align*}
\| [h(v) , i \omega ] - [ v , i \omega ] h'( v ) \| \lesssim t^{-2\beta+\kappa}.
\end{align*}
It remains to estimate $\| \p_t h(v) - (\p_t v)h'(v) \|$. It is not difficult to verify that $D(b_\e)$ is independent of $t$.
Using the notations of the proof of Lemma \ref{lem:rem1-est} and the fact that $\p_t h(v) = \slim_{R\to\infty} \p_t (\varphi_R h)(v)$, we compute
\begin{align*}
\p_t ( \varphi_{R} h ) ( v ) &= \frac{1}{\pi} \int \partial_{\bar z} ( \widetilde{\varphi}_{R} \widetilde{h} ) (z) \p_t ( v - z )^{-1}  \ddre z \ddim z    \nonumber \\
&= - \frac{1}{\pi} \int \partial_{\bar z} ( \widetilde{\varphi}_{R} \widetilde{h} ) (z) ( v - z )^{-1} (\p_t v) ( v - z )^{-1} \ddre z \ddim z     \nonumber \\
&= (\p_t v)  ( \varphi_{R} h )^{\prime} ( v ) + r'_{R} ,
\end{align*}
where
\begin{align}
r'_{R} &= - \frac{1}{\pi  }  \int \partial_{\bar z} ( \widetilde{\varphi}_{R} \widetilde{h} ) (z) \big [ ( v - z )^{-1} , \p_t v \big ] ( v - z )^{-1}  \ddre z \ddim z    \nonumber \\
&= \frac{1}{\pi  } \int \partial_{\bar z} ( \widetilde{\varphi}_{R} \widetilde{h} ) (z) ( v - z )^{-1} [ v , \p_t v ] ( v - z )^{-2} \ddre z \ddim z . \label{a692}
\end{align}
Now using $\p_t v = - \frac{\beta b_\e}{c t^{\beta+1}}+\frac{1}{c t^{\beta}}\p_t b_\e$ together with \eqref{dtb2}, we estimate
\begin{align*}
[ v , \p_t v ] = \CO ( t^{-1-2\beta + \kappa} ) b_\e + \CO( t^{-1-2\beta+2\kappa} ).
\end{align*}
From this, the properties of $\tilde \varphi$, $\tilde h$, and $\kappa \le \beta$, we deduce that $\| r'_R \| \lesssim t^{ - 1 - \beta + \kappa } \lesssim t^{-2\beta + \kappa}$ uniformly in $R\ge1$. This concludes the proof of the corollary.
\end{proof}

The following lemma is taken from \cite{BoFaSig}. Its proof is similar to the proof of Lemma \ref{lem:rem1-est}
\begin{lemma}\sl \label{a65}
Let $h$ be a smooth function satisfying the estimates $\big\vert \partial_s^n h(s) \big\vert \leq \mathrm{C}_n \< s \>^{-n} \text{ for } n \geq 0$ and $0 \le \delta \leq 1$. Let $w = y^2 / (c t^\alpha )^2$ with $0<\alpha\le1$. We have
\begin{equation*}
\big[ h (  w  ) , i \omega \big] = \frac{1}{ c t^\alpha } h^{\prime}  ( w ) \big( \frac{ y }{ ct^\alpha} \cdot \nabla \omega + \nabla \omega \cdot \frac{ y }{ ct^\alpha } \big) + \mathrm{rem} ,
\end{equation*}
with
\begin{equation*}
\big\Vert \omega^{\frac{\delta}{2}} \, \mathrm{rem} \, \omega^{\frac{\delta}{2}} \big\Vert \lesssim t^{-\alpha (1+ \delta) } .
\end{equation*}
\end{lemma}
Now we prove a localization lemma.
\begin{lemma} \label{lem:local-est}
Let $\kappa<\al$. We have, for $c<c'/2$,
\begin{align}\label{local-est'}
\chi_{b_\e \geq c' t^\al} \chi_{\mid y \mid \leq c t^{\alpha}} = \CO(t^{-(\al-\kappa)}).
\end{align}
\end{lemma}
\begin{proof}
Observe that by the definition of $\chi$ (see Introduction)  and the condition $c < c'/2$, we have $\chi_{ |y| \ge c' t^\al } \chi_{ |y| \le c t^\al } = 0$. Let $c<\bar c  < c'/2$ and let $\tilde \chi_{ |y| \le \bar c t }$ be such that $\chi_{ |y| \le c t } \tilde \chi_{ |y| \le \bar c t } = \chi_{ |y| \le c t }$ and $\chi_{ |y| \ge c' t } \tilde \chi_{ |y| \le \bar c t } = 0$. Define $\bar b_\e := \tilde \chi_{\mid y \mid \leq \bar c  t^{\alpha}} b_\e \tilde \chi_{\mid y \mid \leq \bar c t^{\alpha}}$. It follows from the expression of $b_\e$ that $| \langle u , b_\e u \rangle | \le \| u \| \| |y| u \|$, and hence we deduce that $| \langle u , \bar b_\e u \rangle | \le \bar c  t^{\alpha} \| u \|^2$. This gives $\chi_{\bar b_\e \geq c' t^\alpha}=0$. Using this, we write
  \begin{align}\label{local-deco}
\chi_{b_\e \geq  c't^\alpha}\chi_{\mid y \mid \leq ct^{\alpha}}=(\chi_{b_\e \geq  c't^\alpha}-\chi_{\bar b_\e \geq  c't^\alpha})\chi_{\mid y \mid \leq ct^{\alpha}}.
\end{align}
Let $a:=\frac{b_\e}{c't^\alpha}$ and $\bar a:=\frac{\bar b_\e}{c' t^{\alpha}}$.
Denote $g(a):=\chi_{b_\e \geq  c' t^{\alpha}}$ and $g(\bar a):=\chi_{\bar b_\e \geq  c't^\alpha}$. We will use the construction and notations of the proof of Lemma \ref{lem:rem1-est}.  Using the Helffer-Sj\"ostrand formula for $( \varphi_{R} g ) (  c  )$, we write
\begin{align}
( \varphi_{R} g )(a)-( \varphi_{R} g )(\bar a ) &= \frac{1}{\pi} \int \partial_{\bar z}(\widetilde{\varphi_{R} g}) (z) \big[ (  a - z )^{-1} -( \bar  a - z )^{-1} \big] \ddre z \ddim z    \nonumber \\
&= - \frac{1}{\pi} \int \partial_{\bar z} ( \widetilde{\varphi}_{R} \widetilde{g} ) (z) (  a - z )^{-1} ( a -\bar a ) (\bar  a - z )^{-1} \ddre z \ddim z .   \label{a72'}
\end{align}
Now we show that $( a -\bar a )(\bar  a - z )^{-1}\chi_{\mid y \mid \leq c t^{\alpha}}= \CO(t^{-(\al-\kappa)} |\im z|^{-2})$. We have
\begin{equation*}
a - \bar a = ( 1 - \tilde \chi_{\mid y \mid \leq \bar c t^{\alpha}} ) \frac{ b_\e }{ c' t^\alpha } +\tilde  \chi_{\mid y \mid \leq \bar c t^{\alpha}} \frac{ b_\e }{ c' t^\alpha } ( 1 - \tilde \chi_{\mid y \mid \leq \bar c t^{\alpha}} ),
\end{equation*}
and we observe that, by Lemma \ref{lm:comm_est},
\begin{equation}\label{eq:h1}
[ ( 1 - \tilde \chi_{\mid y \mid \leq \bar c t^{\alpha}} ) , b_\e ] =  \CO(t^{\kappa}).
\end{equation}
Thus
\begin{equation*}
a - \bar a = (1 + \tilde \chi_{\mid y \mid \leq \bar c t^{\alpha}} ) \frac{ b_\e }{ c' t^\alpha } ( 1 - \tilde \chi_{\mid y \mid \leq \bar c t^{\alpha}} ) + \CO(t^{- ( \alpha -\kappa )}),
\end{equation*}
Moreover, we can write
\begin{align*}
( 1 - \tilde \chi_{\mid y \mid \leq \bar c t^{\alpha}} ) (\bar  a - z )^{-1} \chi_{\mid y \mid \leq c t^{\alpha}} &= \big [ ( 1 - \tilde \chi_{\mid y \mid \leq \bar c t^{\alpha}} ) , (\bar  a - z )^{-1} \big ] \chi_{\mid y \mid \leq c t^{\alpha}} \notag \\
&= - (\bar a - z )^{-1} \big [ ( 1 - \tilde \chi_{\mid y \mid \leq \bar c t^{\alpha}} ) , \frac{ b_\e }{ c t^\alpha } \big ] (\bar  a - z )^{-1} \chi_{\mid y \mid \leq c t^{\alpha}} \notag \\
& = \CO ( t^{-(\alpha-\kappa)} |\im z|^{-2} ) ,
\end{align*}
where we used \eqref{eq:h1} to obtain the last estimate. This implies the statement of the lemma.
\end{proof}
\noindent \textbf{Remark.} The estimate \eqref{local-est'} can be improved to
$\chi_{b_\e \geq c' t^\al} \chi_{\mid y \mid \leq c t^{\alpha}} = \CO(t^{-m(\al-\kappa)})$,
for any $m>0$, if we replace $\om_\e:=\om+\e$ in the definition of $b_\e$ by the smooth function $\om_\e:=\sqrt{\om^2+\e^2}$.

\medskip

In conclusion of this appendix we reproduce a statement corresponding to \cite[Lemma 3.1]{BoFaSig} with $b_\e$ instead of $|y|$. The proof is the same.
\begin{lemma}\label{lem:chig-est}
Assume Hypothesis \eqref{g-est} on the coupling function $g$ is satisfied for some $-\frac12 \le \mu \le \frac12$. Then
\begin{equation*}
\big\Vert \eta \chi_{ b_\e \ge c t^{\alpha} } g(k) \big\Vert_{L^{2} ( \R^{3} ; \cH_p )} \lesssim t^{- \tau } , \qquad \tau < ( \frac32 + \mu ) \alpha.
\end{equation*}
\end{lemma}

\bigskip

\section{Technicalities}\label{sec:tech2}

In this appendix we prove technical statements used in the main text. Most of the results we present here are close to known ones. We begin with the following standard result, which was used implicitly at several places.
\begin{lemma}\sl \label{a21}
Let $a,b$ be two self-adjoint operators on $\mathfrak{h}$ with $b \geq 0$, $D (b) \subset D (a)$ and $\Vert a \varphi \Vert \leq \Vert b \varphi \Vert$ for all $\varphi \in D (b)$. Then $D ( \d\Gamma(b) ) \subset D( \d\Gamma(a) )$ and $\Vert \d \Gamma(a) \Phi \Vert \leq \Vert \d \Gamma(b) \Phi \Vert$ for all $\Phi \in D ( \d\Gamma(b) )$.
\end{lemma}
 We recall that, given two operators $a,c$ on $\mathfrak{h}$, the operator $\d\Gamma( a , c )$ was defined in \eqref{dG'}, and $\d\check{\Gamma}(a,c) := U \d \Gamma( a , c )$.
\begin{lemma} \label{lem:checkG-ineq}
Let $j = (j_0, j_\infty)$ and $c = \diag (c_0, c_\infty)$, where $j_0, j_\infty, c_0, c_\infty, c_1, c_2$ are operators on $\mathfrak{h}$.
 Furthermore, assume that $j_0^* j_0 + j_\infty^* j_\infty \le 1$. Then we have the relations
 \begin{align} \label{checkG-ineq'}
| \lan \hat\phi, \d \check{\Gamma} (j, c) \psi\ran | &\le  \| \d \Gamma( |c_0| )^{\frac12} \otimes \one \hat \phi \| \| \d {\Gamma} ( |c_0| )^{\frac12} \psi \| \notag \\
&+ \| \one \otimes \d \Gamma( |c_\infty| )^{\frac12} \hat \phi \| \| \d {\Gamma} ( |c_\infty| )^{\frac12} \psi \|,\\
 \label{eq:est-dG(j,c1c2)}
| \langle u , \d \Gamma ( j , c_1  c_2) v \rangle | &\le \| \d\Gamma( c_1 c_1^* )^{\frac12} u \| \| \d\Gamma( c_2^* c_2 )^{\frac12} v \|.
\end{align}
\end{lemma}
\begin{proof}
 Let $\tilde \phi = U^* \hat \phi$ and for an operator $b$ on $\fh$ define operators $i_0 b:= \diag (b , 0)$ and $i_\infty b:= \diag (0, b)$ on $\fh \oplus \fh$. Since $U^* \d \Gamma( |c_0| )^{\frac12} \otimes \one U = \d \Gamma( i_0 |c_0|  )^{\frac12}$ and $U^* \one \otimes \d \Gamma( |c_\infty| )^{\frac12} U = \d \Gamma( i_\infty |c_\infty| )^{\frac12}$, the statement of the lemma is equivalent to
 \begin{align} \label{checkG-ineq''}
| \lan \tilde \phi, \d \Gamma (j, c) \psi\ran | \le & \| \d \Gamma( i_0 |c_0| )^{\frac12} \tilde \phi \| \| \d {\Gamma} ( |c_0| )^{\frac12} \psi \| \notag \\
&+ \| \d \Gamma( i_\infty |c_\infty| )^{\frac12} \tilde \phi \| \| \d {\Gamma} ( |c_\infty| )^{\frac12} \psi \|.
\end{align}
We decompose $\d \Gamma (j, c) = \d \Gamma(j , i_0 c_0) + \d \Gamma( j , i_\infty c_\infty ) $ and estimate each term separately. We have, using that $\|Êj \|Ê\le 1$,
\begin{align*}
| \lan \tilde \phi, \d \Gamma(j , i_0 c_0) \psi\ran | & \le \sum_{l=1}^n | \lan | i_0 c_0 |_l^{\frac12} \tilde \phi ,  |i_0 c_0|_l^{\frac12} \psi \ran | ,
\end{align*}
 where $|i_0 c_0|_l := \one \otimes \cdots \otimes \one \otimes i_0 |c_0| \otimes \one \otimes \cdots \otimes \one$, with the operator $|i_0 c_0|$ appearing in the $l^{\text{th}}$ component of the tensor product. By the Cauchy-Schwarz inequality, we obtain
\begin{align*}
| \lan \tilde \phi, \d \Gamma(j , i_0 c_0) \psi\ran | &\le \sum_{l=1}^n \| | i_0 c_0 |_l^{\frac12} \tilde \phi \| \| | i_0 c_0 |_l^{\frac12} \psi \|\le \Big ( \sum_{l=1}^n \| | i_0 c_0 |_l^{\frac12} \tilde \phi \|^2 \Big )^{\frac12} \Big ( \sum_{l=1}^n \| | i_0 c_0 |_l^{\frac12} \psi \|^2 \Big )^{\frac12} \notag \\
 &= \| \d \Gamma ( |i_0 c_0| )^{\frac12} \tilde \phi \| \| \d \Gamma ( |i_0 c_0| )^{\frac12} \psi \|.
\end{align*}
Since $\| \d \Gamma ( |i_0 c_0| )^{\frac12} \psi \|_{ \cF( \mathfrak{h} \oplus \mathfrak{h} ) } = \| \d \Gamma ( |c_0| )^{\frac12} \psi \|_{ \cF( \mathfrak{h} ) }$, we obtain the first term in the r.h.s. of \eqref{checkG-ineq''}. The second one is obtained exactly in the same way.
\eqref{eq:est-dG(j,c1c2)} can be proven in a similar manner.
\end{proof}

In the following lemma, as in the main text, the operator $j_\infty$ on $L^2( \R^3 )$ is of the form $j_\infty = \chi_{b_\e \ge c t^\alpha }$, where, recall, $b_\e =\frac{1}{2}(v_\e(k) \cdot y + {\hbox{ h.c.}})$, where $v_{\e}(k) := \theta_\e \nabla \omega $, $\theta_\e = \frac{\omega}{\omega+\e}$, and $\e = t^{-\kappa}$, $\kappa>0$.

\begin{lemma}\label{lem:Gjinfty-est}
Assume $\alpha + \kappa > 1$. Let $u \in \mathcal{F}$. Then $\big \| ( \Gamma(j_\infty) - \one ) e^{ - i H_f t } u \big \| \to 0,$ as $t \to \infty$.
\end{lemma}
\begin{proof}
Assume that $u \in D( \d \Gamma( \langle y \rangle ) )$. Using unitarity of $e^{ - i H_f t}$ and the fact that $e^{ - i H_f t} = \Gamma( e^{-i \omega t})$, we obtain
\begin{align}
\big \| ( \Gamma(j_\infty) - \one ) e^{ - i H_f t } u \big \| = \big \| ( \Gamma( e^{ i \om t} j_\infty e^{ - i \om t} ) - \one ) u \big \| \le \big \| \d\Gamma( e^{ i \om t} \bar j_\infty e^{ - i \om t} ) u \big \|, \label{eq:c1}
\end{align}
where $\bar j_\infty = \one - j_\infty$.
Using the identity $e^{it\om} b_\e e^{-it\om} = b_\e + \theta_\e t$ and the Helffer-Sj{\"o}strand formula show that
\begin{align*}
e^{i t \om} \chi \Big ( \frac{ b_\e }{ c t^\alpha } \le 1  \Big ) e^{ - i t \om} = \chi \Big ( \frac{ b_\e + \theta_\e t }{ c t^\alpha } \le 1  \Big ).
\end{align*}
Since $\al+\kappa >1$, we have $\chi_{ \frac{b_\e + \theta_\e t }{ c t^\alpha } \le 1}= \chi_{ \frac{b_\e +  t }{ c t^\alpha } \le 1 } + \CO(t^{-(\al+\kappa -1) })$. Due to $\frac{- 2b_\e}{ t } \ge 1$ on $\supp\chi_{ \frac{b_\e +  t }{ c t^\alpha } \le 1 }$ for $t$ sufficiently large, we have
\begin{align*}
\|Ê\chi_{ \frac{b_\e +  t }{ c t^\alpha } \le 1 } \phi \|Ê\le \| \frac{- 2b_\e}{ t } \chi_{ \frac{b_\e +  t }{ c t^\alpha } \le 1 } \phi \| \le \| \frac{ 2\lan y\ran}{ t }  \phi \|,
\end{align*}
and therefore
\begin{align*}
\Big \|Ê\d \Gamma \big ( \chi_{ \frac{ b_\e + \theta_\e t }{ c t^\alpha } \le 1} \big ) u \Big \| \le \frac{2}{ t } \big \| \d \Gamma \big ( \lan y\ran \big ) u \big \|.
\end{align*}
Together with \eqref{eq:c1}, this shows that $\big \| ( \Gamma(j_\infty) - \one ) e^{ - i H_f t } u \big \| \to 0 ,$ for $u \in D\big (\d \Gamma ( \lan y\ran  )\big) $. Since $D\big (\d \Gamma ( \lan y\ran  )\big)$ is dense in $\mathcal{F}$, this concludes the proof.
\end{proof}

\begin{lemma}\label{lm:domainGS}
Assume \eqref{g-est} with $\mu > -1/2$ and \eqref{eta-bnd}. Then $\mathrm{Ran}(P_{\mathrm{gs}}) \subset \D( N^{\frac12} ) \cap \D ( \d \Gamma( b_\e^2 )^{\frac12} )$, in other words, the operators $N^{\frac12} P_{\mathrm{gs}} $ and $ \d \Gamma ( b_\e^2 )^{\frac12} P_{\mathrm{gs}}$  are bounded. Moreover, we have $\| \d \Gamma ( b_\e^2 )^{\frac12} P_{\mathrm{gs}} \| = \CO( t^\kappa )$.
\end{lemma}
\begin{proof}
Let $\Phi_{\mathrm{gs}} \in \mathrm{Ran}(P_{\mathrm{gs}})$. The statement of the lemma is equivalent to the properties that
\begin{align}\label{eq:d1}
k \mapsto \| a(k) \Phi_{\mathrm{gs}} \| , \quad k \mapsto \| b_\e a(k) \Phi_{\mathrm{gs}} \| \in L^2( \mathbb{R}^3 ) ,
\end{align}
and that $\| b_\e a(k) \Phi_{\mathrm{gs}} \|_{L^2( \mathbb{R}^3 )} = \CO( t^\kappa )$. The well-known pull-through formula gives
\begin{align*}
a(k) \Phi_{\mathrm{gs}} = - ( H - E_{\mathrm{gs}} + |k| )^{-1} g(k) \Phi_{\mathrm{gs}}.
\end{align*}
Since $\| ( H - E_{\mathrm{gs}} + |k|)^{-1} \| \le |k|^{-1}$ one easily deduces that $\| a(k) \Phi_{\mathrm{gs}} \| \in L^2( \mathbb{R}^3 )$ for any $\mu>-1/2$.
Likewise, using in addition that $b_\e = \omega_\e^{-1} \frac{i}{2} ( k \cdot \nabla_k + \nabla_k \cdot k ) - i \omega / ( 2 \omega_\e^{2} )$, together with
\begin{align*}
 \| [ ( k \cdot \nabla_k + \nabla_k \cdot k ) , ( H - E_{\mathrm{gs}} + |k| )^{-1} ] \| \lesssim \| |k| ( H - E_{\mathrm{gs}} + |k| )^{-2} \| \lesssim |k|^{-1},
\end{align*}
and \eqref{g-est}--\eqref{eta-bnd}, one easily deduces that $\| b_\e a(k) \Phi_{\mathrm{gs}} \|_{L^2( \mathbb{R}^3 ) } = \CO( t^{\kappa} )$ for any $\mu>-1/2$.
\end{proof}

\bigskip

\section*{Supplement I. The wave operators}\label{sec:WO}

\numberwithin{equation}{foo}

\setcounter{theorem}{0}
\setcounter{equation}{0}

In this supplement we briefly review the definition and properties of  the wave operator $\Omega_+$, and establish its relation with $W_+$ in Theorem I.2 below.  Let $\cH_b\equiv \mathcal{H}_{ \mathrm{pp} } (H) \cap \mathds{1}_{ ( - \infty , \Sigma ) }( H )$ be the space spanned by the eigenfunctions of $H$ with the eigenvalues in the interval $(-\infty, \Sigma)$. Define $\tilde{\mathfrak{h}}_0 := \{ h \in L^2( \R^3 ) , \int |h|^2 ( |k|^{-1}+|k|^2 ) dk < \infty \}$. The wave operator $\Omega_+$ on the space $\cH_b \otimes \cF_{ \mathrm{fin} }( \tilde{\mathfrak{h}}_0 )$, is defined by the formula
\begin{align}\label{Om+-def1}
\Omega_+ := \slim_{t \to \infty} e^{ i t H } I ( e^{ - i t H } \otimes  e^{ - i t H_f }  ) . \tag{I.1}
\end{align}
As in \cite{DerGer2,FrGrSchl1,FrGrSchl2,GrZ}, it is easy to show \\ 

\noindent \textbf{Theorem I.1.} \begin{it}
Assume \eqref{g-est} with $\mu \ge -1/2$ and \eqref{eta-bnd}. The wave operator $\Omega_+$  exists on $\cH_b \otimes \cF_{ \mathrm{fin} }( \tilde{\mathfrak{h}}_0 )$ and extends to an isometric map, $\Omega_+ :  \mathcal{H}_{ \mathrm{as} }  \to \mathcal{H} $, on the space of asymptotic states,  $\mathcal{H}_{ \mathrm{as} } := \cH_b \otimes \mathcal{F}.$
\end{it}
\begin{proof}
 Let $h_t(k) := e^{-i t |k|} h(k)$. For $h \in D( \omega^{-1/2} )$, s. t. $\p^\al h \in D( \omega^{|\al|-1/2} ),\ |\al|\le 2$, we define the asymptotic creation and annihilation operators by (see \cite{DerGer2,FrGrSchl1,FrGrSchl2,Ger,GrZ})
\begin{align*}
a^{\#}_{\pm} ( h ) \Phi := \lim_{t \to \pm \infty} e^{ i t H } a^{\#}( h_t ) e^{ - i t H } \Phi ,
\end{align*}
for any $\Phi \in D( |H|^{1/2} )\bigcap {\rm Ran} E_{(-\infty, \Sigma)} (H)$. Here $a^{\#}$ stands for $a$ or $a^*$. To show that $a^{\#}_{\pm} ( h )$ exist (see \cite{FrGrSchl1, GrZ}), we define $a^{\#}_{t} ( h ):=e^{ i t H } a^{\#}( h_t ) e^{ - i t H }$ and compute $a^{\#}_{t'} ( h )- a^{\#}_{t} ( h )=\int_t^{t'}ds\partial_s a^{\#}_{s} ( h )$ and $\partial_s a^{\#}_{s} ( h ) =i e^{i{H} t} G e^{-iHt},\ \quad \mbox{where} \ \quad G:=[H,  a^{\#}( h_s )] - a^{\#} (\om h_t )=\lan g, h_t\ran_{L^2(dk)}$ for $a^{\#}=a^{*}$ and $-\lan  h_t,g\ran_{L^2(dk)}$ for $a^{\#}=a$. Thus the proof of existence  reduces to showing that one-photon terms of the form $\lan \eta g, h_t\ran$ are integrable in $t$. By \eqref{g-est}, we have $\|\lan \eta g, h_t\ran_{L^2(dk)}\|_{\chp}\lesssim (1+t)^{-1-\ve},$ with $0<\ve<\mu+1$, which is integrable. Moreover, as in  \cite{FrGrSchl1, GrZ} one can show that $a^{\#}_{\pm} ( h )$  satisfy  the  canonical commutation relations and relations $a_{\pm} ( h )\Psi=0$,  and
\begin{align}\label{as-fields-prod}
\lim_{t \to \pm \infty} e^{ i t H } a^{\#}( h_{1,t} ) \cdots a^{\#}( h_{n,t} ) e^{ - i t H } \Phi = a^{\#}_{\pm}( h_{1} ) \cdots a^{\#}_{\pm}( h_{n} ) \Phi , \tag{I.2}
\end{align}
 for any $\Psi\in \cH_b,\ h, h_1, \cdots , h_n \in \tilde{\mathfrak{h}}_0$, and any $\Phi \in \mathds{1}_{(-\infty , \Sigma ) }(H)$. We define the wave operator $\Omega^+$ on $\cH_{\mathrm{fin}} $ by
\begin{align}\label{Om+-def2}
\Omega_+ (\Phi \otimes a^*( h_1 ) \cdots a^*( h_n ) \Omega ):= a^*_+( h_1 ) \cdots a^*_+( h_n ) \Phi . \tag{I.3}
\end{align}
Using the  canonical commutation relations, one sees that $\Omega_+$ extends to an isometric map $\Omega_+ : \mathcal{H}^+_{\mathrm{as}} \to \mathcal{H}$. Using the relation $e^{ i t \hat H } (\Phi\otimes a^{\#}( h_{1} ) \cdots a^{\#}( h_{n} )\Om) = (e^{ i t H }\Phi_{\textrm{gs}})\otimes (a^{\#}( h_{1,t} ) \cdots a^{\#}( h_{n,t} ) \Om) $, the definition of $I$ and \eqref{as-fields-prod}, we identify the definition \eqref{Om+-def2} with \eqref{Om+-def1}.
\end{proof}

Recall that $P_{\mathrm{gs}}$ denotes the orthogonal projection onto the ground state subspace of $H$. Let $\bar P_{ \mathrm{gs} } := \one - P_{ \mathrm{gs} }$ and $\bar P_\Omega := \one - P_\Omega$, where, recall, $P_\Omega$ is the projection onto the vacuum sector in $\cF$. Theorem \ref{thm:MVEraAC} and its proof imply the following result. \\ 

\noindent \textbf{Theorem I.2.} \begin{it} Under the conditions of Theorem \ref{thm:MVEraAC}, we have on $\Ran \chi_\Delta(H)$
\begin{align}\label{OmW}
\Omega_+ ( P_{\mathrm{gs}} \otimes \bar P_\Omega )W_+\bar P_{\mathrm{gs}}  +  P_{\mathrm{gs}} =\one . \tag{I.4}
\end{align}
\end{it}
\begin{proof}
Let $\psi_0 \in \Ran \chi_\Delta(H)$. For every $\e''>0$ there is $\del''=\del(\e'')>0$, s.t.
\begin{equation}\label{psi0-appr}
 \|\psi_0-\psi_{0\e''}-P_{\textrm{gs}}\psi_0\|\le \e'', \tag{I.5}
\end{equation}
where $\psi_{0\e''} = \chi_{\Delta_{\e''}}(H)\psi_{0},\  \mbox{with}\ \Delta_{\e'} =[E_{\mathrm{gs}}+\del, a]$. Proceeding as in the proof of Theorem \ref{thm:MVEraAC} with $\psi_{0\e''}$ instead of $\psi_0$, we arrive at (see \eqref{psite7})
\begin{align} \label{psite7_2}
\psi_{0\e''} & = e^{-iHt} I(e^{ - i E_{\mathrm{gs}} t} P_{\mathrm{gs}} \otimes e^{ - i H_f t}\chi_{ ( 0,a-E_{\mathrm{gs}}]}(H_f)) \phi_{0\e'} + \mathcal{O}( \e' ) + C(\e',m) o_t( 1 ) + C(\e')o_m(1), \tag{I.6}
\end{align}
where we choose $\phi_{0\e'}$ such that $\phi_{0,\e'} \in D( \d\Gamma( \lan y \ran ) ) \otimes \cF_{ \mathrm{fin} }( \tilde{\mathfrak{h}}_0 )$ and $ \| W_+ \psi_{0\e''} - \phi_{0\e'} \| \le \e'$. Now using Theorem I.1, we let $t \to \infty$, next $m \to \infty$ to obtain
\begin{align} \label{psite7_3}
\psi_{0\e''} & = \Omega_+ ( P_{\mathrm{gs}} \otimes \chi_{ ( 0,a-E_{\mathrm{gs}}]}(H_f)) \phi_{0\e'} + \mathcal{O}( \e' ). \tag{I.7}
\end{align}
Since $\Omega_+$ is isometric, hence bounded, we can let $\e' \to 0$, which gives
\begin{align} \label{psite7_4}
\psi_{0\e''} & = \Omega_+ ( P_{\mathrm{gs}} \otimes \chi_{ ( 0,a-E_{\mathrm{gs}}]}(H_f)) W_+ \psi_{0\e''}  = \Omega_+ ( P_{\mathrm{gs}} \otimes \bar P_\Omega ) W_+ \bar P_{ \mathrm{gs} } \psi_{0\e''}. \tag{I.8}
\end{align}
Here we used that $\chi_{ ( 0,a-E_{\mathrm{gs}}]}(H_f) = \bar P_\Omega \chi_{ ( 0,a-E_{\mathrm{gs}}]}(H_f) $, together with $\chi_{ ( 0,a-E_{\mathrm{gs}}]}(H_f) W_+ \psi_{0\e''} = W_+ \psi_{0\e''}$ and $\psi_{0\e''} = \bar P_{ \mathrm{gs}} \psi_{0\e''}$. Introducing \eqref{psite7_4} into \eqref{psi0-appr} and letting $\e'' \to 0$, we obtain
\begin{align*}
\psi_0 = \Omega_+ ( P_{\mathrm{gs}} \otimes \bar P_\Omega ) W_+ \bar P_{ \mathrm{gs} } \psi_0 + P_{\mathrm{gs}} \psi_0,
\end{align*}
which gives \eqref{OmW}.
\end{proof}

\bigskip

\section*{Supplement II. Creation and annihilation operators on Fock spaces}\label{sec:crannihoprs}

With each function $f \in \fh$, one associates \emph{creation}  and \emph{annihilation operators} $a(f)$ and $a^*(f)$ defined, $\mbox{for}\  u\in \otimes_s^n\fh,$ as
\begin{equation*}
a^*(f)  : u\ra \sqrt{n+1}  f\otimes_s u\ \quad   \mbox{and}\  \quad  a(f)  : u\ra \sqrt{n}  \lan f, u\ran_\fh,\
\end{equation*}
with $\lan f, u\ran_\fh:=\int \overline{f(k)} u(k, k_1, \ldots, k_{n-1}) \, \d k$. They are unbounded, densely defined operators of $\G(\fh)$, adjoint of each other (with respect to the natural scalar product in $\cF$) and satisfy the \emph{canonical commutation relations} (CCR):
\begin{equation*}
\big[ a^{\#}(f) , a^{\#}(g) \big] = 0 , \qquad \big[ a(f) , a^*(g) \big] = \lan f, g\ran ,
\end{equation*}
where $a^{\#}= a$ or $a^*$. Since $a(f)$ is anti-linear and $a^*(f)$ is linear in $\varphi$, we write formally
\begin{equation*}
a(f) =  \int \overline{f(k)} a(k) \, \d k , \qquad a^*(f) = \int f(k) a^*(k) \, \d k ,
\end{equation*}
where $a (k)$ and $a^*(k)$  obey (again formally) the canonical commutation relations
\begin{equation*}
\big[ a^{\#}(k) , a^{\#}(k') \big] = 0 , \qquad \big[ a(k) , a^*(k') \big] = \delta (k-k') ,
\end{equation*}
Finally, given an operator $b$ acting on the one-photon space, the operator $\d \Gamma( b )$ defined on the Fock space $\cF$ by \eqref{dG} can be written (formally) as $\d \Gamma( b ) : =  \int a^* ( k ) b a ( k ) \, \d k ,$  where $b$ acts on the variable $k$.

The following bounds on  $a(f)$ and $a^*(f)$  are standard (see e.g. \cite{GuSig}). \\ 

\noindent \textbf{Lemma II.1.} \begin{it}
Recall the notation $\| h \|_\omega := \int dk (1+\omega^{-1}) |h(k)|^2$. Let $f \in \mathfrak{h} = L^2( \R^3 )$. The operators $a(f) ( N +1 )^{-1/2}$ and $a^{*} (f) (N+1 )^{-1/2}$ extend to bounded operators on $\mathcal{H}$ satisfying
\begin{gather*}
\big\Vert a(f) (N+1 )^{- \frac{1}{2} } \big\Vert \leq \Vert f \Vert ,   \qquad \big\Vert a^{*} (f) ( N+1 )^{- \frac{1}{2} } \big\Vert \leq \sqrt{2} \Vert f \Vert .
\end{gather*}
If, in addition, $f $ satisfy $\omega^{-1/2} f  \in L^2( \mathbb{R}^3 )$, then the operators $a(f) ( H_f + 1 )^{-1/2}$ and $a^{*} (f) ( H_f + 1 )^{-1/2}$ extend to bounded operators on $\mathcal{H}$ satisfying
\begin{gather*}
\big\Vert a(f) ( H_f + 1)^{- \frac{1}{2} } \big\Vert \leq \big\Vert \omega^{-\frac{1}{2}} f \big\Vert ,  \qquad \big\Vert a^{*} (f) (H_f + 1 )^{- \frac{1}{2} } \big\Vert \leq \Vert f \Vert_\omega .
\end{gather*}
\end{it}


\bigskip


\begin{thebibliography}{10}

\bibitem{AmBoGe96_01}
W.~Amrein, A.~Boutet~de Monvel, and V.~Georgescu, \emph{{$C_0$}-groups,
  commutator methods and spectral theory of {$N$}-body {H}amiltonians},
  Progress in Mathematics, vol. 135, Birkh\"auser Verlag, 1996.

\bibitem{Ar}
A. Arai,  \emph{A note on scattering theory in nonrelativistic quantum electrodynamics}, J. Phys. A, 16, (1983), 49--69.


\bibitem{Ar2}
A. Arai,  \emph{Long-time behavior of an electron interacting with a quantized radiation field}, J. Math. Phys., 32, (1991), 2224--2242.

\bibitem{bach}
V.~Bach, \emph{Mass renormalization in nonrelativisitic quantum electrodynamics},  \newblock {\em in Quantum Theory from Small to Large Scales}, Lecture Notes of the Les Houches Summer Schools, volume 95. Oxford University Press, 2011.

\bibitem{BFS1}
V.~Bach, J.~Fr\"ohlich, and I.M. Sigal, \emph{Quantum electrodynamics of confined non-relativistic particles}, Adv.~in Math., 137, (1998), 205--298 and 299--395.


\bibitem{BFS2}
V.~Bach, J.~Fr\"ohlich, and I.M. Sigal,
\emph{Spectral analysis for systems of atoms and molecules coupled to the  quantized radiation field},
Commun.~Math.~Phys., 207, (1999), 249--290.


\bibitem{BFSS}
V.~Bach, J.~Fr\"ohlich, I.M. Sigal and A. Soffer,  \textit{Positive commutators and spectrum of {P}auli-{F}ierz {H}amiltonian of   atoms and molecules},
Commun.~Math.~Phys., 207, (1999), 557--587.


\bibitem{BoFa} J.-F. Bony and J. Faupin, \emph{Resolvent smoothness and local decay at low
  energies for the standard model of non-relativistic {QED}},  J. Funct. Anal. 262, (2012), 850--888.


\bibitem{BoFaSig} J.-F. Bony, J. Faupin and I.M. Sigal, \emph{Maximal velocity  of photons in non-relativistic QED}, arXiv (2011).

\bibitem{CFFS}
T. Chen, J. Faupin, J. Fr\"ohlich and I.M. Sigal, \emph{Local decay in non-relativistic QED}, Commun.~Math.~Phys., 309, (2012), 543--583.


\bibitem{DRK}
W. De Roeck and A. Kupiainen, \emph{Approach to ground state and time-independent photon bound for massless spin-boson models}, arXiv:1109.5582.



\bibitem{Der} J. Derezi\'{n}ski, \textit{Asymptotic completeness of long-range $N$-body quantum systems}, Ann. of Math., (1993), 138, 427--476.


\bibitem{DerGer1} J. Derezi\'{n}ski and  C. G\'erard,
\textit{Scattering Theory of Classical and Quantum $N$-Particle Systems},
Texts and Monographs in Physics, Springer (1997).


\bibitem{DerGer2} J. Derezi\'{n}ski and C. G\'erard,  \emph{Asymptotic
completeness in quantum field theory. Massive Pauli-Fierz
Hamiltonians}, Rev. Math. Phys., 11, (1999), 383--450.


\bibitem{DerGer3} J. Derezi\'{n}ski and C. G\'erard, \textit{Spectral and Scattering Theory of Spatially Cut-Off $P(\varphi)_2$ Hamiltonians}, Comm. Math. Phys., 213, (2000), 39--125.


\bibitem{FrGrSchl1}
J.~Fr\"ohlich, M.~Griesemer and B.~Schlein,
\emph{Asymptotic electromagnetic fields in models of quantum-mechanical matter
interacting with the quantized radiation field}, Adv. in Math., 164, (2001), 349--398.


\bibitem{FrGrSchl2}
J.~Fr\"ohlich, M.~Griesemer and B.~Schlein,
\emph{Asymptotic completeness for Rayleigh scattering}, Ann. Henri Poincar\'e, 3, (2002), 107--170.


\bibitem{FrGrSchl3}
J.~Fr\"ohlich, M.~Griesemer and B.~Schlein, \emph{Asymptotic completeness for Compton scattering}, Comm. Math. Phys., 252, (2004), 415--476.


\bibitem{FrGrSchl4}
J.~Fr\"ohlich, M.~Griesemer and B.~Schlein, \emph{Rayleigh scattering at atoms with dynamical nuclei}, Comm. Math. Phys., 271, (2007), 387--430.


\bibitem{FGSig1}
J.~Fr\"ohlich, M.~Griesemer and I.M.~Sigal,
\emph{Spectral theory for the standard model of non-relativisitc QED}, Comm. Math. Phys., 283, (2008), 613--646.



\bibitem{FGSig2}
J.~Fr\"ohlich, M.~Griesemer and I.M.~Sigal, \emph{Spectral renormalization group and limiting absorption principle
for the standard model of non-relativisitc QED}, Rev. Math. Phys., 23, (2011), 179--209.


\bibitem{GGM1}  V. Georgescu,  C. G\'erard and J.S. M{\o}ller, \emph{Commutators, $C_0$-semigroups and resolvent estimates}, J. Funct. Anal., 216, (2004), 303--361.

\bibitem{GGM2}    V. Georgescu,  C. G\'erard and J.S. M{\o}ller,  \emph{Spectral theory of massless Pauli-Fierz models}, Comm. Math. Phys., 249, (2004), 29--78.

\bibitem{Ger} C. G\'erard, \emph{On the scattering theory of massless Nelson models}, Rev. Math. Phys., 14, (2002), 1165--1280.


\bibitem{Graf} G.-M. Graf and D. Schenker, \emph{Classical action and quantum $N$-body
asymptotic completeness. In Multiparticle quantum scattering with applications
to nuclear, atomic and molecular physics} (Minneapolis, MN, 1995),
pages 103--119, Springer, New York, 1997.


\bibitem{Gr}
M. Griesemer, \emph{Exponential decay and ionization thresholds in
  non-relativistic quantum electrodynamics}, J. Funct. Anal., 210,
  (2004), 321--340.


\bibitem{GLL}
M. Griesemer, E.H. Lieb and M. Loss,
\emph{Ground states in non-relativistic quantum electrodynamics}, Invent. Math., 145, (2001), 557--595.

\bibitem{GrZ} M. Griesemer and H. Zenk, \emph{Asymptotic electromagnetic fields in non-relativistic QED: the problem of existence revisited}, J. Math. Anal. Appl., 354, (2009), 239--246.

\bibitem{GuSig}
S. Gustafson and I.M. Sigal, \emph{Mathematical Concepts of Quantum
  Mechanics}, Universitext, Second edition, Springer-Verlag, 2006.

\bibitem{HuSp}  M. H\"ubner and  H. Spohn, \emph{Radiative decay: nonperturbative approaches}, Rev. Math. Phys., 7, (1995), 363--387.



\bibitem{HuSig}
W.~Hunziker and I.M. Sigal, \emph{The quantum {$N$}-body problem}, J. Math.
  Phys., 41, (2000), 3448--3510.


\bibitem{HunSigSof} W. Hunziker, I.M. Sigal and A. Soffer,  \emph{Minimal escape velocities},  Comm. Partial Differential Equations, 24, (1999),  2279--2295.



\bibitem{Sig1} I.M. Sigal, \emph{Ground state and resonances in the standard model of the  non-relativistic {QED}}, J. Stat. Phys., 134, (2009),   899--939.


\bibitem{Sig2} I.M. Sigal, \emph{Renormalization group and problem of radiation}, {\em in Quantum Theory from Small to Large Scales}, Lecture Notes of the Les Houches Summer Schools, volume 95. Oxford University Press, 2011; arXiv.


\bibitem{SigSof1} I.M. Sigal and A. Soffer, \emph{The $N$-particle scattering problem: asymptotic completeness for short-range quantum systems}, Ann. of Math., 125, (1987), 35--108.


\bibitem{SigSof2} I.M. Sigal and A. Soffer,  \emph{Local decay and propagation estimates for time   dependent and time independent hamiltonians}, preprint, Princeton University  (1988).


\bibitem{SigSof3}  I.M. Sigal and A. Soffer,  \emph{A. Long-range many-body scattering. Asymptotic clustering for Coulomb-type potentials}, Invent. Math., 99, (1990), 115--143.


\bibitem{Sk}  E. Skibsted, \emph{Spectral analysis of $N$-body systems coupled to a bosonic field}, Rev. Math. Phys., 10, (1998), 989--1026.


\bibitem{Sp1}  H. Spohn, \emph{Asymptotic completeness for Rayleigh scattering}, J. Math. Phys., 38, (1997),  2281--2288.


\bibitem{Sp2}
H. Spohn, {\em Dynamics of Charged Particles and their
Radiation Field}, Cambridge University Press, Cambridge, 2004.


\bibitem{Yaf} D. Yafaev, \emph{Radiation conditions and scattering theory for $N$-particle Hamiltonians}, Comm. Math. Phys., 154, (1993),  523--554.
\end{thebibliography}
\end{document}